\documentclass[10pt,aps,pra,twocolumn,showpacs,superscriptaddress,nobalancelastpage,notitlepage,longbibliography]{revtex4-2}

\usepackage{graphicx,color,mathbbol,amsthm,amsmath,amsfonts,amssymb,bm,braket,footmisc,multirow,latexsym,times,dsfont}

\usepackage[colorlinks]{hyperref}
\hypersetup{
 colorlinks=true,
 citecolor=blue,
 linkcolor=blue,
 urlcolor=cyan
 }

\newtheorem{proposition}{Proposition}
\newtheorem{theorem}{Theorem}
\newtheorem{lemma}{Lemma}
\newtheorem{cor}{Corollary}
\usepackage[utf8]{inputenc}
\usepackage{bbm}
\usepackage{bm}
\usepackage[caption = false]{subfig}
\color[rgb]{0.,0.,0.} %grey
\usepackage{enumitem}   
\usepackage{tikz}
\usetikzlibrary{mindmap,backgrounds}
\usepackage[normalem]{ulem}
\usepackage{float}
\usepackage{hyperref}
\usepackage{tikz}
\usetikzlibrary{mindmap}
\usetikzlibrary{trees}
\usepackage{bm}
\usepackage{comment}
\usepackage{amsfonts}
\usepackage{amsmath}
\usepackage{tablefootnote}
\usepackage{pifont}
\usepackage{amsthm}
\usepackage{tikz}

\usetikzlibrary{decorations.pathreplacing}

\tikzstyle{transition}=[circle,draw=black!50,fill=black!20,thick,
inner sep=0pt,minimum size=4mm]

\usepackage{braket}
\usepackage{multirow}
\usepackage{cancel}
\usepackage{mathtools}

\usepackage{hyperref}

\renewcommand{\inf}{\infty}

\newcommand{\mb}{\mathbf}

\newcommand{\tr}{\text{tr}}
\newcommand{\norm}[1]{\lVert#1\rVert}

\begin{document}

\title[]{The Interplay of Finite and Infinite Size Stability in Quadratic Bosonic Lindbladians}

\author{Mariam Ughrelidze}
\affiliation{\mbox{Department of Physics and Astronomy, Dartmouth College, 6127 Wilder Laboratory, Hanover, New Hampshire 03755, USA}} 

\author{Vincent P. Flynn}
\affiliation{\mbox{Department of Physics and Astronomy, Dartmouth College, 6127 Wilder Laboratory, Hanover, New Hampshire 03755, USA}} 

\author{Emilio Cobanera}
\affiliation{\mbox{Department of Mathematics and Physics, SUNY Polytechnic Institute, 100 Seymour Avenue, Utica, NY 13502, USA}}
\affiliation{\mbox{Department of Physics and Astronomy, Dartmouth College, 6127 Wilder Laboratory, Hanover, New Hampshire 03755, USA}} 

\author{Lorenza Viola}
\affiliation{\mbox{Department of Physics and Astronomy, Dartmouth College, 6127 Wilder Laboratory, Hanover, New Hampshire 03755, USA}} 

\begin{abstract}
We provide a framework for understanding {\em dynamical metastability} in open many-body systems of free bosons, whereby the dynamical stability properties of the system in the infinite-size (thermodynamic) limit may sharply differ from those of any finite-size truncation, and anomalous transient dynamics may arise. By leveraging pseudospectral techniques, we trace the discrepancy between asymptotic and transient dynamics to the {\em non-normality} of the underlying quadratic bosonic Lindbladian (QBL) generator, and show that two distinct flavors of dynamical metastability can arise. QBLs exhibiting \textit{type I dynamical metastability}, previously discussed in the context of anomalous transient amplification [Phys. Rev. Lett. \textbf{127}, 245701 (2021)], are dynamically unstable in the infinite-size limit, yet stable once open boundaries are imposed. \textit{Type II-dynamically metastable} QBLs, which we uncover in this work, are dynamically stable for infinite size, but become unstable under open boundary conditions for arbitrary finite system size. We exhibit representative models for both types of metastability in the dissipative, as well as the limiting closed-system (Hamiltonian) settings, and analyze distinctive physical behavior they can engender. We show that dynamical metastability manifests itself in the generation of \textit{entanglement entropy}, by way of a transient which reflects the stability phase of the infinite (rather than the actual finite) system and, as a result, is directly tied to the emergence of {\em super-volume scaling} in type I systems. Finally, we demonstrate how, even in Hermitian, and especially in highly non-normal regimes, the spectral properties of an infinite-size QBL are reflected in the linear response functions of the corresponding finite QBLs, by way of {\em resonant pseudospectral modes}.
\end{abstract}

\date{\today}

\maketitle 

%\tableofcontents

\section{Introduction}

\subsection{Context and main results}

While realistic many-body systems always comprise a finite number of constituents and have a finite spatial extent, taking the infinite-size (thermodynamic) limit plays a foundational role and is standard practice across statistical mechanics. In particular, only in the infinite-size limit is it possible to use statistical ensembles for recovering well-defined macroscopic observables or describe phase transitions \cite{Kardar}. In the infinite-size limit, the boundary conditions (BCs) that are imposed on the system are taken to have negligible influence on its thermodynamic properties; in practice, finite-size scaling methods are used to infer quantities that are defined in the thermodynamic limit (for example, critical exponents) from experimental or numerical data on finite systems, under the assumption that the system size may be treated as a scaling parameter \cite{FiniteSizeScaling}. Although the existence of the thermodynamic limit is not guaranteed and its justification is not devoid of subtleties \cite{Kuzemski,Palacios}, it is fair to say that our intuition for how properties of many-body quantum systems change with {system size} and/or BCs is derived from our experience with Hamiltonian -- hence, {Hermitian} (and normal) dynamical systems. To what extent does this intuition carry over to more general -- {\em non-Hermitian} (NH) and possibly {\em non-normal} -- dynamical systems, and why would this be relevant?

Scenarios where the dynamics are governed by a NH generator arise in many quantum settings of interest. On the one hand, non-Hermiticity appears naturally in open quantum systems \cite{Breuer}, a prominent class being Markovian dissipative systems, that evolve under a {\em Lindblad master equation}   \cite{LindbladGenerators}. In either a semiclassical or a measurement-post-selected regime where ``quantum jumps'' can be neglected, the latter simplifies to a probability-non-conserving evolution described by an {\em NH effective Hamiltonian}. Likewise, explicitly NH descriptions have long been used to model open-system behavior phenomenologically -- from the decay of unstable states to anomalous wave propagation and localization, with implications ranging from photonic, electrical, and mechanical systems all the way to quantum materials \cite{UedaReview}. On the other hand, even for a closed system, unitary dynamics can still be retained for a class of NH, but parity-time ``${\cal PT}$-symmetric'' \cite{BenderPTRep}, or ``pseudo-Hermitian'' \cite{Ali} Hamiltonians, if one allows for a modified inner product. Most remarkably, as a {\em sole} consequence of quantum statistics, {\em effectively NH} dynamics arises for systems of non-interacting (free) bosons under ``pairing interactions'', despite their physical Hamiltonian remaining Hermitian at the many-body level \cite{bla86}.
What, then, is the correct picture of the thermodynamic limit and its interplay with finite size in such more general settings? 

One of the counter-intuitive features of 
NH systems is the eponymous \textit{NH skin effect} (NHSE) \cite{NHSE0,YaoSkin,SatoTopoSkin}, whereby the spectrum of a spatially extended, ``bulk'' system, that one may think of as being in the thermodynamic limit, changes drastically upon imposing system terminations via open BCs (OBCs), and a macroscopic number of eigenstates localize at the boundary.
An illustrative example is the paradigmatic Hatano-Nelson (HN) chain \cite{HNChain}. The eigenvalue spectrum under periodic or bi-infinite BCs (PBCs/BIBCs) forms an ellipse in the complex plane; in dramatic opposition, under OBCs, the spectrum of the HN chain is purely real. Stranger still, the spectrum under semi-infinite BCs (SIBCs) contains the ellipse of the BIBC spectrum, in addition to {\em all} the complex points in its interior \cite{UedaTopo,SatoTopoSkin}. 

Dramatic as it is, the NHSE is in fact a manifestation of a more general phenomenon: The spectrum of a non-normal dynamical generator lacks robustness against perturbations \cite{TrefethenPS}. Rather, under changes in, say, BCs, system size, or other system-parameters, it can be strongly deformed in seemingly unpredictable ways. Therefore, the complete characterization of such systems requires going beyond the spectrum and adopting mathematical tools specifically designed for understanding highly non-normal dynamical systems. Such tools are predicated on the fundamental notion of the {\em pseudospectrum} \cite{TrefethenPS,BottcherToe}. In essence, the pseudospectrum is a norm-dependent generalization of the usual spectrum of a linear transformation. Among its many useful features, it can provide an understanding of how the spectrum of a dynamical system will change under a perturbation of a given size, with the extent of the resulting spectral deformation depending on the degree of non-normality. 

In this work, we will make use of pseudospectral techniques to tackle the above questions for Markovian systems of free bosons, described by \textit{Quadratic bosonic Lindbladians} (QBLs), that is, Lindblad generators that are quadratic in some set of canonical bosonic creation and annihilation operators. The effects of strong non-normality are particularly rich in bosonic systems, as such systems have the potential to be dynamically unstable \cite{Prosen3QBoson,BarthelQuadLindblad}. Consequently, in addition to purely unstable ones, there exist bosonic systems which have stable asymptotic dynamics, yet display unstable, amplifying dynamics for transient regimes whose duration scales with system size. Such systems are of interest from the point of view of nonreciprocal transport and topological amplification \cite{ClerkBKC,PorrasTopoAmp,NunnenkampTopoAmp}, 
along with, as of recently, the search of topological bosonic zero modes \cite{Bosoranas,PostBosoranas}. In the latter context, in particular, the notion of {\em dynamical metastability} was put forward to describe situations where the infinite-size dynamics controls the transient of finite-size dynamics -- in the sense that the dynamical stability properties of a QBL ``in the thermodynamic limit" can drastically differ from those of {\em any} finite portion of the same system; yet,  they nonetheless ``imprint'' themselves in and are observable throughout the transient window.

More specifically, throughout this work we shall focus on one-dimensional (1D) ``bulk-translationally invariant'' QBLs, for which (discrete) translation symmetry is possibly broken {\em only} by BCs \cite{PostBosoranas}, and mathematical results from Toeplitz operator theory can be brought to bear on the problem of characterizing the relevant spectral and pseudospectral properties \cite{TrefethenPS,BottcherToe}. Our key findings can be summarized as follows. 

\textbf{(1)} By casting the discrepancy between asymptotic and transient regimes of a system as an interplay between finite-size effects and a non-trivial imprint of the infinite-size limit, we establish two distinct ways in which the dynamics of a half-infinite system can differ from that of its finite-size, open-boundary incarnations. Proper retention of boundary information requires making the notion of ``semi-infinite BCs'' mathematically precise; we do so by considering a semi-infinite system to be the {\em union} of two 
systems, each with one (left or right) boundary (see also Fig.\,\ref{fig:BCs}). 
In this way, one mechanism for disagreement is provided by a {\em spectral discontinuity}, whereby the limit of the finite-size spectra for increasing system size differs from the spectrum of the corresponding infinite system. Such a spectral discontinuity does not necessarily translate into a disagreement between dynamical stability phases, however. Hence, an independent notion of disagreement arises with regards to the dynamical stability in the finite-size, as compared to the infinite-size limit, {\em regardless} of spectral convergence. We dub systems exhibiting this kind of dynamical-stability disagreement as {\em dynamically metastable} (DM), and concentrate on them in the remainder of our analysis. 

\smallskip

\textbf{(2)} We identify {\em two} different types of dynamical metastability, which we refer to as {\em Type I} and {\em Type II} DM, respectively. Type I DM systems are characterized by stable asymptotic dynamics for any finite system size, but a dynamically unstable semi-infinite limit. In contrast, type II DM systems are dynamically stable in the semi-infinite limit, but display asymptotically unstable finite-size dynamics. Both type I and type II systems host {\em anomalous transient} dynamics. Type I systems have amplifying transient dynamics, whereby, despite asymptotic stability, they appear dynamically unstable. In this regime, observable expectation values can grow exponentially for a period of time that diverges with system size. Likewise, type II systems appear stable, for increasingly long times with growing system-size, until exponential instabilities eventually set in. In both scenarios, the transient dynamics of a finite-size truncation reflects the stability phase of the corresponding (semi-)infinite system, more and more reliably so as the system size grows. From a physical standpoint, one may think of the the above as reflecting the fact that, the larger the system, the longer its dynamics takes to become aware of its boundaries; therefore, the finite system temporarily behaves akin to its infinite-size counterpart. Mathematically, the transient behavior can be gleaned from the pseudospectrum of the system, which acts like the spectrum for a transient period of time. We introduce several model QBLs to concretely illustrate these ideas, both in the open Markovian setting and in the closed-system limit, whereby the dynamics is generated by a quadratic bosonic Hamiltonian (QBH).  

\smallskip

\textbf{(3)} We conjecture that the discrepancy between the transient and asymptotic dynamics in a DM system of both type I and type II can be detected by any physical quantity whose scaling with system size changes depending on the dynamical stability phase of the underlying QBL (or QBH). We test this conjecture by examining the rate of generation of bipartite {\em entanglement entropy} (EE) over time, for different system sizes. For QBHs initialized in a pure Gaussian state, the scaling of asymptotic EE has been shown to be directly tied to the dynamical stability properties of the Hamiltonian \cite{Rigol2018}. Consequently, it stands to reason that in the transient period, during which a DM system appears to be
in a stability phase different from the true one, the EE will also scale in accordance with the stability phase of the semi-infinite system. Indeed, for our model QBHs, we find that type I DM systems experience size-dependent transient EE growth with the rate characteristic of their unstable infinite-size limits, whereas type II ones exhibit transiently stable EE behavior, before an unstable EE growth regime sets in. Notably, these results provide an explanation for the recently reported ``super-volume scaling'' law of the asymptotic EE in a family of QBH models \cite{ClerkInterpolation}, which we recognize as belonging to our type I DM class.

\smallskip
\textbf{(4)} 
We show that the nature of the semi-infinite limit has direct implications for the {\em linear response} of the finite system to a weak external drive, computed under OBCs. Surprisingly, we find that, even in the Hermitian (normal) limit, the susceptibility matrix can display resonant-like behavior for a range of drive frequencies that are not in the exact spectrum of the physical system. Instead, despite the finite system size, such ``pseudoresonances" arise for drive frequencies that are in the infinite-size spectrum, hence the pseudospectrum of the physical system. To drive this idea home, we verify that the response of the system can change drastically as the pseudospectrum is varied, but the spectrum is kept unchanged. Furthermore, we show that the the response function exactly reflects the spatial structure of the pseudospectral modes at the drive frequency, away from the spectrum of the system.
Our results strongly support the idea that the pseudospectrum, rather than the spectrum, provides the appropriate tool for predicting and interpreting the linear response of non-normal dynamical systems in general. 

In more detail, the content of the paper is organized as follows. In Sec.\,\ref{sec:background}, we describe 
the class of QBLs of interest, recall relevant notions of stability, and introduce the different kinds of BCs we will use; a self-contained introduction to the pseudospectrum and related mathematical tools is also provided. The exposition of original work starts in Sec.\,\ref{sec:dynmeta}, with an in-depth exploration of the notion of dynamical metastability and its relationship with the finite/infinite system dichotomy. The two main types of dynamical metastability are introduced and discussed in this section, as well as the relationship to anomalous relaxation. In Sec.\,\ref{illmod}, we show how our general framework is actionable by turning some of its tenets into design principles for constructing concrete models of interest, which serve to illustrate various facets of dynamical metastability. Finally, in Sec.\,\ref{sec:further}, we explore the implications of dynamical metastability from the point of view of EE generation and linear response behavior. We conclude in Sec.\,\ref{sec:conclusions} with a summary and an outlook to future research. 

The Appendices collect important supporting calculations, along with mathematical results which may be of independent interest. In particular, in Appendix\,\ref{app:WH}, we adapt the Wiener-Hopf matrix factorization method, introduced in \cite{WienerHopf} in the context of free-fermion physics, for determining partial indices of a class of Toeplitz operators. In Appendix\,\ref{BoundProof}, we rigorously justify an upper-bound relating the dynamical stability properties of finite- vs. infinite-size block-Toeplitz systems, while in Appendix\,\ref{app:LR}, the general formalism for linear response of Markovian open quantum systems \cite{ZanardiResponse} is specialized to QBLs.

\subsection{Relation to existing work}

Pseudospectra have been extensively used in applied mathematics \cite{TrefethenPS} and various areas of physics and engineering \cite{Reddy,Schmid,Baggett,ViolaJ,Baggio}, as a method for studying the sensitivity of a non-normal matrix or an operator to perturbations, in addition to describing the transient evolution of a dynamical system. The pseudospectrum of a non-normal dynamical generator, which can be very different from its spectrum, is instrumental for making sense of the intuition-defying consequences of non-normality that have been observed in different contexts. Applications of pseudospectra to non-normal dynamics of {\em many-body} quantum systems have been considered only recently. First employed in \cite{SatoOkumaPseudo} to elucidate the robustness of emerging topological zero modes, pseudospectra have then been further used to explore the interplay between robustness and sensitivity \cite{Makris} in lattice systems described by explicitly NH Hamiltonians. In a dynamical context, a non-trivial pseudospectrum was conjectured in \cite{SatoOkumaNHSE} to be responsible for the anomalous relaxation behavior reported in \cite{UedaSkin} for dissipative systems exhibiting a ``Liouvillian NHSE''. For discrete-time dynamics, similar ``two-step'' relaxation behavior has been reported for random quantum circuits and eventually interpreted, again, on the basis of the pseudospectrum \cite{ZnidaricPseudo0, ZnidaricPseudo}. Our use of pseudospectral techniques in this work builds and expands upon our previous contributions \cite{Bosoranas, PostBosoranas}, by maintaining a specific emphasis on non-interacting \emph{many-body bosonic systems}, evolving under a QBL or QBH. 

We stress that other notions of metastability have also been explored, in particular, in the context of generalizing notions from classical stochastic dynamics to Markovian open quantum systems \cite{GarrahanMeta,GarrahanClassMetast}. The key observation therein is that, in certain parameter regimes of some {\em fixed} (and finite-dimensional) Lindblad generators, a separation in dynamical time scales can be engendered by a separation of the real parts of a number of eigenvalues from those of the remaining ones. 
As a consequence, a manifold of long-lived, ``metastable states'' emerges, which appear stationary for a long time before the system relaxes to the true steady state at a much larger time scale, set by an eigenvalue with the smallest (in magnitude) real part. While in our case a delayed relaxation to the steady state manifold can also occur, the notion of dynamical metastability we consider is fundamentally different, in that it is associated to the ``spontaneous'' loss or restoration of dynamical stability in the infinite
system-size limit. Further to that, finite-dimensional systems as considered in the above works are guaranteed to be dynamically stable, since their spectra are bounded in the left-half complex plane.  Thus, the behavior we characterize here as dynamical metastability is distinctively afforded to us by the {\em infinite-dimensional} nature of the bosonic Fock space.

\section{Background}
\label{sec:background}

\subsection{Quadratic bosonic Lindbladians}

We focus on a class of multi-mode, non-interacting bosonic systems linearly coupled to a Markovian reservoir. The dynamics of any Markovian open quantum system is governed by the Lindblad master equation, $\dot{\rho}(t)=\mathcal{L}(\rho(t))$, with the density operator $\rho(t)$ describing the state of the system at time $t\geq 0$ and the Lindbladian generator having, in units $\hbar=1$, the canonical (diagonal) form
\begin{align}
\mathcal{L}(\rho)\nonumber&=-i[H,\rho]+\sum_{\mu=1}^d\Big(L_\mu\rho L_\mu^{\dag}-\frac{1}{2}\{L_\mu^{\dag}L_\mu,\rho\}\Big)\\
&\equiv -i[H,\rho]+\sum_{\mu=1}^d\mathcal{D}[L_\mu](\rho).
\end{align}
Here, the commutator term, with $H=H^{\dag}$, accounts for the unitary contribution to the dynamics, whereas the $d$ dissipators $\mathcal{D}[L_\mu]$ encode the action of $d$ different dissipative channels, characterized by Lindblad (or jump) operators that we take to be traceless. We further assume that none of the operators $H, L_\mu$ are explicitly time-dependent, making ${\cal L}$ a linear, time-invariant superoperator. Oftentimes it is more convenient to express the dissipative part in terms of a set of physically relevant system operators, say, $\{A_j\}_{j=1}^n$, such that $L_{\mu}\equiv \sum_{j=1}^n l_{j\mu}A_j,$ $l_{j\mu}\in\mathbb{C}$. In this way, we arrive at the Gorini-Kossakowski-Lindblad-Sudarshan (GKLS, non-diagonal) representation of the Lindbladian:
\begin{align*}
\mathcal{L}(\rho)&=-i[H,\rho]+\sum_{j,k=1}^n\mb M_{jk} \Big(A_k\rho A_j^{\dag}-\frac{1}{2}\{A_j^{\dag}A_k,\rho\}\Big), 
\end{align*}
where the $n\times n$ GKLS matrix $\mb M_{jk}\equiv \sum_{\mu=1}^dl^{*}_{j\mu}l_{k\mu}$ is positive-semidefinite. Equivalently, we may describe the system dynamics in the Heisenberg picture, where states are stationary and instead, operators representing observables, $B=B^\dag$, evolve in time according to $\dot{B}(t)=\mathcal{L}^{\star}(B(t))$, with
\begin{align*}
\mathcal{L}^{\star}(B)= i[H,B]+\sum_{jk}\mb M_{jk} \Big( A_j^{\dag}B A_k -\frac{1}{2}\{A_j^{\dag}A_k,B\}\Big).
\end{align*}
The dual generator $\mathcal{L}^{\star}$ is the Hilbert-Schmidt adjoint of $\mathcal{L}$, ensuring that the relationship $\tr[B\rho(t)]=\tr[B(t)\rho(0)]$ is satisfied, for arbitrary observables $B$ and initial states $\rho(0)$.

Specializing to bosonic lattice systems, we introduce $N$ bosonic modes with creation and annihilation operators $a_j^\dag$ and $a_j$, $j=1,\ldots,N$, satisfying the canonical commutation relations (CCRs) $[a_i,a_j^\dag]=\delta_{ij} \mathds{1}_F$, with $\mathds{1}_F$ denoting the identity on Fock space. While terms that are linear in the creation and annihilation operators may be considered in the broader class of ``quasi-free'' \cite{BarthelQuadLindblad} (or Gaussian \cite{GenoniGaussian}) Markovian dynamics, in what follows we will be interested in \textit{purely quadratic} bosonic Lindbladians. In this case, the Hamiltonian $H$ is taken to be a \textit{quadratic bosonic Hamiltonian} (QBH), namely, a Hamiltonian quadratic in the bosonic creation and annihilation operators, while the operators $L_\mu$ are assumed to be linear in the creation and annihilation operators. Physically, QBLs describe a set of non-interacting (or mean-field interacting) $N$ bosonic modes, which in general are both coupled to one another coherently, via the QBH $H$, and dissipatively to a Markovian quantum bath, via the set $\{L_\mu\}_{\mu=1}^d$. 
 
Mathematically, QBHs can be expressed in the form
\begin{align*}
H&=\frac{1}{2}\sum_{i,j}^N\left(\mb K_{ij}a_i^{\dag}a_j+\mb \Delta_{ij}  a_i^{\dag}a_{j}^{\dag}+\text{H.c}\right), 
\end{align*}
with $\mb K=\mb K^{\dag}, \mb \Delta=\mb \Delta^T$ in order to preserve both Hermiticity of $H$ and bosonic CCRs.  In analogy with fermionic systems, we call $\mb K,\mb \Delta$ the \textit{hopping} and \textit{pairing} matrices, respectively. Physically, $\mathbf{K}_{ij}$ encodes passive hopping between modes $i$ and $j$ when $i\neq j$ and an onsite energy term when $i=j$, while $\bm{\Delta}_{ij}$ represents a (coherent) ``bosonic pairing'' coupling. For instance, in photonic implementations, it describes a pairwise non-degenerate parametric amplification (or two-photon) process when $i\neq j$ and a degenerate one when $i=j$. The latter processes typically arise in physical systems as mean-field (or ``linearized") incarnations of three- or four-wave mixing with a set of auxiliary modes \cite{Boyd}.

A QBH can be more compactly defined in terms of the \textit{bosonic Nambu array}, $\Phi\equiv [a_1,a_1^{\dag},\ldots, a_N,a_N^{\dag}]^T$, namely, 
 \begin{align*}
 H=\frac{1}{2}\Phi^{\dag}\mb H\Phi + \frac{1}{2}\tr\,\mathbf{K},
 \end{align*}
where $\mb H$ is a $2N\times 2N$ block-Hermitian matrix with the $(ij)$-th block given by 
$$ [\mb H]_{ij}= \begin{pmatrix}
    \mb K_{ij}, \mb \Delta_{ij}\\
    \mb \Delta_{ij}^{*},\mb K_{ij}^{*}
\end{pmatrix}.$$ 
Additionally, $\mathbf{H} = \bm{\tau}_1 \mathbf{H}^T \bm{\tau}_1$ in terms of the Nambu-space Pauli matrix $\bm{\tau}_1 \equiv \mathds{1}_N\otimes \bm{\sigma}_1$, with $\mathds{1}_N$ the $N\times N$ identity matrix and $\bm{\sigma}_1$ the usual Pauli matrix. By defining $\bm{\tau}_2$ and $\bm{\tau}_3$ in an analogous way, we have 
$$\Phi^\dag = (\bm{\tau}_1\Phi)^T,\; [\Phi_i,\Phi_j] = (i\bm{\tau}_2)_{ij}1_F, \; [\Phi_i,\Phi_j^\dag] = (\bm{\tau}_3)_{ij} \mathds{1}_F.$$
In this Nambu formalism, the Lindblad operators that enter the dissipative part of the evolution may be expressed as $L_\alpha=\sum_{j=1}^d l_{j\alpha}\Phi_j$. 
It follows that
\begin{align*}
\mathcal{D}(\rho) = \sum_{jk}\mb M_{jk} \Big( \Phi_k \rho \Phi_j^\dag -\frac{1}{2}\{\Phi_j^{\dag}\Phi_k,\rho\}\Big),
\end{align*}
with the $2N\times 2N$ GKLS matrix $\mathbf{M}$ defined as before.

As it turns out, the Heisenberg picture proves most convenient for our purposes. The associated equations of motion (EOMs) for the creation and annihilation operators are
\begin{equation}
\label{linear eqn of motion}
\dot{\Phi}(t)=\mathcal{L}^{\star}(\Phi(t))=-i\,\mb G \Phi(t),
\end{equation}
where  $\mb G$ is the \textit{dynamical matrix} given by 
\begin{equation}
\label{dynamical matrix}
\mb G =\boldsymbol{\tau}_3\mb H-\frac{i}{2}\boldsymbol{\tau}_3\left(\mb M-\boldsymbol{\tau}_1\mb M^{T}\boldsymbol{\tau}_1\right).
\end{equation}
Generically, $\mathbf{G}$ is NH and, in fact, \textit{non-normal}. That is, $\mb G\mb G^{\dag}\neq \mb G^{\dag}\mb G$. As such, it need not be diagonalizable and has generally complex eigenvalues. However, it satisfies the ``charge conjugation'' symmetry property $\mb G=-\boldsymbol{\tau}_1\mb G^*\boldsymbol{\tau}_1$, and hence its spectrum, denoted by $\sigma(\mathbf{G})$, obeys $\sigma(\mb G)=-\sigma(\mb G)^*$. For unitary evolution ($\mathcal{D}= \mathbf{M}=0$), the dynamical matrix for a QBH is given by $\mathbf{G} = \bm{\tau}_3 \mathbf{H}$. However, the same simplification can occur also for a non-vanishing dissipator, provided that the GKLS matrix satisfies the condition $\mathbf{M}=\bm{\tau}_1\mathbf{M}^T\bm{\tau}_1$. As one may verify, this condition implies that the semi-group generated by $\mathcal{L}$ is \textit{unital} (or {\em bistochastic}), that is, we have $\mathcal{L}(1_F)=0$ (note that $\mathcal{L}^\star (1_F)=0$ always holds, given the trace-preserving property of the generator) \footnote{Note that this condition is sufficient, but not necessary. One may verify that a necessary and sufficient condition for ${\cal L}$ to be unital is that $\text{Tr}(\boldsymbol{\tau}_3 (\mb M +\mb M^{T})) =0$.}.
In this case, in addition to the symmetry property described above, $\mathbf{G}$ obeys a \textit{pseudo-Hermitian} constraint, $\mathbf{G}=\bm{\tau}_3 \mathbf{G}^\dag \bm{\tau}_3$. Altogether, $\sigma(\mathbf{G})$ then enjoys a fourfold symmetry about both the real and imaginary axes. 

It is tempting to conclude that, armed with the dynamics of $\Phi(t)$, we can construct the dynamics of arbitrary observables, algebraically built from creation and annihilation operators. Indeed, without dissipation ($\mathcal{D}=0$), this is the case. However, in the presence of dissipation, the Heisenberg dynamics of a product $(B_1B_2)(t)$ need not be equivalent to the product of the individual Heisenberg dynamics $B_1(t)B_2(t)$. In particular, for QBLs and by letting $Q\equiv \Phi\Phi^{\dag}$, the EOM for quadratic forms is given by
\begin{equation}
\label{quadratic eqn of motion}
\dot{Q}(t)=\mathcal{L}^{\star}(Q(t))=
-i \big(\mb G Q(t)-Q(t)\mb G^{\dag}\big)+\boldsymbol{\tau}_3\mb M \boldsymbol{\tau}_3 1_F.
 \end{equation}
While complete information about the dynamics resulting from arbitrary initial conditions requires knowledge of the evolution of arbitrary high-degree operators, solving Eqs.\,\eqref{linear eqn of motion} and \eqref{quadratic eqn of motion} suffices to completely determine the dynamics generated from {\em Gaussian initial states}, since their Gaussian character is preserved \cite{GenoniGaussian}. Furthermore, a significant amount of information about the stability of a QBL can be inferred from its dynamical matrix, as we discuss next.

\subsection{Notions of stability for QBLs}
\subsubsection{Dynamical stability}

Since, for bosonic systems, the underlying Fock space is infinite-dimensional, the expectation values of observables of interest can, in principle, be unbounded. Hence, the potential exists for dynamical instabilities to arise. We say that a QBL (or QBH, in the unitary case) is \textit{dynamically stable} whenever it generates bounded evolution of all observable expectation values for arbitrary (normalizable) states, and \textit{dynamically unstable} otherwise. While a finite-dimensional Lindbladian always possesses at least one steady state (SS), and the structure of its fixed points is well-characterized (see e.g., \cite{PeterFixedPoints}), the SS manifold can be empty in infinite dimension. For a QBL, it is known that it is dynamically stable if and only if it admits a SS, that is, a state $\rho_\text{ss}$ satisfying $\mathcal{L}(\rho_\text{ss})=0$. Further, the existence and number of SSs are nearly completely determined by the dynamical matrix through the {\em rapidity spectrum}, $\sigma (-i \mb G)$ \cite{Prosen3QBoson}. Let the \textit{stability gap} of the QBL be defined as \cite{PostBosoranas}:
\begin{align}
\Delta_S\equiv \text{max}\,\text{Re}[\sigma (-i \mb G)].
\label{stabgap}
\end{align}
It follows that for $\Delta_S<0$, a unique SS exists and the QBL is dynamically stable. In this case, asymptotic relaxation to the SS is characterized in terms of the \textit{Lindblad gap} (also known as the spectral or dissipative gap): 
$$\Delta_{\mathcal{L}} \equiv  |\sup\text{Re}[\sigma(\mathcal{L})\setminus\{0\}]|. $$
Quantitatively, the worst-case distance from the SS is bounded exponentially in time according to:
\begin{align}
\label{expbound}
\sup_{\rho(0)} \,\norm{\rho(t) - \rho_\text{ss}}_\text{tr}\leq K e^{-\Delta_{\mathcal{L}}t},\quad K>0,
\end{align}
where the trace norm is $\norm{A}_\text{tr} \equiv 
\tr[\sqrt{A^\dag A}]$.
As it turns out, $\Delta_{\mathcal{L}} = |\Delta_S| = -\Delta_S$ when the SS is unique \cite{Prosen3QBoson}. On the contrary, if $\Delta_S>0$, the system possesses no SS and is thus dynamically unstable. Finally, the marginal case $\Delta_S=0$ can feature either infinitely many or zero SSs, with the former (latter) case being dynamically stable (unstable). These results are summarized in Table\,\ref{t: stability and steady state}. 

While the above considerations apply to arbitrary QBLs, additional conclusions can be made for the special case where the dynamical matrix in Eq.\,\eqref{dynamical matrix} attains the simple form $\mathbf{G} = \bm{\tau}_3 \mathbf{H}$. As we discussed, the latter can equivalently describe purely Hamiltonian, unitary dynamics or (a class of) non-unitary but unital dynamics. Either way, the effective decoupling of arbitrary linear operators from the dissipation engenders a constraint on the stability gap. As remarked, $\mathbf{G}$ is pseudo-Hermitian in this case, implying that its rapidity spectrum enjoys the same fourfold symmetry the spectrum does. It follows that $\Delta_S\geq 0$ for this subclass of QBLs. This is not only consistent with the fact that any unitary or unital dynamics is generically rich in SSs, but it grants us further tools for assessing dynamical stability in the marginal case $\Delta_S=0$. As shown in \cite{Decon}, a system described by such a pseudo-Hermitian dynamical matrix $\mathbf{G}$ is dynamically stable if and only if $\Delta_S=0$ {\em and} $\mathbf{G}$ is diagonalizable. If there exists a point in parameter space where $\mathbf{G}$ loses diagonalizability -- a so-called \textit{exceptional point} \textit{(EP)} -- then the normal mode corresponding to the associated generalized eigenvector will diverge polynomially in time. Notably, the existence of an EP implies that $\mathbf{G}$ is non-normal, as otherwise it could be unitarily diagonalized. The extreme nature of this non-normality may be appreciated by noting that the emergence of an EP requires that (at least) two eigenvectors coalesce, implying maximal overlap while, in the normal case, distinct eigenvectors always have vanishing overlap. 

\begin{table}[t]
\begin{center}
\begin{tabular}{ ||c| c| c ||}
\hline
Stability gap & Dynamical stability & Number of SSs \\ 
\hline\hline
 $\Delta_S<0$ & Stable & One \\ 
 \hline
$ \Delta_S=0$& Stable/Unstable & Zero or Infinite \\  
\hline
 $\Delta_S>0$& Unstable &   Zero\\
 \hline
\end{tabular}
\caption{The relationships between the stability gap $\Delta_S$ [Eq.\,\eqref{stabgap}], the dynamical stability, and the number of SSs of a QBL.}
\label{t: stability and steady state}
\end{center}
\end{table}

Even when the system described by $\mathbf{G}$ is dynamically stable, it may be susceptible to instabilities resulting from \textit{arbitrarily weak} perturbations. Such a scenario corresponds to the occurrence of a so-called \textit{Krein collision} (KC) in the dynamical matrix spectrum \cite{Decon}. Recall that a KC is said to occur at an eigenvalue $\omega\in\sigma(\mathbf{G})$ if there exist corresponding eigenvectors $\vec{\psi}_\pm$ with opposite Krein signature, that is, $\,\vec{\psi}_{\pm}^\dag \bm{\tau}_3\vec{\psi}_\pm = \pm 1$. For a QBH, one may check that this ensures the existence of a pair of equal and opposite excitation energies.
Altogether, instabilities of a QBH or a unital QBL are signaled by either non-real eigenvalues, in which case $\Delta_S>0$, or by spectral degeneracies of either a EP or KC type, when $\Delta_S=0$.

\subsubsection{Thermodynamic stability}

A distinct notion of stability exists in the unitary context of QBHs, namely, \textit{thermodynamic stability}. A QBH $H$ is thermodynamically stable if the energy expectation value $\braket{H}$ is bounded from either below or above (in the case where $\braket{H}$ is bounded from above, one may simply work with the Hamiltonian $H'\equiv -H$, which is bounded from below). This requires the existence of a many-body ground state and quasi-particle excitation energies that are either all nonnegative or nonpositive, respectively. Mathematically, one may assess thermodynamic stability based on whether or not $\mb H$ is positive (negative)-semidefinite \cite{Decon}.

A thermodynamical instability (sometime also referred to as a ``Landau instability'' \cite{UedaBEC}) may arise due to a dynamical instability -- a notable example being a squeezing QBH $H\propto a^\dag{}^2 + a^2$, which is both dynamically and thermodynamically unstable; or it can arise due the simultaneous emergence of a positive and negative energy excitation, i.e., a KC at a nonzero energy -- for example, in QBHs of the form $H\propto a^\dag a - b^\dag b$, which arise in certain cavity-QED contexts \cite{WiersigQED,Decon}. Notably, however, dynamically unstable systems need \emph{not} be thermodynamically unstable. The simplest such example is a free particle QBH, $H=p^2/2m$, which exhibits a linear-time dynamical instability $\braket{x(t)} = \braket{p(0)}t/m + \braket{x(0)}$,  despite being bounded from below. We note in passing that, in the context of QBLs, it is tempting to consider a notion of a QBL exhibiting a form of thermodynamical stability, if it admits a Gibbs state of a QBH as its (unique) SS. While a characterization of such semigroups is available \cite{Toscano}, we do not elaborate further on this here.

\subsection{One-dimensional bulk-translationally invariant systems}
\label{Toeplitz}

We now specifically focus on QBLs that are defined on a 1D lattice and enjoy {\em bulk-translation symmetry}, that is, whose dynamics are invariant under discrete translations up to BCs. More formally, if $S$ denotes any (unitary) discrete-translation operator, we demand it to be a {\em weak symmetry} of the dynamics \cite{VictorSymCQ,PostBosoranas}, in the sense that the corresponding superoperator commutes with the QBL, $[ {\cal L}, {\cal S}]=0$, with ${\cal S}(\rho)\equiv S \rho S^{-1}$. This will allow us to explore the dynamical consequences of changing the BCs while keeping the bulk invariant.

Consistent with the above requirement, we assume that all coherent and incoherent couplings depend only on the relative separation between sites (subject to BCs). Furthermore, we will assume that all couplings are of {\em finite range}, that is, the coupling between sites $j$ and $j+r$ vanishes for $r\geq R$, for some $R >0$. As such, BCs are encoded into modifications to the couplings between modes within $R$ sites of the boundary. Four main types of BCs play an important role in our analysis (see Fig.\,\ref{fig:BCs} for an illustration\vspace*{-1.5mm}):
\begin{itemize}
    \item {\em Open BCs (OBCs)}: A finite chain of lattice sites with a hard-wall boundary on each \vspace*{-1.5mm}side. 
    \item {\em Periodic BCs (PBCs)}: A finite number of lattice sites arranged on a ring, with no boundaries\vspace*{-1.5mm}.
    \item {\em Semi-infinite BCs (SIBCs)}: A pair of disjoint chains, each 
    extending infinitely in one direction, to the left or the right, with a boundary on the \vspace*{-1.5mm}other end.
    \item {\em Bi-infinite BCs (BIBCs)}: A chain extending infinitely in both directions, with no \vspace*{-1.5mm}boundaries.
\end{itemize}

\begin{figure}[t]
\centering
\includegraphics[width=0.9\columnwidth]{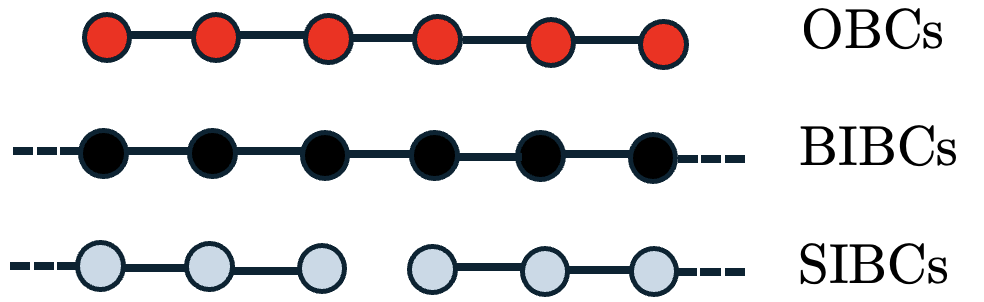}
\vspace*{-2mm}
\caption{Pictorial representation of boundary conditions of interest.}
\label{fig:BCs}
\end{figure}

The Hamiltonian and the dissipative contributions to a bulk-translationally invariant QBL then take the form:
\begin{align*}
    H&=\frac{1}{2}\sum_j\sum_{r=-R}^R\phi_j^{\dag}\mb h_r\phi_j, \\\
    \mathcal{D}^{\star}(A)&=\sum_j\sum_{r=-R}^R\Big(\phi_j^{\dag}A\,\mb m_r\phi_{j+r}-\frac{1}{2}\{\phi_j^{\dag}\mb m_r\phi_{j+r},A\}\Big),
\end{align*}
where $\phi_j \equiv [a_j,a_j^{\dag}]^T$ is the local Nambu array for site $j$, and $\mb h_r, \mb m_r$ are $2\times 2$ matrices that encode, respectively, coherent and incoherent couplings between sites $j,j+r$. Furthermore, 
$$\mb h_r^{\dag}=\mb h_{-r}, \quad \mb m_r^{\dag}=\mb m_{-r}, \quad \mb h_r=\boldsymbol{\tau}_1\mb h_r^{*}\boldsymbol{\tau}_1.$$ 
Imposing BCs is accomplished in two steps. First, restrict the sum over $j$ in the above expressions to include only physical lattice sites. For example, for a finite chain under OBCs or PBCs we have $j=1,\ldots,N$, while for SIBCs and BIBCs we take $j=1,2,\ldots$ and $j=0,\pm 1,\ldots$, respectively. Second, if imposing OBCs or SIBCs, set $\phi_{j+r} = 0$ if $j+r$ lies outside of the lattice label range. If imposing PBCs, take instead $\phi_{j+r} = \phi_{|N-(j+r)|}$ so that, for example, $\phi_{N+r} = \phi_{r}$.

From here, the relevant dynamical matrix is found to be:
\begin{align}
    \mb g_r\nonumber&=\boldsymbol{\tau}_3\mb h_r-\frac{i}{2}\boldsymbol{\tau}_3(\mb m_r-\boldsymbol{\tau}_1\mb m_r^{*}\boldsymbol{\tau}_1) =-\boldsymbol{\tau}_1\mb g_r^{*}\boldsymbol{\tau}_1 ,\\
    \mb G&=\mathds{1}_N\otimes \mb g_0+\sum_{r=1}^{R}\left(\mb S^r\otimes \mb g_r+\mb S^{\dag r}\otimes \mb g_{-r}\right).
    \label{dynmat}
\end{align}
Denoting $\vec{e}_j$ as the $j$-th canonical basis vector in ${\mathbb C}^N$, the BCs are then imposed through an operator $\mb S$ specified as follows:
\begin{align}
\label{Sshift}
    \mb S&\equiv \left\{
                \begin{array}{llll}
                \mb T_N &\!\!\equiv & \sum\limits_{j=1}^{N-1}\vec{e}_j\vec{e}_{j+1}^{\,\dag} & \text{OBCs,}\\ 
               { \mb T}&\!\!\equiv &\sum\limits_{j=-\infty,\neq 0}^{\infty}\vec{e}_j\vec{e}_{j+1}^{\dag} &\text{SIBCs,}\\
\mb V_N&\!\!\equiv & \sum\limits_{j=1}^{N-1} (\vec{e}_j\vec{e}_{j+1}^{\,\dag}+ \vec{e}_N\vec{e}_{1}^{\,\dag}) & \text{PBCs,}\\
\mb V&\!\!\equiv &\sum\limits_{j=-\inf}^{\inf}\vec{e}_j\vec{e}_{j+1}^{\,\dag} & \text{BIBCs.}
                \end{array}
              \right.
\end{align}
We will drop the subscript $N$ when the context is clear. 

For a translationally invariant system under PBCs, or BIBCs, the dynamical matrix in Eq.\,\eqref{dynmat} takes a \textit{banded block-circulant} matrix, or a \textit{banded block-Laurent} operator form, respectively. Here, the ``banded'' qualifier is a result of considering finite range couplings. The rapidities of such a system can be obtained through a block-diagonalization of the translational symmetry, by moving to momentum (Fourier) space. Focusing on the case of BIBCs for simplicity, let 
$$b_k\equiv \sum\limits_{j\in \mathbb{Z}} e^{-ikj}a_j, \quad k\in [-\pi,\pi],$$ 
with $k$ being the conserved crystal momentum in the  Brillouin zone. The dynamics of the resulting Fourier modes $\widetilde{\phi}_k \equiv [b_k, b_{-k}^{\dag}]^T$ are then governed by the \textit{Bloch dynamical matrix}: 
\begin{align}
 \dot{  \widetilde{\phi}}_k(t) = -i\mathbf{g}(k)\widetilde{\phi}_k(t),\quad  \mb g(k)\equiv\sum\limits_{r=-R}^R\mb g_re^{ikr}.
\end{align}
The rapidities of the system under BIBCs are thus given by the \textit{rapidity bands} $\lambda_n(k) \in \sigma(-i\mathbf{g}(k))$, with $n$ being the band index. For PBCs, the argument follows an identical route, except for the fact that the continuous Brillouin zone is replaced with an appropriate discrete subset of $N$ crystal momenta $k_m$. The full set of PBC rapidities are then $\{\lambda_n(k_m)\}$. Thus, the PBC rapidities converge to the BIBC rapidities as $N\rightarrow\inf$, in which case the set $\{k_m\}$ becomes dense in $[-\pi,\pi]$. Importantly, we will see that such a convergence property need {\em not} hold for systems under OBCs or SIBCs. 

Imposing hard-wall boundaries greatly complicates the problem of computing rapidities. Mathematically, the dynamical matrix of a single-boundary system 
%case 
(SIBCs) corresponds to a direct sum of two \textit{banded block-Toeplitz operators}, one acting on a semi-infinite system with a left and the other on the one with a right boundary (see Appendix \ref{app: defineSIBCs}). Thankfully, the spectrum of this class of operators has been completely characterized in terms of the corresponding Laurent operator (i.e., the BIBC dynamical matrix) and its \textit{symbol}, 
\begin{equation}
\mb g(z) \equiv \sum_{r=-R}^R \mathbf{g}_r z^r, \quad z\in{\mathbb C}. 
\label{Sym}
\end{equation}
Notably, the restriction of the symbol to the unit circle corresponds to the Bloch dynamical matrix \cite{BottcherToe}. Applying these results to our context, it follows that the rapidity spectrum of the SIBC system is given by the rapidity spectrum of the corresponding BIBC system, together with all points $\lambda$ such that $-i\mb{g}(z) - \lambda \mathds{1}_2$ has nonzero \textit{partial indices}. While a complete technical account of partial indices is beyond the scope of this paper, it is
appropriate to think of them as a suitable generalization of winding numbers of the rapidity bands $\lambda_n(k)$ about the point $\lambda$ \cite{PostBosoranas}, for operators with a block structure. Partial indices may be computed through the use of Wiener-Hopf matrix factorization techniques, originally developed for 1D fermionic Hamiltonians in \cite{WienerHopf} (see Appendix \ref{app:WH}). In Appendix \ref{App: WHApplication}, we provide an explicit application of these methods to QBLs of interest.
 
As for a system under OBCs, the corresponding dynamical matrix $\mb G_N$ is a \textit{banded block-Toeplitz matrix}. Unlike the previous three BCs, there is no simple prescription one can follow to characterize its rapidities. 
At variance with PBCs, in general we may now expect that 
\begin{equation}
\sigma(\mb G^{\text{OBC}}_N)\nsubseteq \sigma(\mb G^{\text{SIBC}})\equiv \sigma(\mb G_{\inf}).
\label{specdis0}
\end{equation}
Most surprisingly, the spectrum of a system under OBCs \textit{need not converge} to that of the semi-infinite system in general \footnote{In this work, the limit of a sequence of sets of complex numbers is taken to be the \textit{uniform limiting set} (see Chap. 3.5 of Ref.\,\cite{BottcherToe}). Formally, $\lambda$ is in the uniform limiting set of a sequence of sets of complex numbers $\{A_n\}$ if there is a sequence $\{\lambda_n\}$, with $\lambda_n\in A_n$ for all $n$ and $\lambda_n\to \lambda$.}:
\begin{align}
\label{specdis}
\lim_{N\to\infty}\sigma(\mb G^{\text{OBC}}_N)\neq \sigma(\mb G_{\inf}),
\end{align}
This so-called \textit{spectral disagreement}, which represents the degree to which a finite-size system truncation fails to display the characteristics of the infinite-size system (or, vice-versa, the degree to which the infinite-size idealization fails to predict characteristics of the realistic finite-size system), is known to be a consequence of extreme non-normality \cite{TrefethenPS,PostBosoranas}. Notwithstanding, useful connections between finite and infinite-size systems may be established if we instead consider a suitably generalized notion of the spectrum. 

\subsection{The pseudospectrum}
\label{pseudospectra}

Due to their extreme sensitivity to perturbations, non-normal operators are better characterized by their {\em pseudospectra}, rather than spectra \cite{TrefethenPS}. Simply put, the pseudospectrum of a matrix or operator is its ``approximate spectrum''. However, for non-normal matrices, the points in the pseudospectrum are not necessarily ``close" to the points in the spectrum. Mathematically, the $\epsilon$-pseudospectrum of a matrix $\mb A$ is defined as:
\begin{align}
\sigma_{\epsilon}(\mb A) \equiv \{z \in \mathbb{C}: \norm{ (\mb A-z  \mathds{1})^{-1}}  > \epsilon^{-1}\}.
\label{pseudospectra1}
\end{align}
We will work with $\norm{ \cdot}=\norm{\cdot}_2$, the induced matrix $2$-norm, however, unless otherwise specified, all results hold for other induced norms. While for normal matrices the \textit{resolvent norm} $\norm{(\mb A-z \mathds{1})^{-1}}$ is only ``large" when $z$ is close to some $\lambda\in\sigma(\mb A)$ (say, within an $\epsilon$-ball), this is not generically true for non-normal ones. Generally, one can show that:  
\begin{align*}
\sigma_{\epsilon}(\mb A) \supseteq \set{\lambda : |\lambda-\lambda_0|<\epsilon,\,\lambda_0\in\sigma(\mathbf{A})}.
\end{align*}
For the specific choice of a matrix $2$-norm, the inclusion becomes an inequality when $\mb A$ is normal. However, for a highly non-normal $\mb A$, $\sigma_{\epsilon}(\mb A)$ is no longer just the union of open balls of radius $\epsilon$ around each  point in the spectrum. Thus, arbitrarily small perturbations can drastically affect the spectrum of a non-normal matrix. To understand why, an equivalent definition of the pseudospectrum to the one given in Eq.\,\eqref{pseudospectra1} is useful. Namely,
\begin{equation}
\label{pseudospectra2}
\sigma_{\epsilon}(\mb A) =\{ z \in {\mathbb C} : z \in  \sigma (\mb A+\mb E), \, \forall \norm{\mb E} < \epsilon \}.
\end{equation}
In words, every point in the $\epsilon$-pseudospectrum of $\mb A$ is an eigenvalue of some matrix obtained by perturbing $\mb A$ with a perturbation of size less than $\epsilon$. This alternative characterization of pseudospectrum also reveals its intrinsic robustness against perturbations: Eq.\,\eqref{pseudospectra2} implies that 
$$\sigma_{\epsilon}(\mathbf{A}+\mathbf{F}) \subseteq \sigma_{\epsilon + \delta}(\mathbf{A}), \quad \norm{\mathbf{F}} = \delta. $$

The pseudospectra for the classes of operators considered in the preceding subsection have been extensively studied in the mathematical literature \cite{BottcherToe}. One of the most useful results for our purposes captures the well-behaved nature of the pseudospectrum as $N\to\infty$, that is (see also Appendix \ref{app: defineSIBCs}):
\begin{align}
\label{wellbehavedps}
\underset{\epsilon\rightarrow 0}{\lim} \underset{N\rightarrow\inf}{\lim}\sigma_{\epsilon}(\mb G^{\text{OBC}}_N) 
=\sigma(\mb G_{\infty}).
%\sigma(\mb{G}_\infty)\cup\sigma(\mb{\tilde{G}}_\infty).
\end{align}
%Here, $\mb{\tilde{G}}$ is the so-called ``associated block-Toeplitz operator" of $\mb{G}$, obtained by swapping the left and right translation operators $\mathbf{T}$ and $\mathbf{T}^\dag$ in Eq.\,\eqref{Sshift}; physically, this corresponds to interchanging a semi-infinite system with a right boundary for one with a left boundary. Despite this technicality, one may safely think of the right hand-side of Eq.\,\eqref{wellbehavedps} as the proper infinite-size spectrum.} 
This equality is one of the powerful tools of pseudospectral theory. In words, {\em as $N$ grows, the pseudospectrum of the OBC dynamical matrix 
%$\mb G_N$ 
approximates increasingly better the semi-infinite spectrum.} To put it another way, given $z$ in the SIBC spectrum, for large enough $N$, $z$ will be in the $\epsilon$-pseudospectra of the OBC chain, for arbitrary $\epsilon>0$. In this sense, the $\epsilon$-pseudospectra, with $\epsilon$ small, are the ``imprints'' of the exact infinite-size spectrum. 

Two remarks are in order. First, in light of the pseudospectral formulation of Eq.\,\eqref{wellbehavedps},  the generic spectral disagreement for Toeplitz operators captured by Eq.\,\eqref{specdis} can then be re-expressed as a non-commuting limit
\begin{align}
\label{non-commutative limits}
    \underset{N\rightarrow \inf}{\lim}\underset{\epsilon\rightarrow 0}{\lim}\,\sigma(\mb G^{\text{OBC}}_N)\neq  \underset{\epsilon\rightarrow 0}{\lim} \underset{N\rightarrow\inf}{\lim}\sigma_{\epsilon}(\mb G^{\text{OBC}}_N).
\end{align}
This appears strikingly similar to the non-commuting limits typically encountered in the theory of spontaneous symmetry breaking. Second, a non-rigorous physical argument can be provided to understand why the $\epsilon$-pseudospectrum does capture the infinite-size spectrum increasingly well. Suppose that $\vec{v}$ is an eigenvector of the infinite-size system corresponding to an eigenvalue $z$ with {\mbox{$|z|<1$}} (and thus, not in the BIBC spectrum). Truncating this edge-localized vector by retaining only the weights on the first $N$ sites will, generically, not provide an exact eigenvector of the finite system. Nonetheless, it is reasonable to expect that this truncation will provide an \textit{approximate eigenvector}, with a correction depending on the size of the removed tail piece (which is set by the localization length $|\log|z||^{-1}$). Since the removed tail is exponentially small as $N$ grows, the error $\epsilon$ will shrink accordingly. 

As it turns out, the notion of an approximate eigenvector introduced above is a useful conceptual tool in and of itself, and can be used to obtain yet another equivalent definition of the pseudospectrum. Namely, 
\begin{align}
   \sigma_{\epsilon}(\mathbf{A})=\{z\in \mathbb{C}: \norm{(\mb A-z\mb I)v}<\epsilon, \norm{\vec{v}}=1\}.
\end{align}
Here, $z\in \mathbb{C}$ is an \textit{$\epsilon$-pseudoeigenvalue} of $\mb A$, with $\vec{v}\in \mathbb{C}^n$ the corresponding \textit{$\epsilon$-pseudoeigenvector}. From a dynamical standpoint, for sufficiently small time-scales, pseudo-eigenvectors behave like exact normal modes \cite{TrefethenPS,PostBosoranas}. This means that while the spectrum governs the long-time dynamics, the pseudospectrum is responsible for the short-time transient dynamics. In highly non-normal cases, this is known to engender sharply different transient and asymptotic dynamics \cite{TrefethenPS}. In the following section, we turn to characterizing this phenomenon, which we term dynamical metastability, using tools and insights of pseudospectral theory.

\section{Dynamical metastability and the imprint of the infinite-size limit}
\label{sec:dynmeta}

\subsection{The landscape of possibilities}

Let us summarize the four key facts identified so far:

\smallskip

(1) The dynamical stability of a QBL is largely controlled by the spectral structure of its dynamical matrix.

(2) The dynamical matrix of a 1D bulk-translationally invariant QBL with open boundaries or full translation invariance is, for finite $N$, either a block-Toeplitz or block-circulant matrix, respectively. For an infinite number of sites, $N\to \infty$, it is a block-Toeplitz or a block-Laurent operator, respectively.

(3) The spectra of these well-investigated matrices can depend strongly on the system size and fail to converge to those of their infinite-size counterparts; the pseudospectrum provides a more reliable, well-behaved tool for bridging matrices ($N$ finite) to operators ($N$ infinite).

(4) The pseudospectrum  plays an essential role in determining the transient behavior of a non-normal dynamical system. 

\smallskip

Altogether, the above observations lead us to introduce a dynamical classification scheme for 1D bulk-translation invariant QBLs based on the notions of {\bf (i)} dynamical-stability disagreement and {\bf (ii)} stability gap discontinuity -- between finite and infinite-size systems (see Table \ref{t: classes }).

\begin{table}[t]
\begin{center}
\begin{tabular}{ ||c| c| c ||}
\hline
QBL Class & \begin{tabular}{c} Dynamical Stability\vspace*{-1mm} \\
Disagreement \end{tabular} 
& \begin{tabular}{c} Stability Gap\vspace*{-1mm} \\ 
Discontinuity \end{tabular}  \\
\hline\hline
 Type I DM & \ding{51} & \ding{51}  \\ 
 \hline
 Type II DM & \ding{51} & \ding{55}  \\ 
\hline
 Anomalously relaxing & \ding{55} & \ding{51}  \\ 
 \hline
 Well-behaved & \ding{55} & \ding{55}  \\ 
  \hline
\end{tabular}
\caption{Classification of bulk-translationally invariant QBLs with boundaries, based on the notions of dynamical and spectral disagreement between finite vs. infinite size. Systems with dynamical disagreement exhibit a stability phase transition either as $N\to\infty$, or for $N > N_c$. For systems with a discontinuous (as $N\rightarrow \infty$) stability gap, the finite-size spectrum does not converge to that of the corresponding infinite-size system with growing system size. Note that, in principle, it is possible for the full spectrum to display a discontinuity in the infinite-size limit while maintaining continuity of the stability gap. We leave explorations of this case to future work.
}
\label{t: classes }
\end{center}
\end{table}

To be more concrete, given a fixed set of (coherent and incoherent) finite-range couplings, one may consider two families of QBLs. The translationally-invariant family consists of the sequence of QBLs defined on finite rings of arbitrary size $N$ (PBCs), in addition to the bi-infinite system (BIBC). The open-boundary family consists of a sequence of QBLs defined on finite chains of arbitrary size $N$ (OBCs), in addition to the semi-infinite system with only one boundary (SIBC). We say that a sequence, in system size, of QBLs exhibits: 

\smallskip

{\bf (i)} A \textit{dynamical-stability disagreement} if the dynamical stability phases of the finite- (perhaps up to a ``critical size'' $N_c$) and infinite-size members of the family differ, with all other parameters fixed; 

{\bf (ii)} A \textit{stability gap discontinuity} if the stability gap of the infinite-size system is different from the limit, as $N\to\infty$, of the finite-size stability gaps. 

\smallskip 

In light of these concepts, the translationally-invariant and open boundary families of QBLs are drastically different. On the one hand, the members of the translationally-invariant family cannot exhibit dynamical disagreement nor a discontinuity of the stability gap. This follows from the aforementioned property of the PBC spectrum being a strict, but increasingly dense, subset of the BIBC spectrum. Such translationally-invariant dynamical matrices are unitarily block-diagonalizable (with the block-size being independent of $N$), through the Fourier transform. Thus, for PBCs, effects of non-normality cannot be exacerbated by increasing system size.  On the other hand, the open boundary QBL family splits into four distinct classes based on the presence or absence of dynamical-stability disagreement and the (dis-)continuity of the stability gap with system-size. For reasons we shall momentarily explain, we describe systems exhibiting dynamical-stability disagreement as \textit{dynamically metastable} (DM) and further subdivide them into types I and II, depending on whether there is a stability gap discontinuity or not, respectively. For non-DM families, we dub those with a stability gap discontinuity \textit{anomalously relaxing} and those without \textit{well-behaved}. The resulting dynamical classification of QBLs is summarized in Table \ref{t: classes }. 

Before providing physical realizations of these possible scenarios, it is worth reiterating how they arise mathematically. The properties of block-Toeplitz matrices prevent one from establishing any precise relationships between the finite-OBC stability gap of a given QBL, $\Delta_{S,N}^{\text{OBC}}$, and its infinite-size counterpart, $\Delta_{S}^{\text{SIBC}}$. That said, we can prove the following limiting bound (Appendix \ref{BoundProof}):
\begin{align}
\label{gap constraint}
\underset{N\rightarrow\inf}{\lim}\Delta_{S,N}^{\text{OBC}} \equiv \Delta_{S,\infty}^{\text{OBC}} \leq \Delta_{S}^{\text{SIBC}}.
\end{align} 
The fact that this bound is {\em not} always an equality is a manifestation of the spectral discontinuities characteristic of highly non-normal Toeplitz matrices. The  strict inequality is the defining feature of a stability gap discontinuity. In addition, the inequality Eq.\,\eqref{gap constraint} provides further insight into the case of a dynamically stable infinite system, for which $\Delta_S^{\text{SIBC}}< 0$. It follows that the finite system must become more and more stable as $N$ grows (if it is unstable at any point to begin with), as the finite gaps must eventually land on, or below, $\Delta_S^{\text{SIBC}}$.

We conclude with two remarks. First, the relationship between the spectrum of a block-Laurent operator and its block-Toeplitz counterpart implies that the stability gap of a semi-infinite system is always bounded below by that of the associated bi-infinite (that is, bulk) system, $\Delta_{S}^{\text{SIBC}}\geq \Delta_{S}^{\text{BIBC}}$. Hence, inserting a single boundary is not sufficient to stabilize an unstable bi-infinite system. Nonetheless, it is not forbidden for a bulk-stable system to become dynamically unstable upon imposition of a single boundary; see \cite{BarnettEdgeInstab} and Appendix \ref{app: Boundary-dependent stability} for more details. 
Second, the way in which $\Delta_{S,N}^{\text{OBC}}$ converges with growing \(N\) need {\em not} be monotonic. In particular, $\Delta_{S,N}^{\text{OBC}}$ may approach its limiting value $\Delta_{S,\infty}^{\text{OBC}}$ from below or above, or it can encircle it. These behaviors lead to dynamical disagreement {\em without} a stability gap discontinuity.

\begin{figure}[t]
\centering
\includegraphics[width=.42\textwidth]{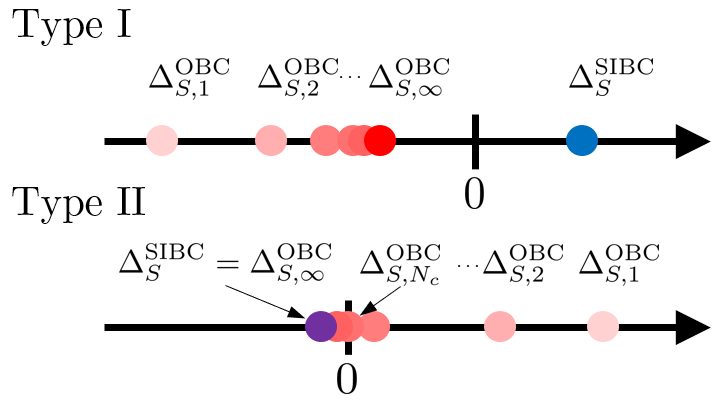}
\vspace*{-2mm}
\caption{Relationship between finite and infinite-size stability gaps in type I and type II dynamically metastable QBLs. }
\label{fig:stability gap comparison}
\end{figure}

\subsection{Type I dynamical metastability }

Type I DM and anomalously relaxing systems have been discussed in detail in Refs.\,\cite{Bosoranas, PostBosoranas}. We recount the basic physical and mathematical arguments here for completeness, and refer the reader to the above papers for additional discussion and connections to related phenomena in the literature -- notably, topological amplification \cite{PorrasTopoAmp,PorrasIO,NunnenkampTopoAmp,NunnenkampDisorder,NunnenkampResoration} 
as well as anomalous relaxation in systems exhibiting the Liouvillian NHSE \cite{UedaSkin,MoriSlow,Bergholtz,WangSkin}. 

A generic example of type I dynamical metastability is an open-boundary family of QBLs with the following properties (see Fig.\,\ref{fig:stability gap comparison}(top) for a pictorial representation):
\begin{itemize}
    \item $\Delta_{S,N}^{\text{OBC}}<0$, for all $N$; and
    \item  $\Delta_{S,\infty}^{\text{OBC}}<0$, 
    %$\lim\limits_{N\to\infty} \Delta_{S,N}^{\text{OBC}}<0$, 
    but $\Delta_S^{\text{SIBC}}>0$. 
\end{itemize}
Concrete models will be presented later in Sec.\,\ref{illmod}. The discontinuous jump from a strictly negative stability gap for finite $N$ to a strictly positive one for the infinite system indicates the spontaneous loss of dynamical stability in the infinite-size limit. Such families break several intuitions about the relationship between finite- and infinite-size systems. Each finite chain possesses a unique, globally attractive SS, while the infinite-size limit lacks one altogether. For finite size, the asymptotic relaxation to the SS happens at a rate set by the finite-size dissipative gap, $\Delta_{\mathcal{L},N}^{\text{OBC}} = |\Delta_{S,N}^\text{OBC}|$ which, by assumption, converges to a finite nonzero value, $|\Delta_{S,\infty}^{\text{OBC}}|$. From this, one might expect that the infinite-size system \textit{should} display asymptotic relaxation to some SS set by this rate. 

The apparent contradiction is resolved by noting that the dissipative gap alone does {\em not} predict the global character of the dynamics,  
combined with the fact that the infinite-dimensional nature of the systems opens the door to dynamical instabilities. Consistent with the 
presence of a constant prefactor $K$ in Eq.\,\eqref{expbound}, the dissipative gap yields no information on the 
time at which the asymptotic relaxation sets in. We may then infer that, as system size grows, type I DM systems exhibit an increasingly long transient timescale during which the general character of the dynamics may differ significantly from the relaxing asymptotic one. Under this assumption, the true infinite-size limit displays only those features that emerge in the finite-size transient regime. This fact reflects the necessary relationship between bulk dynamics and transient dynamics. By causality, an excitation created in the bulk of any finite chain will take a finite time (increasing with system size) to detect the (in this case stabilizing) presence of the boundaries.  Thus, any component of a given excitation deep within the bulk will amplify, until being eventually suppressed. 

The above picture can be made more quantitative with the tools of pseudospectral analysis. The existence of a stability gap discontinuity demands that there are 
eigenvalues of the SIBC system that are neither in the finite-size spectra for any $N$ nor are well-approximated by any sequence of finite-size eigenvalues. In particular, there exists a rapidity $\lambda\in\sigma(-i\mathbf{G}_\infty)$, with $\text{Re}[\lambda] = \Delta_{S}^\text{SIBC}>0$. As $N$ increases, there will necessarily be an $\epsilon_N$-pseudoeigenvector $\vec{v}_N$ of $\mathbf{G}_N^\text{OBC}$ corresponding to $\lambda$, such that $\epsilon_N$ decreases with $N$. Mathematically, we may write $\mathbf{G}_N^\text{OBC}\vec{v}_N = \lambda\vec{v}_N + \vec{w}$, with $\norm{\vec{w}}<\epsilon_N$, so that:
\begin{align*}
\norm{e^{-i\mathbf{G}_N^\text{OBC} t}\vec{v}_N-e^{\lambda t} \vec{v}_N} = \epsilon_N \left( t +\mathcal{O}(t^2)\right).
\end{align*}
Accordingly, as long as higher order terms are sufficiently small, $\vec{v}_N$ amplifies exponentially at a rate set by $\Delta_{S}^\text{SIBC}>0$, for a transient time set by $1/\epsilon_N$.  As $N\to\infty$ and \mbox{$\epsilon_N\to 0$} (such that the duration of the transient diverges), $\vec{v}_N$ approaches an amplifying normal mode. It follows that any overlap of the initial condition $\braket{\Phi(0)}$ with $\vec{v}_N$ will amplify in the transient regime.

\subsection{Type II dynamical metastability}

Let us now consider an open-boundary family of QBLs without a stability gap discontinuity. The stability gap constraint in Eq.\,\eqref{gap constraint} yields no information about the trajectory of $\Delta_{S,N}^\text{OBC}$ as it converges to $\Delta_{S,\infty}^\text{OBC}$ which, under our assumptions, equals $\Delta_S^\text{SIBC}$. It is possible for such a system to have a stable infinite-size limit, while the finite-size truncations are unstable for generic $N$, smaller than some critical (possibly infinite) size $N_c$. In this scenario (see Fig.\,\ref{fig:stability gap comparison} (bottom)), the dynamical stability phases of the finite chains and the infinite-size chain are reversed with respect to the type I scenario. 

Physically, type II DM systems require two basic features: 
\begin{itemize}
    \item bulk stability, $\Delta_S^\text{SIBC}<0$; 
    \item a destabilizing mechanism that can
    make the system unstable under OBCs, $\Delta_{S,N}^{\text{OBC}}>0$, for generic $N<N_c$.
\end{itemize}
When these prerequisites are met, type II DM systems may be thought of as exhibiting ``system-size assisted dynamical stabilization''. Finite systems of size $N<N_c$ will generically display asymptotically amplifying dynamics with a characteristic amplification rate given by 
\[  \Delta_{S,N}^\text{OBC}   \overset{N\to N_c}{\longrightarrow} 0.\]
Thus, as the system-size increases, even the most unstable normal mode (i.e., one whose rapidity has real part equal to $\Delta_{S,N}^\text{OBC}$) will amplify at an ever-decreasing rate. Moreover, as $N\to\infty$, generic excitations within the bulk will take longer and longer to detect the boundary, and hence any boundary-induced instabilities. 
Therefore, the onset of unstable asymptotic dynamics will become increasingly delayed. In this way, system size effectively becomes a parameter one may tune to achieve dynamical stability.

Pseudospectral analysis again provides uniquely effective tools for characterizing type II DM systems. The generic relationship between the infinite-size spectrum and the finite-size pseudospectrum allows us to identify the necessary emergence of transiently stable pseudomodes. Following a nearly analogous argument as for transient amplification in type I DM systems, any pseudoeigenvector corresponding to a pseudoeigenvalue present in the infinite-size spectrum must necessarily appear more and more stable as $N$ increases (and thus $\epsilon_N$ decreases). The stable, infinite-size normal modes imprint into the transient dynamics of type II DM systems, despite their overall amplifying, unstable dynamics.

To conclude the discussion of type II DM, we remark that, while extreme non-normality is a necessary ingredient for type I DM (and the soon-to-be-discussed anomalously relaxing) systems, it {\em need not} play as prominent a role in type II DM systems. In fact, to the best of our knowledge, there is no analogous mathematical characterization of type II DM systems. Nonetheless, in Sec.\,\ref{illmod} we will introduce a promising design principle for type II DM systems in the limiting case of purely Hamiltonian dynamics, built upon known facts about the stability properties of spectral degeneracies, which can emerge when a QBH loses thermodynamic stability \cite{Decon,Squaring}.

\subsection{Anomalous relaxation}

Like type I DM systems, anomalously relaxing QBLs exhibit a stability gap discontinuity. However, the gap does not
become positive in the infinite-size limit, and thus dynamical stability persists for the infinite system. That is, we have 
\[ \Delta_{S,\infty}^{\text{OBC}}\leq \Delta_{S}^{\text{SIBC}} <0.\]
The arguments used to characterize the transient dynamics of type I DM systems carry over to anomalously relaxing ones, with any mention of ``amplification" replaced with ``decay at a rate slower than any normal mode decay rate in the system." Formally, the stability gap discontinuity guarantees that there will be an $\epsilon_N$-pseudoeigenvector of the finite chain that appears to decay with a rate set by $|\Delta^\text{SIBC}_S|$ during a transient time (whose duration again scales with $1/\epsilon_N\to\infty$), until eventually becoming bounded by an exponential decay with a rate $|\Delta_{S,N}^\text{OBC}|>|\Delta_S^\text{SIBC}|$. We refer the reader to Ref.\,\cite{PostBosoranas} for a demonstration of this separation of timescales.

\section{Illustrative models}
\label{illmod}

\subsection{A dissipative model: The coupled Hatano-Nelson chain}

\begin{figure}[t]
\centering
\includegraphics[width=.9\columnwidth]{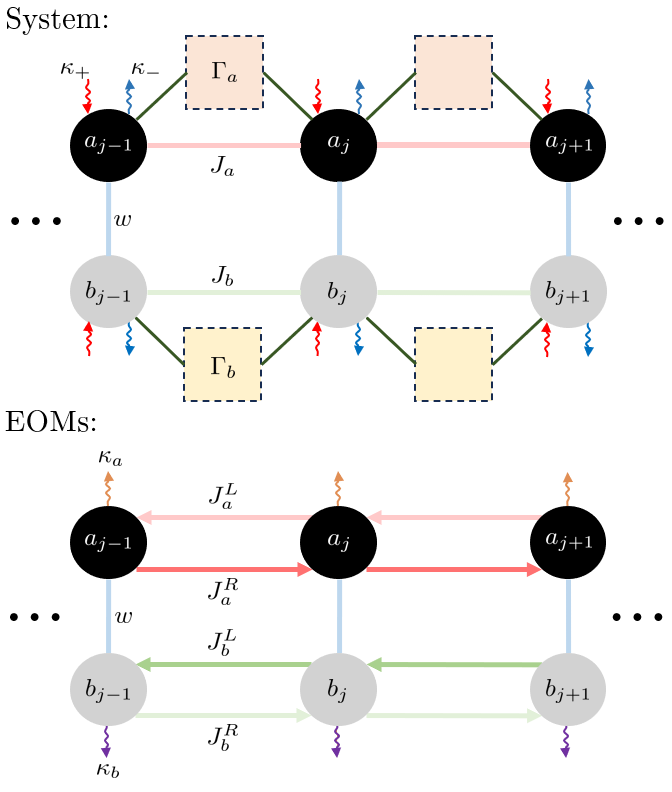}
\vspace*{-2mm}
\caption{Pictorial representation of the bosonic coupled HN chain model. {\bf Top:} The physical system under consideration. In particular, the dissipative couplings $\Gamma_a$ and $\Gamma_b$ are represented by coupling of nearest neighbors to pairwise-common baths (denoted by boxes with dashed lines). The site-local loss and gain rates are understood to stem from an ambient (say, thermal) bath not depicted. {\bf Bottom:} A schematic representation of the effective couplings as informed directly by the EOMs. In both diagrams, the ellipses represent the appropriate BC-dependent continuations. }
\label{fig:Coupled HN}
\end{figure}

In order to exemplify the above general framework and gain concrete insight, we now aim to design a model which can be driven through different dynamical metastability regimes upon appropriately tuning relevant parameters.  The guiding design principle we employ is to identify a model that unambiguously exhibits one form of dynamical metastability (say, type I) for one set of parameters and is fully dynamically stable (not DM) for another. 

The first such model we propose describes the dissipative dynamics of two coherently coupled species of bosons, denoted by $a_j$ and $b_j$ (see Fig.\,\ref{fig:Coupled HN} for a depiction). Specifically, the system Hamiltonian consists of three parts, $H \equiv H_a + H_b + H_{ab}$, where 
\begin{align*}
    H_a \nonumber&= iJ_a\sum_{j} a_{j+1}^\dag a_j -\text{H.c.}, \\
    H_b \nonumber&= iJ_b\sum_{j} b_{j+1}^\dag b_j -\text{H.c.}, \\
    H_{ab} &= iw\sum_{j} a_j^\dag b_j - \text{H.c.}, 
\end{align*}
respectively. In the above, $J_{a,b}\in\mathbb{R}$ denote coherent intra-chain hopping amplitudes, while $w\in\mathbb{R}$ is a coherent inter-chain hopping amplitude. As to the the dissipative part of the QBL, for simplicity we take each mode to undergo local single-photon loss and gain at uniform rates $\kappa_-\geq 0$ and $\kappa_+\geq 0$, respectively. In addition, we introduce an intra-chain, nearest-neighbor dissipation mechanism with rate $\Gamma_{a,b}\geq 0$ and an inter-chain phase difference $\theta\in\mathbb{R}$. Altogether, the intra-chain dissipators are 
\footnote{Imposing OBCs in this system can be done in several ways. Here, we will do so in such a way that site-local dissipative processes are unaffected while all dissipative couplings across sites $1$ and $N$ vanish. Since terms like $\sum_j \Gamma_a \mathcal{D}[a_j+a_{j+1}]$ (and similarly for $b$) mediate both a dissipative coupling of strength $\Gamma_a$ and an additional site-local loss of rate $2\Gamma_a$, the only way to accomplish this is to physically ``cleave'' 
the common bath that couples to sites $1$ and $N$ (see Fig.\,\ref{fig:Coupled HN}), such that each maintains 
a loss rate of $2\Gamma_a$ in addition to the intrinsic local loss $\kappa_a$. That is, we first symmetrize the incoherent coupling according to $\sum_{j}\mathcal{D}[a_j+a_{j+1}] =\frac{1}{2}\sum_{j}\left(\mathcal{D}[a_j+a_{j+1}]+\mathcal{D}[a_{j-1}+a_{j}]\right)$, and impose $a_{0}=a_{N+1}=0$ (similarly for the $b$-dissipator). Altogether, this process amounts to passing from a block-circulant GKLS matrix for PBCs to a block-Toeplitz one for OBCs.}:
\begin{align*}
    \mathcal{D}_a \nonumber&= 2\sum_j \big(\kappa_-\mathcal{D}[a_j] + \kappa_+\mathcal{D}[a_j^\dag]+\Gamma_a\mathcal{D}[a_j+a_{j+1}]\big) , \\
    \mathcal{D}_b &= 2\sum_j \big(\kappa_-\mathcal{D}[b_j] + \kappa_+\mathcal{D}[b_j^\dag]+\Gamma_b\mathcal{D}[b_j+e^{i\theta}b_{j+1}]\big).
\end{align*}

Note that the resulting QBL, $\mathcal{L} = -i[H,\cdot] + \mathcal{D}_a + \mathcal{D}_b$, is invariant under a global U(1) transformation of the form $(a_j,b_j)\mapsto (e^{i\phi}a_j,e^{i\phi}b_j)$, with $\phi\in\mathbb{R}$. As a consequence, the dynamics of the full Nambu array can be reduced to that of the sub-array $[a_1,b_1,\ldots,a_N,b_N]^T$. By introducing a new set of parameters, 
 \begin{align}
\kappa_{a,b} &\equiv \kappa_--\kappa_+ + 2\Gamma_{a,b}, \notag \\
J^L_{a} &\equiv J_{a}-\Gamma_{a},\quad  \;\;\;\;\;\;J^R_{a} \equiv J_{a}+\Gamma_{a}, \notag \\
J^L_{b} &\equiv J_{b}-e^{-i\theta}\Gamma_{b},\quad  J^R_{b} \equiv J_{b}+e^{i\theta}\Gamma_{b},
\label{newP}
\end{align}
we obtain the EOMs as follows:
\begin{equation} 
\left\{ \begin{array}{ll}
\dot{a}_j &=J^L_a a_{j-1}-J^{R}_a a_{j+1}+w b_j-\kappa_a a_j,\\
 \dot{b}_j&=J^L_b b_{j-1}-J^{R}_b b_{j+1}-w a_j-\kappa_b b_j. \end{array}\right.
\label{CHN}
\end{equation}
Remarkably, these EOMs are equivalent to two coherently coupled HN \cite{HNChain} (or, asymmetric hopping) chains.  We remark that our new set of parameters in Eq.\,\eqref{newP} satisfies the following constraint:
\begin{align}
 \!\!   \kappa_a-\kappa_b &= 2(\Gamma_a-\Gamma_b) = \frac{(J_b^L-J_b^R)}{\cos\theta} -(J_a^L-J_a^R).
    \label{constraint}
\end{align}
For the remainder of our analysis, we assume that the coherent couplings dominate over the incoherent ones, in the sense that $|J_x|\geq \Gamma_x,$, $x=a,b$.
Further, we will limit ourselves to $\theta=0$ and $\theta=\pi$. For reasons that will become clear, we refer to these as the {\em non-chiral} and {\em chiral} regimes, respectively.

Immediately, two limiting scenarios arise. First, in the limit $w\gg J_x^L,J_x^R$ the dynamics are effectively described by coherent hopping between the two chains. Naturally, we expect this limit to be dynamically stable, independent of system size or BCs. The opposite limit $w\to 0$, i.e., the case of two disconnected HN chains, is known to be extremely sensitive to changes in BCs, hence also to taking the infinite-size limit. While the physics of directional amplification in such 1D chains has been considered in detail (see e.g., \cite{NunnenkampTopoAmp}), it will prove useful to first apply our theory of dynamical metastability in this decoupled regime and use it as a foundation for understanding the (significantly more complicated) intermediate regime $w\sim J_x^L,J_x^R$. 

\subsubsection{The decoupled limit $w=0$}

\begin{figure}[t]
 \centering
\includegraphics[width=.9\columnwidth]{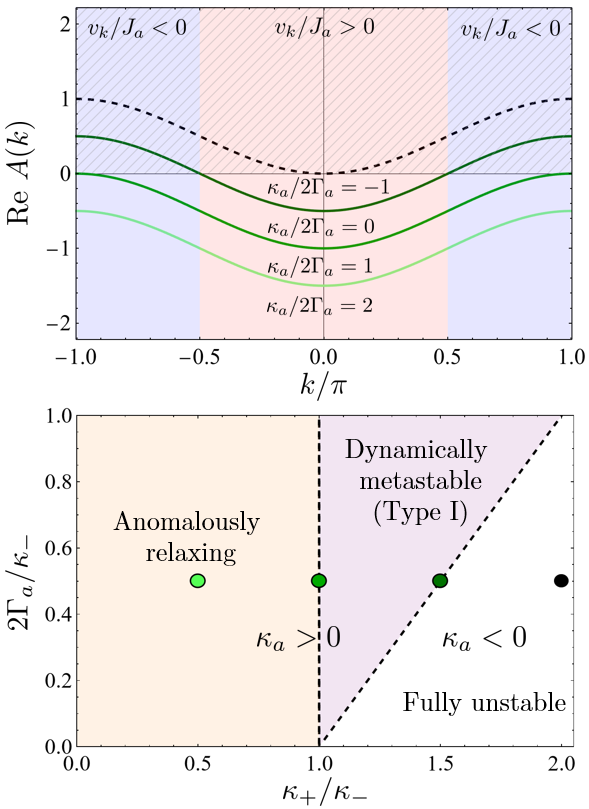}
\vspace*{-3mm}
\caption{{\bf Top:} 
Directional amplification in the HN chain in the DM phase. Plotted are the curves $\text{Re} A(k)$, which determine the damping or growth rate, depending on the sign [Eq.\,\eqref{wpvel}]. Unstable bulk modes correspond to $\text{Re} A(k)>0$ (hatched region). Imposed on the figure are the intervals of $k$ where wavepackets evolve in a given direction. specified by $\text{sgn}\,v_k$. Notably, whenever $0\leq \kappa_a/2\Gamma_a\leq 1$, which corresponds to the type I DM phase, {\em all} unstable modes propagate in the same direction. In contrast, the anomalously relaxing phase has no unstable modes, while the fully unstable phase has unstable modes propagating in 
both directions.
{\bf Bottom:} Dynamical metastability phase diagram of the bosonic 
coupled HN chains in the decoupled limit, in terms of the raw system (dissipative) parameters. The marked (green) points correspond to parameter choices in the upper figure. Units are such that $\kappa_-=1$. }
\label{fig:Chiral HN}
\end{figure}

Since our goal is now to describe the dynamical metastability theory of a single HN chain we can, without loss of generality, focus on the $a$-chain alone. To understand the bulk-dynamics, let us define $a(k) \equiv \sum_j e^{-ijk}a_j$. It then follows that $\dot{a}(k) = A(k)a(k)$, with 
\begin{align*}
    A(k) &= -\kappa_a+(J^L_a-J^R_a)\cos(k) -i (J_a^L+J_a^R)\sin(k)    \\
    &= -\kappa_a-2\Gamma_a\cos(k) -2i J_a\sin(k), 
\end{align*}
which describes an ellipse in the complex plane, centered at $(-\kappa_a,0)$. It follows that $\text{Re} A(k)\in[-\kappa_a-2\Gamma_a,-\kappa_a+2\Gamma_a]$. Thus, there will be unstable bulk modes whenever $2\Gamma_a\geq \kappa_a$ or equivalently, using Eq.\,\eqref{newP}, when the onsite gain $\kappa_+$ exceeds the onsite loss $\kappa_-$. Interestingly, this conditions holds despite the additional effective onsite loss rate of $2\Gamma_a$ provided by the dissipative coupling. This additional loss is compensated by the characteristic amplification rate $J^R_a-J^L_a=2\Gamma_a$ arising from the nonreciprocal nature of the HN chain.

When $\kappa_+>\kappa_-$, the bulk-unstable modes are manifestly directional whenever $\kappa_a>0$ (or $2\Gamma_a>\kappa_+-\kappa_-$). This can be seen by computing the group velocity 
of the $k$-th wavepacket. By writing $a(k,t) = e^{[ \text{Re} A(k) + i \text{Im} A(k)]t} a(k,0)$ and interpreting $-\text{Im} A(k)$ as the wavepacket energy, we have 
\begin{align}
\label{wpvel}
v_k \equiv -\frac{d}{dk}\text{Im}A(k) = 2J_a\cos(k).
\end{align}
We observe that the unstable bulk modes, i.e.,  those with $\text{Re}A(k)>0$, propagate in the same direction specified by the sign of $v_k$ whenever $\kappa_a>0$ (see Fig.\,\ref{fig:Chiral HN}, top). This parameter regime (namely, $\kappa_a >0$ and $\kappa_+>\kappa_-$) is particularly relevant for the OBC dynamics. 

Following the general arguments from pseudospectral analysis in Sec.\,\ref{pseudospectra}, the OBC chain will undergo transient amplification, manifesting in the form of unstable bulk pseudomodes, whenever the bulk is unstable ($\kappa_+>\kappa_-$). However, the overall chain will remain dynamically stable as long as $\kappa_a = 2\Gamma_a + \kappa_--\kappa_+>0$. This follows from the known normal mode spectrum of the HN model under OBCs \cite{HNChain}:
\begin{align*}
    \lambda_m = -\kappa_a+2i\sqrt{J_a^2 - \Gamma_a^2}\cos\left(\frac{m\pi}{N+1}\right),\quad m=1,\ldots,N.
\end{align*}
This implies that the decoupled chain exhibits type I DM, in addition to directional amplification, as long as $\kappa_+-2\Gamma_a\leq \kappa_- \leq \kappa_+$ or, equivalently, $0<\kappa_a/2\Gamma_a<1$ \footnote{The case $\kappa_a=0$ corresponds to $\Delta_{S,N}^{\text{OBC}}=0$.
Thus, the chain may be stable or unstable depending on the existence of a SS. We omit any such analysis, as it is not relevant to our present considerations.}. 

The stabilizing nature of boundaries in the HN chain can be explained from physical arguments. Because the instability of a given bulk mode necessitates a fixed directionality, persistent amplification can only occur if the mode is allowed to propagate indefinitely. Of course, this requires that the system lacks a boundary in the direction of amplification or that the chain is infinite in extent. 

%Finally, we remark that it may appear surprising that a chain with only hopping can engender dynamical instabilities. After all, the continuous U(1) symmetry of the system should imply a local conservation of particles. Quite surprisingly, this is not the case since Noether's theorem generically breaks down in the open system setting \cite{Albert}

\subsubsection{The non-chiral regime $\theta=0$}

Let us now couple the chains and further take $J_a=J_b\equiv J$ and $\Gamma_a=\Gamma_b\equiv \Gamma$ for simplicity. Note that, in view of the constraint in Eq.\,\eqref{constraint}, this also equalizes the effective loss rates of the two chains, i.e., $\kappa_a=\kappa_b\equiv \kappa$. It follows that, for $\theta=0$, the two HN chains support asymmetric hopping in the \textit{same} direction:
\begin{align}
\label{HNnonchiral}
J_b^L &= J_a^L \equiv J^L,  \quad J_b^R = J_a^R \equiv J^R.
\end{align}
It is in this sense that we say the joint system is non-chiral. We thus expect that, when coupling is introduced, the overall dynamical stability phase diagram is left unaffected. In particular, the presence of a boundary should be sufficient for stabilizing the system as long as the effective loss rates $\kappa$ exceed the amplification rate.

These arguments can be made rigorous by transforming the modes in a way that decouples the two chains. Specifically, consider the hybridized modes $c^\pm_j \equiv a_j\pm ib_j$. We find:
\begin{align*}
    \dot{c}^{\,\pm}_j = J^L c^\pm_{j-1} - J^R c_{j-1}^\pm-(\kappa\pm iw)c_j^\pm .
\end{align*}
Next, upon moving to a frame that rotates at frequency $w$, we can define $d_j^\pm(t) \equiv e^{-iwt}c_j^\pm(t)$ so that
\begin{align*}
    \dot{d}^{\,\pm}_j = J^L d^\pm_{j-1} - J^R d_{j-1}^\pm-\kappa d_j^\pm .
\end{align*}
These rotating hybridized modes individually obey the EOM for a single lossy HN chain. The analysis in the rotating frame then proceeds exactly as in the decoupled limit. In the physical frame, the hybridized modes $c_j^\pm$ undergo identical dynamics to that of the previously described fundamental modes of a single HN chain with an additional phase offset $-w$. Since this offset is spatially isotropic (and hence $k$-independent), it does not affect the normal mode wavepacket velocity, and thus, directional amplification is maintained.

\subsubsection{The chiral regime
$\theta=\pi$}

To contrast with the non-chiral case, let us now consider $\theta=\pi$, with all other parameters kept the same. Now the two chains support asymmetric hopping in opposite directions:
\begin{align}
\label{HNchiral}
J_b^L&= J_a^R \equiv J^R,  \quad J_b^R = J_a^L \equiv J^L.
\end{align}
In this case, the presence of coupling can create instabilities where they were not present before. To see why, consider the case of weak coupling and take the $a$-chain to be in a DM phase with right-moving amplifying modes. Previously, the insertion of boundaries was sufficient for stabilizing the chain. However, with nonzero coupling, amplitude can leak into the $b$-chain. This amplitude will propagate and amplify to the right and continue the cycle. In a sense, ``amplifying vortices'' can form within the bulk. 

From these arguments, we may identify at least one critical value of the inter-chain hopping amplitude $w$, for which we expect a noticeable change in behavior. Specifically, consider the regime $0<w<w_1 \equiv 2\Gamma = J^R-J^L$. This weak coupling regime corresponds to the case where amplitude is transferred between chains at a rate \textit{less} than the characteristic amplification rate of each individual chain. Thus, as long as the onsite gain exceeds the onsite loss, there is time for amplification to take place before eventually transferring amplitude to the other chain and the aforementioned cycle can occur. We generically expect that $0<w<w_1$ corresponds to an unstable phase when $\kappa_+>\kappa_-$, \textit{independent} of BCs, unless $w=0$ precisely, and boundaries are imposed. In particular, we expect the type I DM region in Fig.\,\ref{fig:Chiral HN} to disappear for arbitrarily small $w$. 

For $w>w_1$, stability becomes difficult to predict. Certainly, there must exist a suitably large $w$, say $w\equiv w_2$, such that stability is maintained independently of system-size or BCs. A reasonable estimate following from dimensional analysis would be $w_2 \equiv 2|J|$. From Eq.\,\eqref{wpvel}, this characterizes the maximum wavepacket velocity. For $w>w_2$, wavepackets propagate slower within each chain than they do between the two chains. We may then expect amplification to be ``quenched" by the fast inter-chain transfer. Since the interval $w\in(w_1,w_2)$ corresponds to configurations supporting wavepackets moving both slower and faster than the inter-chain hopping, stability prescriptions are not straightforward.

To explore in more detail, consider the dynamics of the array $\varphi \equiv [a_1,b_1,\ldots,a_N,b_N]$. 
The EOM reads $\dot{\varphi}=\mb A \varphi,$ with $\mathbf{A} = \mathbf{I}_{N}\otimes \mathbf{a}_0 + \mathbf{S}\otimes \mathbf{a}_1+\mathbf{S}^\dag\otimes \mathbf{a}_{-1}$, with $\mathbf{S}$ being the BC-dependent shift operator in Eq.\,\eqref{Sshift}, and
\begin{align*}
  \mathbf{a}_0 \nonumber&=\begin{pmatrix}
        -\kappa& w\\
        -w&-\kappa
   \end{pmatrix},
   \\ \mathbf{a}_1 &\equiv -\begin{pmatrix}
   J^R&0\\
       0&J^L
  \end{pmatrix},\quad \mathbf{a}_{-1}\equiv\begin{pmatrix}
       J^L&0\\
       0&J^R
  \end{pmatrix}.
\end{align*} 
The bulk rapidity bands are calculated in a straightforward manner from the spectrum of $\mathbf{a}(k)\equiv \mathbf{a}_0 + \mathbf{a}_1e^{ik} + \mathbf{a}_{-1} e^{-ik}$; specifically, we have 
\begin{align*}
    \lambda_\pm(k) = -\kappa \pm \sqrt{(2\Gamma\cos(k))^2-w^2}-2iJ\sin(k) .
\end{align*}
Since $|\text{Re}\sqrt{(2\Gamma\cos(k))^2-w^2}|\leq 2\Gamma$, we once again have an unstable bulk whenever onsite gain exceeds the onsite loss. As in the uncoupled case, things are far more interesting when boundaries are imposed. 

\begin{figure}[t]
\centering
\includegraphics[width=\columnwidth]{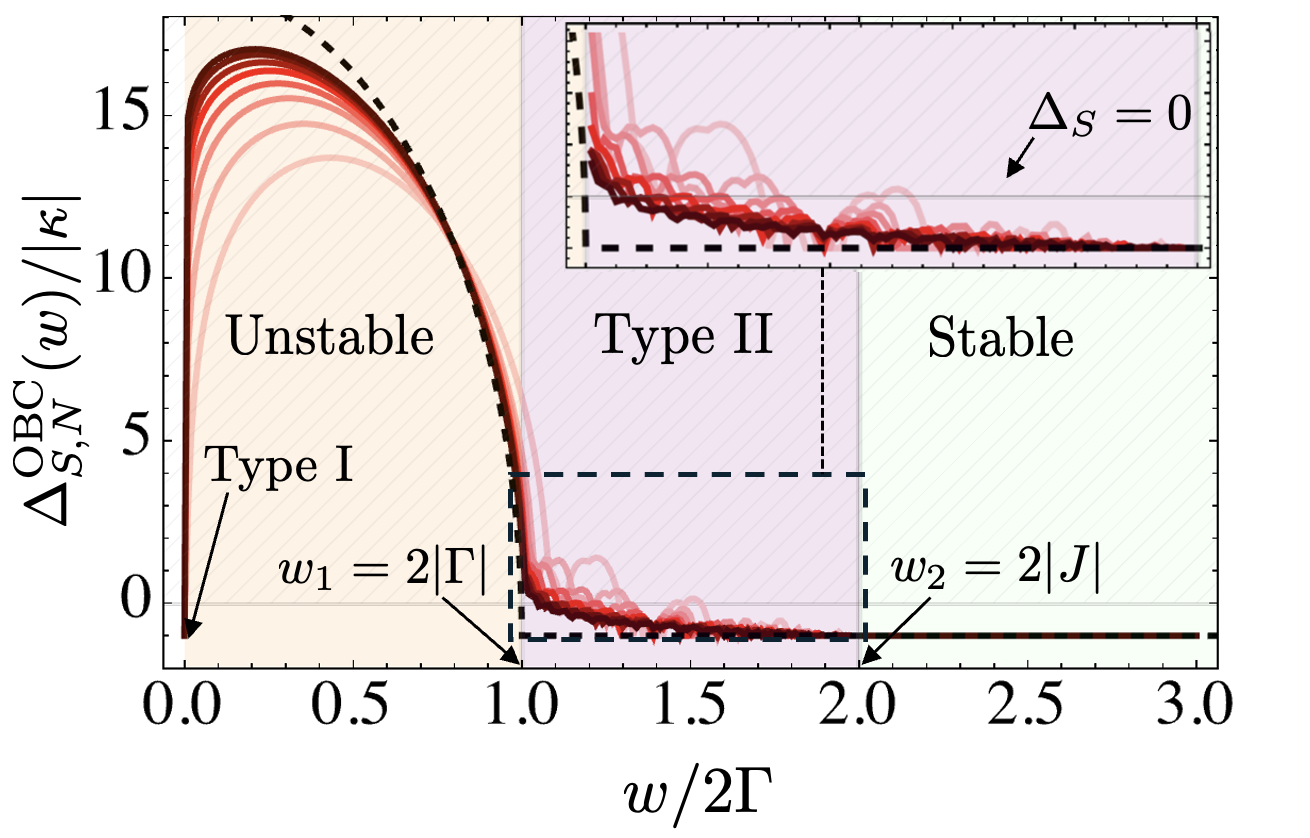}
\vspace*{-4mm}
\caption{Stability gap of the chiral coupled Hatano-Nelson chain model as a function of inter-chain coupling strength $w$ and system size $N$. The hatched region corresponds to the unstable case where $\Delta_{S,N}^{\text{OBC}}(w)>0$ while the orange, purple, and green regions correspond to region I ($0<w<w_1$), II ($w_1\leq w<w_2$), and III ($w\geq w_2)$, respectively. From lightest to darkest, the curves correspond to $N=10,15,20,25,30,35,40,$ and $45$. The dashed black line is the bulk stability gap $\Delta_{S}^{\text{BIBC}}(w)$, which, for this model is equal to the semi-infinite stability gap. The inset shows a zoomed-in section of region II. Relevant parameters are $J=1$, $\Gamma=1/2$, $\kappa_+=1.95$, and $\kappa_-=1$.  }
    \label{fig:size-stabilization}
\end{figure}

In Fig.\,\ref{fig:size-stabilization}, we plot the numerically calculated stability gap $\Delta^{\text{OBC}}_{S,N}(w) = \max \text{Re}[\sigma(\mathbf{A})]$ as a function of $w$ and $N$. The first observation is that, as expected, a small $w$ ($<w_1$) destabilizes the system. Consistent with the generic behavior of highly non-normal matrices, the onset of instability in this region, which we call region I, becomes more and more dramatic as system-size increases. Far more surprising is the behavior of the system in the range $w_1<w<w_2$,
which we call region II. Within this region, $\Delta_{S,N}^\text{OBC}(w)$ jumps erratically above the effective local loss rate $\kappa$. In particular, as long as $\kappa$ is sufficiently small, these erratic jumps can render the system unstable for $w$ in certain subintervals of $[w_1,w_2]$. As demonstrated in Fig.\,\ref{fig:size-stabilization}, the increase in system size generically reduces the stability gap. This is consistent with, but not necessarily implied by, the fact that we may analytically determine stability in the semi-infinite limit for any $\kappa>0$ (see Appendix \ref{app:WH}). Ultimately, we identify the behavior of the model in region II to be that of type II dynamical metastability. Finally, region III, defined by $w>w_2$, is evidently stable independently of BCs or system size, consistent with our general arguments. 

The existence of both type I and type II DM in the chiral coupled HN model is intrinsically tied to the interplay between coherent and dissipative dynamics. The incoherent coupling $\Gamma$ and the coherent inter-chain coupling $w$ jointly play essential roles in producing novel system-size-dependent dynamics. A natural question thus becomes, to what extent is a similarly rich dynamical phase diagram possible in the closed-system setting, when the dynamics are unitary?

\subsection{A Hamiltonian model: The Kitaev-coupled oscillator chain}
\label{interp analysis}

In our next example, we seek to address the above question, by focusing on a purely Hamiltonian setting. For this purpose, let the two real quadrature operators be defined by 
$$ x_j \equiv \tfrac{1}{\sqrt{2}}(a_j^{\dag}+a_j), \quad p_j \equiv \tfrac{i}{\sqrt{2}} (a_j^{\dag}-a_j).$$
We then employ the observation made in Ref.\,\onlinecite{ClerkBKC} that the EOMs for these quadratures under dynamics generated by the so-called bosonic Kitaev chain (BKC) Hamiltonian, 
\begin{align*}
H_{\text{BKC}}&=\frac{1}{2}\sum_{j}(iJa^{\dag}_{j+1}a_j+i\Delta a_{j+1}^{\dag}a_j^{\dag}+\text{H.c.}),
\end{align*}
{\em exactly} mimic those of two decoupled $(w=0)$ HN chains, with opposite directionality ($\theta = \pi$). Here, $J\geq 0$ is a hopping rate, whereas $|\Delta|\leq  J$ is a non-degenerate parametric amplification 
rate. Explicitly, the EOMs for this system in the presence of uniform on-site dissipation rate $\kappa\geq 0$ are:
\[ \left \{ \begin{array}{ll}
    \dot{x}_j &= J_L x_{j-1} - J_R x_{j+1} - \kappa x_j , \\
    \dot{p}_j &= J_R p_{j-1} - J_L p_{j+1} - \kappa p_j, \end{array} \right., \]
where $J_L = (J+\Delta)/2$ and $J_R=(J-\Delta)/2$. Upon the identification $x_j\mapsto a_j$ and $p_j\mapsto b_j$, these EOMs are identical to the $w=0$ limit of the chiral coupled HN model, described by Eq.\,\eqref{CHN}. A notable distinction, however, is that the asymmetric ``hopping rates'' $J_{L,R}$ are now not constrained, and entirely independent of the on-site dissipation rate. This model, in the non-dissipative $\kappa=0$ limit, displays a number of characteristically NH features, including extreme sensitivity to BCs and chiral transport -- much like the chiral HN chain, despite being fully Hermitian. As a result, it has garnered significant theoretical interest \cite{ClerkBKC,ClerkExpEnhanced,Decon}; notably, two experimental realizations of $H_{\text{BKC}}$ in optomechanical and circuit-QED settings have also been recently reported \cite{BKCOptomechanical,BKCSimulation}.

Now, to obtain the analog of the full chiral coupled HN chain, we observe that the role of the coherent coupling $w$ can be fulfilled by harmonic oscillator terms, upon replacing bosonic operators with quadratures. That is, by considering the full Hamiltonian model
\begin{equation}
\label{KOC}
H_{\text{KOC}} \equiv \sum_{j}\Omega\bigg(a_j^\dag a_j + \frac{1}{2}\bigg) + H_{\text{BKC}},\quad \Omega \geq 0,
\end{equation}
we obtain a quadrature incarnation of the coupled HN model in the chiral regime: 
\begin{equation*}
%\label{interp EOMs}
\left\{ \begin{array}{cc} 
    \dot{x}_j &= J_L x_{j-1} - J_R x_{j+1} + \Omega p_j - \kappa x_j,\\
    \dot{p}_j &= J_R p_{j-1} - J_L p_{j+1} - \Omega x_j - \kappa p_j. \end{array}\right.
\end{equation*}
By construction, this model has the {\em same} dynamical stability phase diagram as the coupled HN chain. Furthermore, this new system only has one incoherent contribution, i.e., the onsite dissipation $\kappa$ which, as noted, we are free to set to zero. We are thus left with a purely unitary evolution which, remarkably, retains all of the dynamical metastability properties of the coupled HN chain. Since the Hamiltonian Eq.\,\eqref{KOC} may be interpreted as a chain of oscillators coupled via imaginary hopping and pairing characteristic of the BKC, we dub this model the ``Kitaev-coupled oscillator chain'' (KOC).  

%\textcolor{red}{We also note that the full Hamiltonian Eq.\,\eqref{KOC} may be realized by promoting the BKC NDPA rate to a complex number and periodically driving its phase (see App.\,\ref{app: interpdrive}). Hence, we term the model a parametrically driven bosonic Kitaev chain (PDBKC).}

In analogy with $w=0$ for the HN model, the system at exactly $\Omega=0$ is type I DM \cite{Bosoranas,PostBosoranas}. Specifically, this point is dynamically stable for all finite size, with a purely real, doubly-degenerate spectrum given by \cite{ClerkBKC,Decon}:
\begin{align}
\label{BKCspectrum}
    \omega_m=\sqrt{J^2-\Delta^2}\cos\left(\frac{m\pi}{N+1}\right), \quad m=1,\ldots,N.
\end{align}
As $N\to\infty$, there is a stability gap discontinuity and a spontaneous loss of dynamical stability. More explicitly, the SIBC spectrum is comprised of the unstable bulk spectrum given by two degenerate, but opposite-winding, ellipses with $J$ and $\Delta$ the semi-major and semi-minor axes, respectively, in addition to all points inside the ellipse \cite{Bosoranas}. As we have argued in the general case, such a gap discontinuity {\em cannot} be predicted from the finite-size spectrum alone. Instead, the imprint of the SIBC appears in the finite-size $\epsilon$-pseudospectrum according to Eq. \eqref{wellbehavedps}. This is further illustrated in Fig.\,\ref{fig:pseudospectra}. 

\begin{figure}[t]
\centering
\includegraphics[width=\columnwidth]{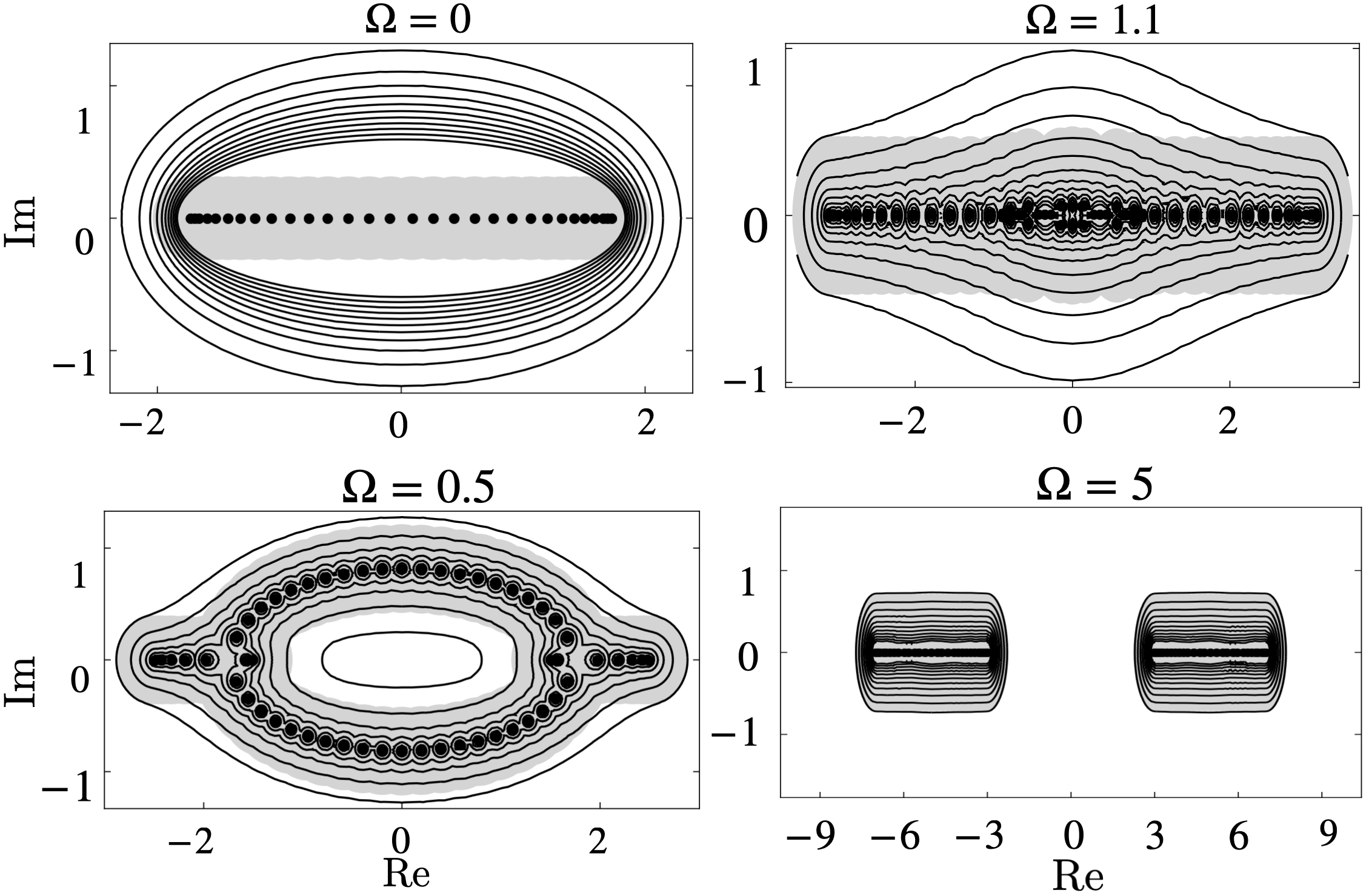}
\caption{The spectra (black dots) and the $\epsilon$-pseudospectra (black contours) of the Kitaev-coupled oscillator chain dynamical matrix ${\mb G}_N(J,\Delta, \Omega)$ under OBCs for $N=30$ and representative choices of $\Omega$ from each region, for $J=2,\Delta=1$. The regions shaded in gray indicate the maximum extent of the pseudospectra for a normal matrix with the {\em same} eigenvalues. The pseudospectral contours for each $\mb G_N$ are chosen as $\epsilon \equiv 10^{-x/\norm{{\mb G}_N}}$, with  
 $x=1,3/2,2,5/2,3,7/2,4,9/2,5,11/2,6.$ The pseudospectra for the dynamical matrices at $\Omega=0$ and $\Omega=1.1$ (region II) are particularly striking as compared to those of their normal counterparts: they extend far beyond the $\epsilon$-neighborhoods around the spectral points, indicating a high degree of non-normality. The pseudospectral contours also highlight a qualitative difference between type II and type I DM: as $\epsilon\rightarrow 0$, the pseudospectrum of $\Omega=0$ remains an ellipse in the complex plane, indicating a discontinuity of the spectrum and the instability of the system in the infinite-size limit, whereas that of region II converges to the BIBC spectrum on the real axis as $N\rightarrow\inf$, signaling the eventual stabilization of the region.}
\label{fig:pseudospectra}
\end{figure}

When $\Omega>0$, the dynamical stability phase diagram of the KOC is equivalent to that of the chiral coupled HN model, thanks to the equivalence of the EOMs. The bulk is unstable for $\Omega\leq  \Delta\equiv \Omega_1$, and stable for $\Omega>\Omega_1$, also matching its dynamical phase diagram under SIBCs (see Appendices \ref{app:WH} and \ref{Gks}). In contrast, the system under OBCs is dynamically unstable for $0<\Omega<J\equiv \Omega_2$ and stable for $\Omega\geq \Omega_2$. In fact, due to equivalence between the linear EOMs, the behavior of $\Delta_{S,N}^{\text{OBC}}$ for this model is identical to the one in Fig.\,\ref{fig:size-stabilization}, upon taking $\kappa\to 0$. Accordingly, from now on, we refer to the regions of the KOC model as region I ($0\leq \Omega\leq \Omega_1$), region II ($\Omega_1<\Omega<\Omega_2$), and region III ($\Omega>\Omega_2$). Physically, it is worth stressing that we may view the Hamiltonian in Eq.\,\eqref{KOC} as combining two dynamically stable Hamiltonians -- the BKC Hamiltonian ($\Omega\to 0$) and a chain of decoupled harmonic oscillators ($J,\Delta\to 0$). Strikingly, however, the path through parameter space between these two stable Hamiltonians is {\em ridden with dynamical instabilities.} 
In fact, this purely unitary model exhibits both type I ($\Omega=0$) and type II ($\Omega_1 < \Omega < \Omega_2$) DM behavior.

The rapid onset of instability when one takes $\Omega$ from zero to a very small non-zero value is consistent with the fact that the BKC under OBCs has been shown to be susceptible to instabilities arising from {\em infinitesimal} perturbations \cite{Decon}. The origin of this extreme sensitivity can be generically explained by extreme non-normality of the generator. However, the fact that the dynamics is unitary grants us an additional interpretation. Specifically, the BKC under OBCs is known to support a \textit{macroscopic} number of KCs in its spectrum \cite{Decon}. In fact, due to oddness under time-reversal symmetry ($\{H_\text{BKC},\mathcal{T}\}=0$), the excitation energy spectrum can be shown to be perfectly chiral: for every excitation of energy $\omega$, there is one with energy $-\omega$. Since the number of these collisions scales with system size, it is natural to expect the degree of destabilization due to perturbations to escalate -- hence explaining the $N$-dependence of region I. Further, it is interesting to note that the decoupled oscillator Hamiltonian has time-reversal symmetry. Thus, this combination of Hamiltonians with even and odd transformation properties under time reversal 
%superposition of even and odd time-reversal systems 
may be a generic way for engendering nontrivial dynamical features.

\begin{figure*}[t]
    \centering 
\includegraphics[width=.9\textwidth]{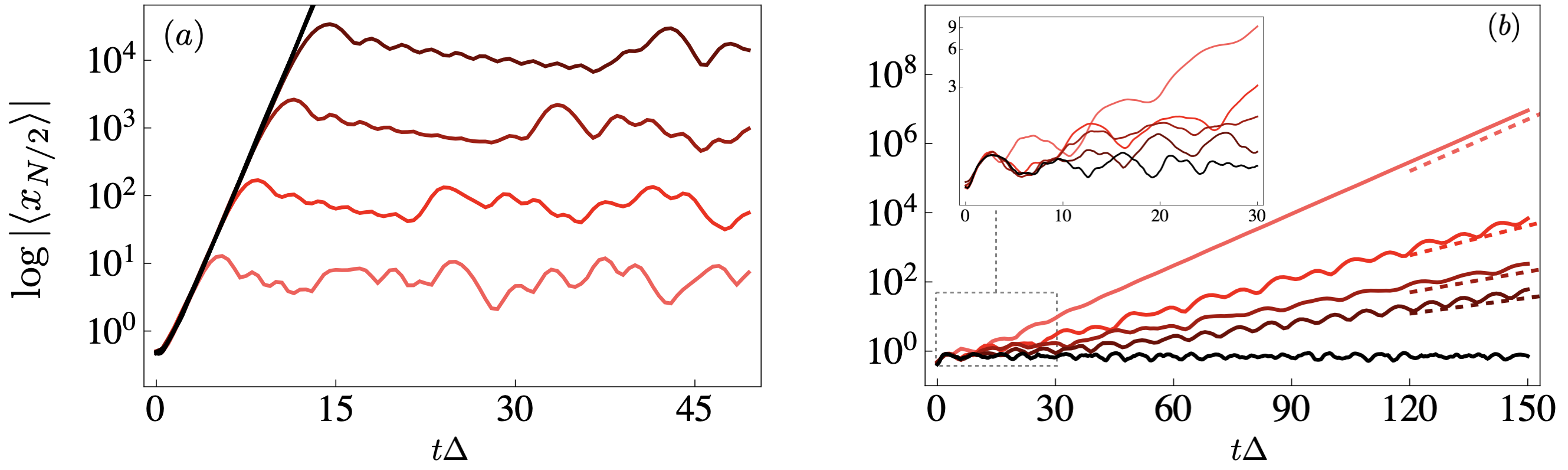}
\vspace*{-2mm}
\caption{Time-dependent expectation value $\langle x_{N/2}(t)\rangle$ for the Kitaev-coupled oscillator chain under OBCs, averaged over 300 trajectories with random initial conditions. From lightest to darkest, $N=16,26,36,46$, with $J=2,\Delta=1$. {\bf (a)} $\Omega=0$. The rate of the initial transient exponential growth is exactly bounded by the corresponding unstable bulk one (black curve). The length of the transient is proportional to system size. The asymptotic limit is stable, with oscillations given by the normal mode frequencies. {\bf (b)} Region II: $\Omega=1.1$. The dashed curves correspond to lines with slope $
\Delta_{S,N}^\text{OBC}$.
The asymptotic limit is unstable, with the exponential growth rate given by $\Delta_{S,N}^\text{OBC}$.
However, as the inset figure shows, in the early time dynamics, the system appears stable for a transient period of time that grows with $N$.}
    \label{fig:Transient}
\end{figure*}

Krein stability theory further offers us an explanation for the onset of type II DM when the KOC is taken
from the oscillator-dominated region III to the intermediate region II. For $\Omega\geq J$, one may check that the bulk energy bands are given by:
\begin{align}
    \omega_\pm(k) = J \sin(k) \pm \sqrt{\Omega^2 - \Delta^2\cos^2(k)}, 
\end{align}
with $\omega_+(k)$ corresponding to the physical excitation energies. The regimes $\Omega > J$ and $\Omega\leq J$ correspond to thermodynamic stability and instability, respectively. In fact, (the thermodynamically unstable) region II is characterized by the emergence of KCs in a momentum interval centered around $-\pi/2$ (where the bulk energy gap closes at $J=\Omega$). That is, for each $k$ in this interval, there is a $k'$ with $\omega_+(k) = \omega_{-}(k')$. 
Under the protection of translation invariance, these Krein-collided modes, which occur in general at different momenta, cannot hybridize and become unstable. The imposition of boundaries, however, serves to do just the job. This `field of KC landmines' offers fertile ground for the sporadic generation of instabilities characteristic of region II. In addition, this region is distinguished by a high level of non-normality, as compared to the non-DM regions the model supports. The extent of non-normality in region II can be appreciated in Fig.\,\ref{fig:pseudospectra}. The pseudospectra in this region extend far beyond the pseudospectra of a normal matrix with an identical spectrum. Finally, in contrast with the chiral coupled HN chain, the type II DM region II in the KOC remains unstable for \textit{all} finite sizes, only to stabilize in the infinite-size limit (see Appendix \ref{app:WH}). In the language of Fig.\,\ref{fig:stability gap comparison}, we have $N_c\to\infty$ for the KOC, whereas $N_c< \infty$ for the chiral coupled HN chains. 

The tension between spectra and pseudospectra for both types of DM manifests directly in the time evolution of physical observables. As an example, we plot $\langle x_{N/2}(t)\rangle$ for a set of random initial conditions in Fig.\,\ref{fig:Transient}. This data confirms that both type I and type II DM regimes engender anomalous transient behavior. Given its overall stability, the system at $\Omega=0$ has an asymptotically bounded evolution, but appears unstable for a transient period of time that grows with system size. Notably, this initial growth period is exactly set by the characteristic amplification rate of the bulk, given by $\Delta$ . An exactly opposite scenario takes place in region II which, while asymptotically unstable, appears stable for a transient period of time that grows with $N$ reflecting a stable infinite-size limit.  

We conclude our discussion of the KOC by calling attention to a model recently proposed in Ref.\,\cite{ClerkInterpolation}. Both our KOC and the model therein have the BKC as a shared point in their larger Hamiltonian parameter spaces but, while our model includes local on-site oscillators, the model in Ref.\,\cite{ClerkInterpolation} adds real hopping to the BKC. That is, 
\begin{align}
\label{ClerkInterp}
   H&= \frac{1}{2}\sum_{j=1}^{N-1}\Big[(g+iJ)a^{\dag}_{j+1}a_j+i\Delta a_{j+1}^{\dag}a_j^{\dag}+\text{H.c.}\Big],
\end{align}
with $J>\Delta>0,$ $g\geq 0$. Like our on-site oscillator terms, the additional real hopping terms serve to quench the nonreciprocal behavior (and thus, type I DM) of the BKC.  In fact, reciprocity is fully restored for $g \geq \Delta$, where the bulk dynamics is stable. It follows that the \textit{bulk} dynamical stability phase diagrams of the two models are equivalent, with $g$ swapped for $\Omega$. However, sharp distinctions arise when one considers OBCs. The Hamiltonian in Eq.\,\eqref{ClerkInterp} is fully stable under OBCs for \textit{all} $g$ and thus, it cannot exhibit type II DM. In our framework, the transition in this model from nonreciprocal to reciprocal is equivalent to one from type I DM to well-behaved (recall Table\,\ref{t: classes }). We have explicitly checked this by using the conditions for the existence of non-trivial partial indices of the Bloch dynamical matrix derived in Appendix\,\ref{app:WH}. Physically, we conjecture that mixing of coupling ranges is necessary to instigate type II dynamical metastability.

\begin{figure*}
\centering
\includegraphics[width=.9\textwidth]{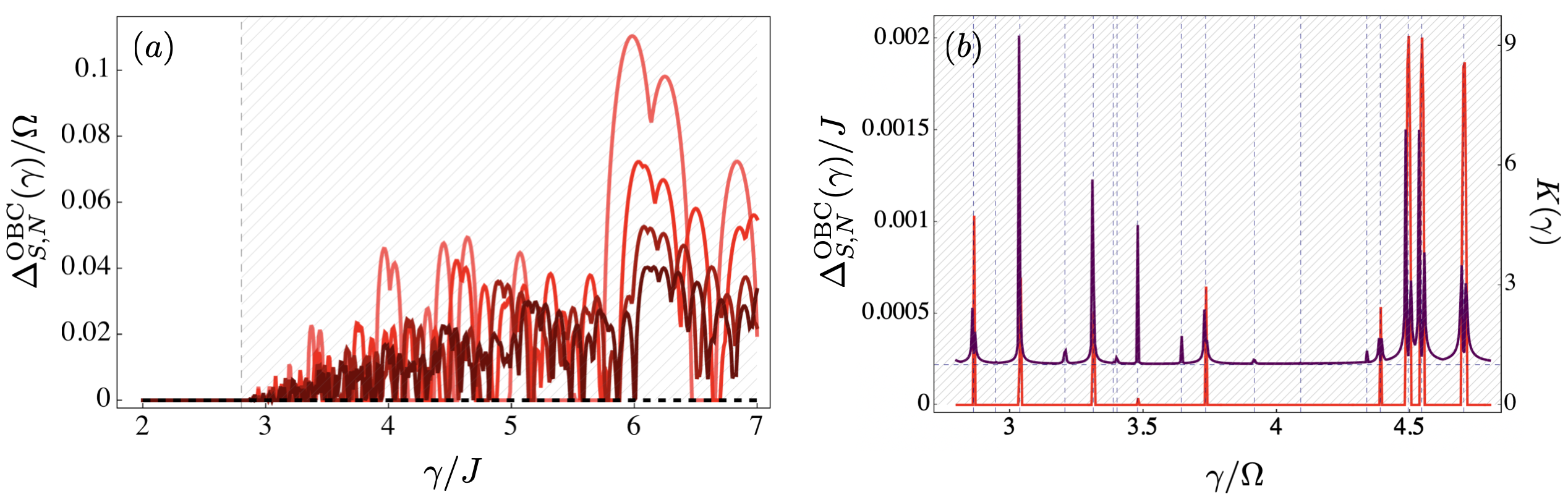}
\vspace*{-3mm}
\caption{Type II DM behavior in the time-reversal-broken harmonic chain under OBCs. {\bf (a)} Stability gap $\Delta_{S,N}^{\text{OBC}}(\gamma)$ for $\Omega=3,J=1.$ Type II DM sets in for $\gamma>\gamma_c$, marked by the grid-line, at the point where thermodynamic stability is lost. From lightest to darkest, $N=16,26,36,46,$ with the black, dashed line denoting the semi-infinite limit. {\bf (b):} 
Peaks of the condition number $K(\gamma)$ (purple), $\Delta_{S,N}^{\text{OBC}}$ (red), and location of KCs for $J=0$ (vertical dashed blue grid-lines).
Spikes above the horizontal dashed grid-line, corresponding to $K(\gamma)=1$, indicate regions of high non-normality. 
Here $N=20$, $\Omega=1,J=0.01$.}
\vspace*{-2mm}
\label{fig:GHC with imaginary hopping}
\end{figure*}

\subsection{A design principle put to the test:\\ The time-reversal-broken harmonic chain}

We argued that the type II DM behavior taking place in region II of the KOC phase diagram may be a result of imposing boundaries on a bulk system that hosts a macroscopic number of KCs. It 
is natural to expect that any QBH supporting bulk KCs in its energy bands may exhibit type II DM. Explicitly, we design a QBH which can be parametrically tuned to close the bulk many-body energy gap, and thus generate a macroscopic number of KCs -- while retaining bulk dynamical stability. Importantly, however, the system must additionally possess a {\em destabilizing mechanism} (say, non-degenerate parametric amplification).

Consider a QBH of the form $H \equiv H_\text{GHC} + H_\text{TRB}$, with 
\begin{align}
\label{TRSB}
    H_\text{GHC} \nonumber&= \!\sum_{j}\Omega\Big( a_j^\dag a_j + \frac{1}{2}\Big)
    -\! \frac{J}{2}\Big( a_{j+1}^\dag a_j + a_{j+1}^\dag a_j^\dag + \text{H.c.}\Big),
     \\
     H_\text{TRB} &= \frac{i\gamma}{2}\sum_{j} \big(a_{j+1}^\dag a_j - \text{H.c.}\big),
    \end{align}
where $\Omega$ is an on-site frequency, $J$ a real hopping and pairing amplitude, and $\gamma$ an imaginary hopping amplitude. The Hamiltonian $H_\text{GHC}$ can be shown, in an appropriate quadrature basis, to be equivalent to a family of oscillators, harmonically coupled in position to their nearest neighbors (the so-called ``gapped harmonic chain'' of Ref.\,\cite{Dualities}). $H_\text{GHC}$ is manifestly time-reversal invariant and is, in fact, dynamically stable whenever $\Omega>2J$, independently of BCs (hence, it is not DM). The second Hamiltonian contribution, $H_\text{TRB}$, serves to explicitly break time-reversal symmetry and, since it may be seen as the $\Delta\to 0$ limit of the BKC, it features a completely chiral excitation energy spectrum. Thus, making the imaginary hopping strength $\gamma$ sufficiently large should serve to both drive a many-body gap closing (and consequently, a thermodynamical stability phase transition) of the GHC and maintain bulk stability.

The bulk excitation energies are given by:
\begin{align*}
    \omega(k) = \gamma \sin(k) + \sqrt{\Omega(\Omega-2J\cos(k))}.
\end{align*}
From this, we find that as the imaginary hopping is increased, there is a point $\gamma=\gamma_c$, such that the bulk energy gap closes:
\begin{align*}
\gamma_c &= \tfrac{\Omega}{\sqrt{2}} \sqrt{1-\sqrt{1-\left(\tfrac{2J}{\Omega}\right)^2}}.
\end{align*}
Beyond that, for $\gamma >\gamma_c$, the system is thermodynamically unstable, but still dynamically stable (see Appendix \ref{Gks}). Just as in region II of the KOC, a macroscopic number of KCs develop in the bulk; thus, the stage is set for type II DM to emerge, upon introducing boundaries. This prediction is verified in Fig.\,\ref{fig:GHC with imaginary hopping}(a).  As we have seen in the other two examples of type II DM, the behavior of the stability gap 
$\Delta_{S,N}^{\text{OBC}}(\gamma)$ is erratic and unpredictable for $\gamma>\gamma_c$. However, for a fixed $N$, and small $J$ (thus, small amplification), we find that the peaks of 
$\Delta_{S,N}^{\text{OBC}}(\gamma)$ are located {\em precisely} where KCs occur for $J=0$, as observed in Fig.\,\ref{fig:GHC with imaginary hopping}(b). Additionally, we plot the \textit{condition number} of the modal matrix of $\mathbf{G}_N^\text{OBC}$ as a measure of non-normality \cite{TrefethenPS}: explicitly, $K \equiv \norm{\mathbf{L}_N} \norm{\mathbf{L}_N^{-1}}$, where the columns of $\mathbf{L}_N$ are the eigenvectors of $\mathbf{G}_N^\text{OBC}$. Since the condition number is always $1$ for normal matrices, spikes above $1$ represent points where non-normality is locally maximal: as Fig.\,\ref{fig:GHC with imaginary hopping}(b) shows, these again line up precisely with the KCs and the spikes in $\Delta_{S,N}^{\text{OBC}}(\gamma)$.  

Based on the model systems we have analyzed, we observe that, while the condition number of type II DM systems does not diverge akin to that of type I, the {\em onset} of type II DM {\em is} signaled by high non-normality and erratic behavior of the condition number as a function of relevant system parameters.

\section{Further Implications} 
\label{sec:further}

Our work thus far has shed light onto several distinctive 
features -- many of them {\em a priori} unexpected -- of the relationship between the dynamical stability phases of finite-size quadratic bosonic systems and those of their infinite-size counterparts, as elicited by non-normality of the generator. Unsurprisingly, a number of physically relevant quantities of a 
quadratic bosonic system are, in turn, explicitly tied to its dynamical stability phase. In what follows, we showcase two particularly compelling manifestations.

\subsection{Anomalous transient in entanglement generation}
\label{EE}

\begin{figure*}[t]
 \centering
\includegraphics[width=\textwidth]{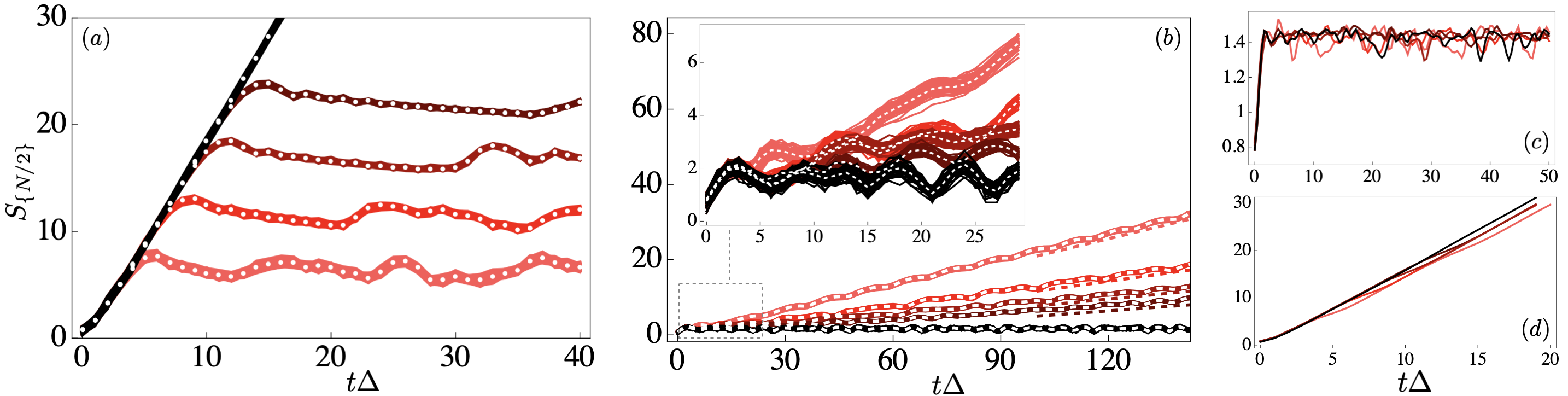}
\caption{{\bf (a)} Entanglement entropy for the Kitaev-coupled oscillator chain under OBCs at $\Omega=0$, initialized in 50 randomly sampled Gaussian states (see Appendix\,\ref{app: Gaussian states}). From lightest to darkest, $N=16,26,36,46$. White-dashed lines correspond to the average of 50 paths. The subsystem is fixed as the middle site for all $N$, where the EE growth is maximized. 
%% LV: Does it matter whether N is even vs odd? 
%% MU: For odd N its maximized for both $(N\pm 1)/2$, so I have restricted to even N for convenience 
As $N$ grows, the transient period grows proportionally, with the rate bounded by that of the most unstable bulk  mode (black curve). Here, $J=2,\Delta=1$. {\bf (b)} Same parameters as in {\bf (a)}, except $\Omega=1.1$ (region II). Asymptotically, the EE grows linearly with the rate set by $2\Delta_{S,N}^{\text{OBC}}$ (dashed lines);
however, the larger the system, the longer it appears to be stable, as shown by the inset for early-time EE growth. {\bf (c)} Same as {\bf (a)}, but only the averages over paths are plotted for the BKC with real hopping in the reciprocal ($J=2, \Delta=1, g=1.2$), non-DM regime.
This illustrates a generic lack of dependence of the initial EE growth on $N$, when the subsystem size is fixed in a dynamically stable, but not DM system. {\bf (d)} Same as {\bf (b)}, but only averages are plotted for $\Omega=0.5$ (region I). There is no transient period of stability in a purely unstable regime.}
\label{fig:EE}
\end{figure*}

Entanglement properties have long provided an important tool for probing the nature and implications of quantum correlations across different quantum many-body phenomena, ranging from equilibrium quantum phase transitions \cite{Fazio,Rolando} to topological order \cite{Balents,Wen}. In particular, for systems described by a (time-independent) QBH $H$, it has been established that, if one initializes a bipartite bosonic system in a Gaussian state, the \textit{asymptotic} production of entanglement entropy (EE) depends on the dynamical stability phase of $H$ \cite{Rigol2018,EElinearKS,HacklLinearEE}. Specifically, for a dynamically stable QBH, the asymptotic EE is bounded and quasi-periodic. For a dynamically unstable QBH with only polynomial instabilities (due to EPs alone, not non-real eigenvalues), the asymptotic EE grows logarithmically in time. Finally, if $H$ displays exponential instabilities, then the asymptotic EE grows linearly with time. These results suggest that dynamically metastable systems may also be peculiar from the point of view of the dynamics of entanglement production. 

As we have seen, the transient behavior of a dynamically metastable system is consistent with a stability phase that is distinct from the true finite-size one. We may then generally predict that a dynamically metastable system will show \textit{transient EE production} that is qualitatively distinct from the asymptotic behavior. Consider a type I (or, respectively, type II) DM QBH of size $N$ and let $\tau_N^\text{DM}$ denote the timescale in which the transient amplification or suppression, respectively, of bulk pseudomodes dies out (e.g., the time at which the quasi-periodic evolution of $\braket{x_{N/2}(t)}$ sets in in Fig.\,\ref{fig:Transient}). Given a bipartition of the system, suppose in addition that $\tau_N^\text{EE}$ sets the timescale for the EE to enter the asymptotic regime. It follows that, if $\tau_N^\text{DM}>\tau_N^\text{EE}$, then the EE will first set into a linearly-growing (bounded)  evolution for $t\in[\tau_N^\text{EE},\tau_N^\text{DM}]$, until eventually becoming bounded (linearly growing) for $t> \tau_N^\text{DM}$. Moreover, since we know $\tau_N^\text{DM}\to\infty$ as $N\to\infty$, this anomalous transient EE behavior will become more and more pronounced as $N$ grows.

To put this picture to the test, we must compute the EE of a Gaussian state for a fixed subsystem $A$ of a larger bosonic system. Let 
${\cal R} \equiv [x_1,p_1,\ldots,x_N,p_N]^T$ denote the array of quadratures of the full system. A Gaussian state may be specified uniquely by its mean-vector and covariance matrix (CM) which, in the quadrature basis, read (see also Appendix\,\ref{app: Gaussian states})
\begin{align}
    \vec{m} = \braket{\cal R},\quad 
    \mathbf{\Gamma}_{ij} &= \braket{\{{\cal R}_i - \braket{{\cal R}_i},{\cal R}_j - \braket{{\cal R}_j}\}},
    \label{covM}
\end{align}
respectively. 
Thus, the EE of a subsystem $A$ of size $N_A$, in an initial Gaussian state with a CM $\mb \Gamma$, is computed as \cite{AdessoEE,Rigol2018}:
\begin{align*}
    S_A&\nonumber=\sum\limits_{i=1}^{N_A} S(\nu_i), \\
    S(v)&=\frac{\nu+1}{2}\log\frac{\nu+1}{2}-\frac{\nu-1}{2}\log\frac{\nu-1}{2},
\end{align*}
with $\nu_i$ denoting the positive eigenvalues of $[i \boldsymbol{\bm{\Omega}} \mb \Gamma]_A$ (i.e., the \textit{symplectic eigenvalues} of $[\bm{\Gamma}]_A$). Here,  $[\bullet ]_A$ denotes a projection of the matrix to a sub-matrix that has support only on $A$ and $\boldsymbol{\Omega}\equiv i\boldsymbol{\tau}_2$ is the \textit{symplectic form}, which encodes the CCRs in the quadrature basis. A Gaussian initial state will remain Gaussian under the time-evolution generated by QBHs and even QBLs. Thus, $S_A(t) = \sum\limits_{i=1}^{N_A} S(\nu_i(t))$, where $\nu_i(t)$ are the instantaneous positive eigenvalues of $[i \boldsymbol{\bm{\Omega}} \mb \Gamma(t)]_A$.

For concreteness, we investigate the KOC initialized in randomly sampled 
Gaussian states, with the single site in the middle of the chain being the subsystem $A$. For a type-I DM regime, we find an anomalous, system-size-dependent period of linear EE growth that is absent in a non-DM and dynamically stable system, compare Figs.\,\ref{fig:EE}(a) and \ref{fig:EE}(c). This anomalous regime of linear EE growth is consistent with the results of Ref.\,\cite{Rigol2018}; it reflects the fact that the 
system is, in the transient, effectively unstable, with the maximal rate of the EE amplification set by the most unstable mode of the SIBC, 
system or, equivalently, for our model, the largest amplifying bulk mode. But, since the system is ultimately dynamically stable, the EE must necessarily become bounded. In the type II DM regime, in contrast with a purely unstable one, the system initially appears stable and the EE is bounded and oscillatory for an increasingly long period of time as the system size is increased. Eventually, instability sets in and the EE settles into linear growth. 

For the corresponding parameter regime, the length of the transient for every $N$ corresponds to the length of the transient dynamics of $\langle x_{N/2}\rangle$ in Fig.\,\ref{fig:Transient}. In addition, notice that Figs.\,\ref{fig:Transient} and \ref{fig:EE} look remarkably similar overall. We conjecture that under Hamiltonian evolution, there will always be a choice of subsystem $A$ such that $S_A(t)$ mimics the average trajectories of the corresponding subsystem quadratures. This is reasonable as the generation of EE is closely related to the expansion of volume in the symplectic subspace of the overall classical phase space spanned by these (now classical) quadratures. 

Interestingly, an unexpected EE phase transition was recently reported for a family of models related to the BKC \cite{ClerkInterpolation}. For the BKC with real hopping in Eq.\,\eqref{ClerkInterp}, the asymptotic EE of a subsystem of size $N_A \sim \mathcal{O}(N)$ was found to obey a volume law scaling in the reciprocal regime, whereas a {\em super-volume law} was found in the non-reciprocal case, namely, the asymptotic EE grows as $\sim N^2$. Remarkably, not only is our observation regarding transient EE production in non-DM vs.\,DM systems consistent with this transition behavior, but in fact {\em it provides a physical explanation}. In Eq.\,\eqref{ClerkInterp}, the transition from the reciprocal to the nonreciprocal regime exactly parallels the transition from a non-DM to a type I DM regime and, therefore, from a regime of size-independent transient EE growth to a size-dependent one. This, in turn, allows for the accumulation of uncharacteristically large asymptotic EE in the nonreciprocal regime. 

Examining EE production through the lens of transient dynamics could also shed light on why a system of two explicitly NH, fermionic coupled HN chains, spectrally equivalent to the BKC with real hopping, was {\em not} found to have a super-volume EE scaling in the nonreciprocal regime instead \cite{RyuEEPhaseTrans}. In type I DM systems, we attribute the size-dependent, linear scaling of the transient EE to the dynamical instability of the SIBC system. This mechanism, however, is inaccessible to quadratic fermionic systems, as they are {\em always} dynamically stable -- even when one considers an explicitly NH Hamiltonian, since one must renormalize the time-evolved state for it to remain consistent with fermionic statistics \cite{Prosen3QBoson,ProsenSpectral, BarthelQuadLindblad}. Hence, the source for enhanced EE production is absent. 

\begin{figure*}[t]
\centering
\includegraphics[width=.94\textwidth]{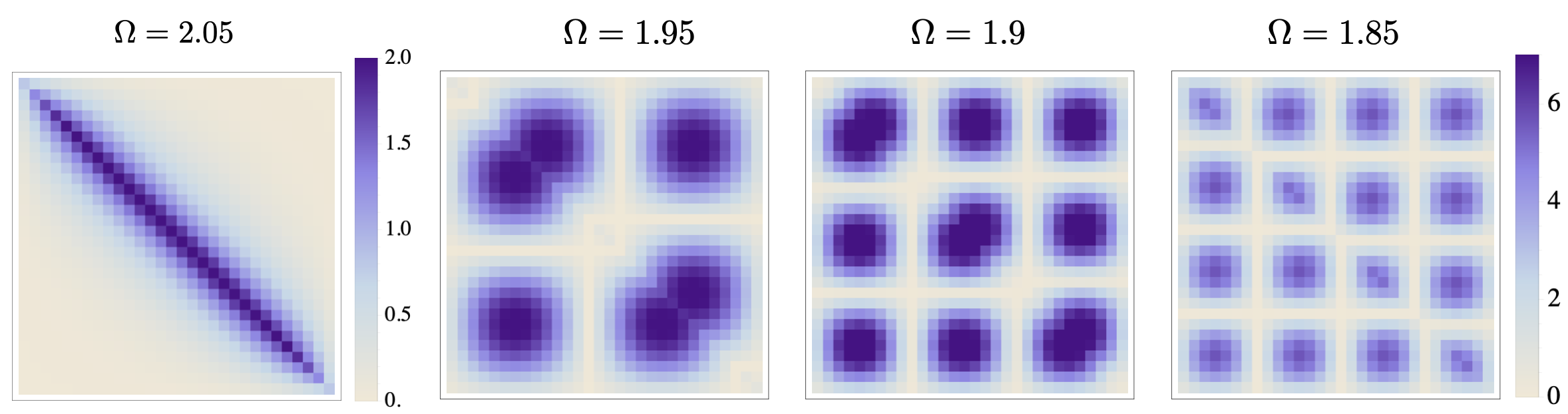}
\caption{Zero-frequency response of the Kitaev-coupled oscillator chain, for $N=30$ and representative values of $\Omega$ from regions II and III, in the {\em Hermitian} (normal) limit, with $J=2$ and $\Delta=0$.
Specifically, for each $\Omega$, we focus on the sub-matrix of $\boldsymbol{\chi}(\omega=0)$ with elements $|\boldsymbol{\chi}_{2\ell -1, 2m-1}(0)|$, $\ell,m\in \{1,\hdots,N\}$; the magnitude of the resulting response is indicated by the intensity of the color. For $\Omega>J$, the drive frequency $\omega$ is outside the SIBC spectrum, the response is short-ranged, and relatively constant against small changes in $\Omega$. For $\Delta<\Omega<J$, $\omega$ is in the SIBC, but not the finite spectrum. The response to the linear drive reveals the spatial structure of the pseudo-normal modes associated with the pseudo-eigenvalue equal to the driving frequency.  The response is highly sensitive to small changes in $\Omega$.}
\label{fig: sensitive response}
\end{figure*}

\subsection{The essential role of pseudospectra in linear response}

In this section, we bring to the fore the significance of the pseudospectrum and, by extension, the infinite-size limit, for the linear response behavior of finite-size QBLs. Specifically, given a system described by a QBL generator, we consider its response to a coherent, weak linear drive of the form: 
\begin{align*}
H_{\text{drive}}(t)&=i\sum_j \big(z_j^{*}(t)a_j-z_j(t)a_j^{\dag} \big). 
\end{align*}
The linear (Kubo) response of $\braket{\Phi}$ to such a drive is ultimately encoded in the \textit{response (susceptibility) matrix}, or the \textit{Green's function}, which turns out to be directly related to the unperturbed system's dynamical matrix \cite{ZanardiResponse}:
\begin{equation}
\label{chiomega}
\boldsymbol{\chi}(\omega)=i(\omega \mathds{1}_{2N} -\mb G)^{-1},
\end{equation}
and can be obtained by means of Fourier transforming (and right multiplying by $\bm{\tau}_3$) an array of two-time correlation functions with $\ell m$-th element proportional to: 
\begin{align*}
     \braket{[\Phi_{\ell}(t),\Phi_m^\dag(\tau)]} = \left[e^{-i\mathbf{G}(t-\tau)}\bm{\tau}_3\right]_{\ell, m}.
\end{align*}
This quantity is manifestly state-independent and a function only of the time delay $t'=t-\tau$ (see also Appendix \ref{app:LR}). In the case of a single bosonic species, we have, for instance, 
\begin{align*}
    \boldsymbol{\chi}_{2\ell-1,2m-1}(\omega) = -i\int_{-\infty}^\infty e^{i\omega t'}\Theta(t')\braket{[a_{\ell}(t'),a_m^\dag(0)]}\,dt',
\end{align*}
with $\Theta$ being the Heaviside step function. Thus, mathematically, Eq.\,\eqref{chiomega} asserts that the frequency-space response function is proportional to the {\em resolvent} of $\mb G$, given by $\mathbf{R}(z) \equiv (\mathbf{G}-z \mathds{1}_{2N})^{-1}$, evaluated at a point $z=\omega\in\mathbb{R}$. Physically, $\boldsymbol{\chi}(\omega)$ is a block-matrix whose $\ell m$'th block sets the strength of the response at site $\ell$ to a drive at site $m$ of frequency $\omega$. As such, the susceptibility matrix plays a central role in the input-output theory of multi-mode bosonic systems \cite{PorrasIO,NunnenkampTopoAmp}. For example, the scaling of $\norm{\boldsymbol{\chi}(\omega=0)}$ with system size in type I DM systems, in combination with topological metastability, can be used to identify distinctive zero-frequency signatures of topological bosonic edge modes \cite{Bosoranas,PostBosoranas}. 

\begin{figure*}[t]
\centering
\includegraphics[width=.94\textwidth]{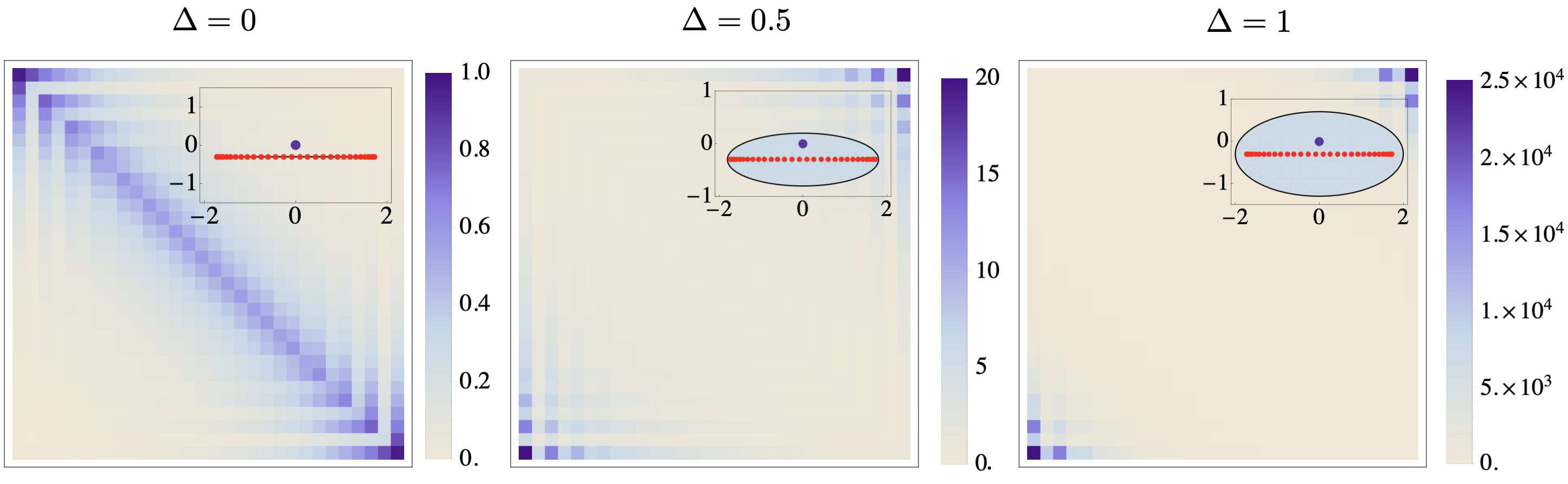}
\caption{Frequency response $|\boldsymbol{\chi}_{2\ell -1, 2m-1}(\omega)|$, $\ell,m\in \{1,\hdots,N\}$, for the bosonic Kitaev chain with onsite dissipation $\kappa$, and increasing level of non-normality, tuned by the magnitude of $\Delta$. Here, $N=30$, $\omega=0$, $\kappa=0.3$, and $J$, $\Delta$ are varied in such a way that $J^2-\Delta^2=3$. While the finite-size spectra are purely real and, by construction, identical in each case, the corresponding complex, SIBC spectra differ (see Sec.\,\ref{interp analysis}). Insets: Probing frequency (purple) relative to the SIBC spectra (black curves and blue interiors) and the finite-size spectra (red). For $\Delta=0$, the probing frequency is outside both the finite- and infinite-size spectra. This is also reflected in the diagonally dominant, off-resonant response function. For $\Delta=0.5$, the probing frequency is outside the range of the finite-size spectrum, but falls near the edge of the SIBC spectrum. Consequently, the response is preferentially end-to-end amplifying. For $\Delta=1$, the probing frequency is again outside the range of finite-size, but deep within the SIBC spectrum. We see this reflected in the very sharp end-to-end amplification, characteristic of nonreciprocal amplifiers. The response changes dramatically upon changes in the pseudospectrum, even when the spectrum remains unchanged.}
\label{fig: nonreciprocal response}
\end{figure*}

Our goal here is to investigate the spatial structure of the response function and how it changes in the different parameter regimes of our models. Whenever $|\omega|>\norm{\mb G}$, with the matrix norm arbitrary, the response is straightforward. In such a case, $\bm{\chi}(\omega)$ can be expanded in a Neumann series in powers of $\mathbf{G}/\omega$ and, given the block-Toeplitz structure of $\mathbf{G}$, one finds that the $n$-th diagonal of $\bm{\chi}(\omega)$ is suppressed on the order of $1/|\omega|^n$. As a result, the response is strongest at the sites of the applied perturbation, with the elements of $\boldsymbol{\chi}(\omega)$ closest to the diagonal being the largest in magnitude and the ones away from the diagonal decaying with distance. More interesting response behaviors emerge when $|\omega|\leq \norm{\mb G}$. This includes (but is not limited to) frequencies embedded within the range of normal-mode frequencies of the system. 

Given the prominent role the resolvent of the dynamical matrix plays in the formulation of linear response theory, pseudospectral theory proves to be {\em essential} for explaining the general structure of $\bm{\chi}(\omega)$. If $\omega$ is in the $\epsilon$-pseudospectrum, then the overall
 response strength characterized by $\norm{\bm{\chi}(\omega)}$ will be at least $1/\epsilon$. If $\mb G$ is nearly normal, this leads to the physically intuitive fact that driving a system near resonance will elicit a large response. However, what is more exotic, and partly anticipated in Ref.\,\cite{SatoOkumaPseudo}, is the abnormally large response of an extremely non-normal system to a highly off-resonant drive: A system is sensitive to a drive at any frequency that belongs to its pseudospectrum. We thus arrive at a remarkable dynamical principle for bulk-translationally-invariant systems: A finite-size QBL exhibits the strongest response when driven at {\em resonant frequencies of the corresponding infinite-size system}, even when these sharply differ from the resonance frequencies of the physical system itself.

Going further, the spatial structure of the response function for a finite system can also be deduced from pseudospectral theory. For $\omega$ in the $\epsilon$-pseudospectrum, with $\epsilon$ suitably small, there will be corresponding a $\epsilon$-pseudo-eigenvector. One way to construct such a pseudo-eigenvector is to take the corresponding exact eigenvector of the infinite-system (which may even be the non-normalizable plane-wave bulk eigenstates) and truncate it to a finite-size. The spatial character of $\bm{\chi}_{\ell m}(\omega)$ will most closely resemble that of the truncated normal mode. On the one hand, if the said mode is a delocalized plane-wave bulk mode with wave-vector $k$, one expects variations of $\bm{\chi}_{\ell m}$ on the order of $|\ell - m|\sim 2\pi/ k$. On the other hand, if extreme non-normality persists and the pseudo-mode is born out of a skin mode (thus, exponentially localized on the boundary), one expects $\bm{\chi}_{\ell m}(\omega)$ to be most pronounced at the corners, due to exponential amplification.

\begin{figure}[t]
\centering
\includegraphics[width=\columnwidth]{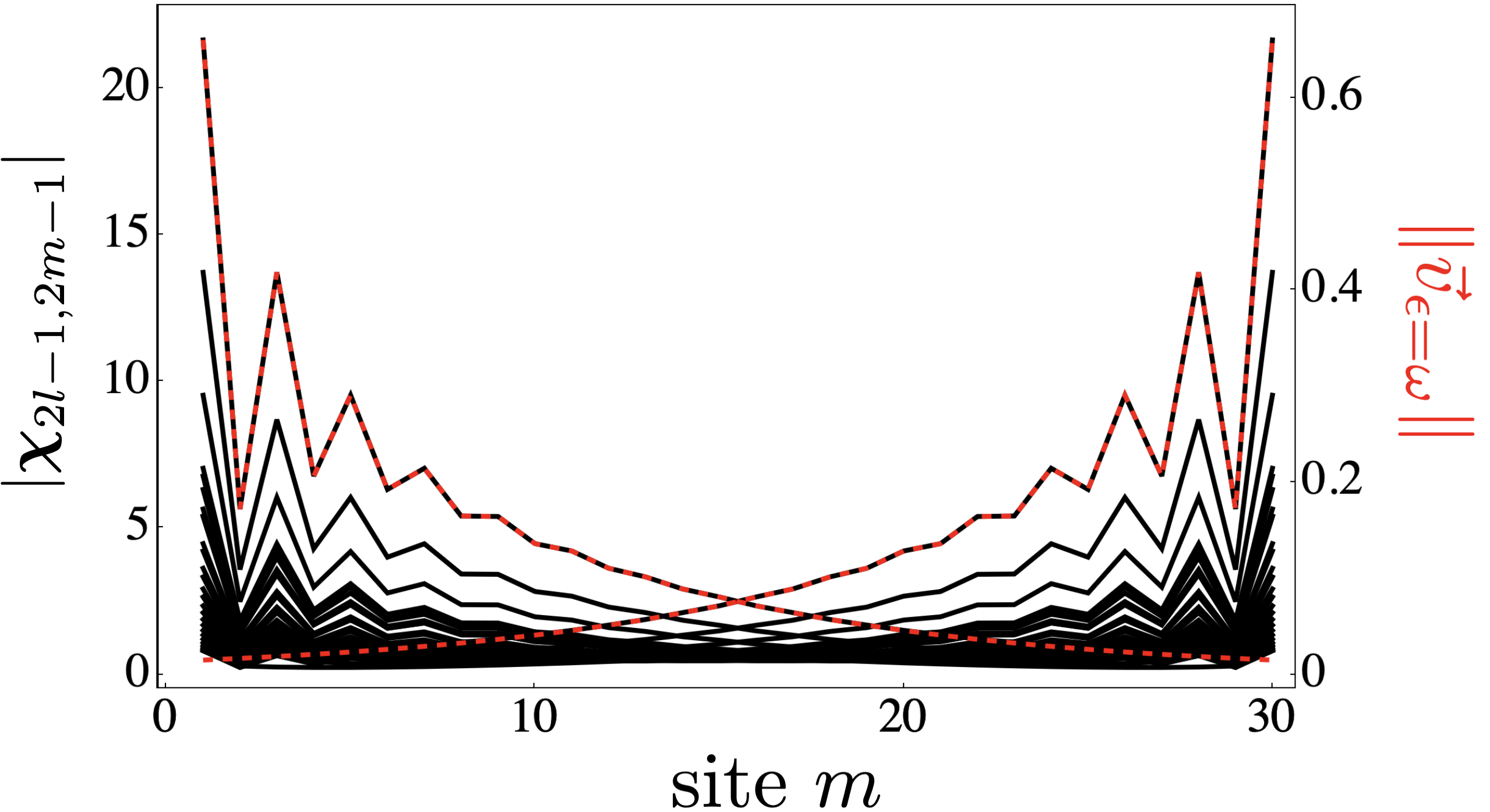}
\vspace*{-3mm}
\caption{Response and pseudospectral modes of the Kitaev-coupled oscillator chain with onsite dissipation $\kappa$, as in Fig.\,\ref{fig: nonreciprocal response}. Here, $J^2-\Delta^2=3,\Delta=0.5,\omega=0,$ and  $\kappa=0.3$. Each black curve corresponds to a row, indexed by $\ell$, of 
$\boldsymbol{\chi}_{2\ell -1, 2m-1}(\omega)$, $\ell,m\in\{1,\hdots,N\}$. The dashed red curves are the $\epsilon$-pseudoeigenvectors, with $\epsilon=\omega$ (scale on the right). The spatial structure of $\boldsymbol{\chi}(\omega)$ reproduces {\em exactly} the structure of the pseudospectral modes.}
\label{fig: pseudospectral mode}
\end{figure}

Let us now see how these general arguments play out in our models. Consider first the zero-frequency response behavior of the KOC in the  normal, $\Delta=0$ limit. While $\omega=0$ is {\em not} in the exact spectrum of the system under OBCs for any $N$, we have that 
$0\in \sigma(\mb G_{\text{SIBC}})$ in the region $0<\Omega<J$. Therefore, in contrast with $\Omega>J$, $\omega=0$ is in the $\epsilon$-pseudospectrum, of every finite system, with $\epsilon$ becoming smaller with $N$. Thus, DC ($\omega=0$) driving elicits a ``pseudoresonant'' response and the susceptibility matrix reveals the spatial structure of the zero-energy pseudo-normal modes. This can be observed in the sharp transformation of the susceptibility matrix between regions $\Omega>J$ and $\Omega<J$, see Fig.\,\ref{fig: sensitive response}. The response in the former region is an off-resonant one, strongest along the diagonals and sharply falling off with distance. In contrast, in the latter region, the response changes drastically with small variations in $\Omega$, reflecting the spatial profiles of the resonant pseudo-modes -- which happen to be delocalized plane waves, with wave-vectors set by $\omega_+(k)=0$ at $\omega=0$. 
Necessarily, this pseudoresonant response will only appear at $\omega=0$ if the system is gapless in the infinite-size limit. For a QBH, this includes systems that have lost, or are on the cusp of losing, thermodynamic stability. Hence, the pseudospectrum can explain finite-size effects even for normal and, in particular, Hermitian systems \footnote{The pseudoresonant behavior we describe may be seen for arbitrary Hermitian dynamical matrices with driving frequencies in the neighborhood of its eigenvalues. However, the Toeplitz structure allows us to interpret the pseudoresonant behavior as the imprint of the infinite-size limit on the finite system.}.

The linear response of nonreciprocal, directionally amplifying systems shows distinctive, unidirectional behavior in the form of exponentially enhanced, end-to-end gain in signal strength \cite{PorrasTopoAmp, NunnenkampTopoAmp, ClerkExpEnhanced, NunnenkampResoration}. Consider for example the BKC under OBCs. This system can display an end-to-end amplifying response when probed with frequencies far outside the convex hull of its finite-size spectrum. In practice, since $\omega \mathds{1}_{2N}-\mathbf{G} = (\omega -i \kappa) \mathds{1}_{2N} -\mathbf{G}'$, with $\mathbf{G}'=\mathbf{G}-i\kappa\mathds{1}_{2N}$ describing the same system as $\mathbf{G}$ with additional uniform dissipation $\kappa$, we can effectively sample the resolvent of the BKC Hamiltonian at complex values.
The large end-to-end response of the BKC can be observed in Fig.\,\ref{fig: nonreciprocal response}. To highlight the importance of the
pseudospectrum, 
%in determining response behavior, 
we vary the parameters in such a way that the spectrum remains fixed while the pseudospectrum changes significantly. As $\Delta$ is increased, with $J^2-\Delta^2$ kept constant, the response changes dramatically from diagonally dominant to directionally amplifying. The reason is that, as $\Delta$ changes away from zero, the spectrum of the infinite-size system separates from its Hermitian counterpart: It expands from the real line to encompass a solid ellipse in the complex plane, such that the driving frequency $\omega$ becomes pseudoresonant with some approximate eigenmodes of the finite system (see also Fig.\,\ref{fig: pseudospectral mode}) or, equivalently, some exact eigenmodes of the infinite-size system. Similarly, the chiral HN chain in the decoupled limit displays an analogous unidirectional response as the BKC for all pseudoresonant $\omega$. These examples convincingly
demonstrate the insufficiency of the spectrum for characterizing response behavior, while simultaneously bolstering the pseudospectral approach.

In the context of input-output theory, the correspondence between directional amplification and topologically non-trivial band structure has been mathematically explained  as a consequence of the emergence of  exponentially localized singular eigenvectors (corresponding to exponentially vanishing, in system size, singular eigenvalues) \cite{NunnenkampTopoAmp,NunnenkampResoration,PorrasTopoAmp}. It is worth noting that the procedure for obtaining these exponentially localized singular vectors is mathematically equivalent to computing the (matrix 2-norm) pseudoeigenmodes corresponding to the driving frequency of interest \cite{PostBosoranas,SatoOkumaPseudo}. Further to that, however, pseudospectral theory offers a satisfying two-part physical explanation for directional amplification: (1) Injecting a signal at a pseudoresonant (with the resonance becoming exact in the infinite-size limit) frequency results in a strong 
excitation of the corresponding pseudoeigenmode; and (2) Due to the `mid-gap' (in the sense of NH topology \cite{KawabataNHSymTopo}) nature of these pseudoresonances, the resulting excitation is localized on the boundaries, thus manifesting significant end-to-end response. While the first part broadly characterizes our pseudospectral approach to linear response, the second is specific to the NH topological features of the dynamical matrix in type I DM systems. 

In light of these demonstrations, we conclude that the pseudospectrum rather than the spectrum is the 
relevant tool for characterizing the response of QBHs and QBLs under classical linear driving, regardless of Hermiticity, or underlying topology. Furthermore, such an approach is most warranted in highly non-normal regimes whereby the pseudospectrum differs dramatically from the exact spectrum.

\section{Conclusions and Outlook}
\label{sec:conclusions}

In this work, we have focused on one-dimensional Markovian lattice systems of free bosons that enjoy translational invariance `up to boundary conditions' and, by leveraging tools from pseudospectral theory for Toeplitz operators, we have provided a classification of the different dynamical stability regimes that the corresponding QBL generators can support. 
Crucially, we have shown how the notion of the pseudospectrum is instrumental to precisely describe the way in which the dynamical stability properties of an idealized system in the thermodynamic, infinite-size limit may imprint themselves onto the physical, finite-size realizations. 

We find that disagreement in dynamical stability properties -- which we have termed dynamical metastability -- manifests itself as an anomalous transient dynamical window in the physical system. During this transient regime, \textit{the system effectively appears to be in the `wrong' dynamical stability phase, for a period of time that scales with its size}. For DM systems whose infinite-size limiting dynamics is unstable,  
which we refer to as {\em type I DM}, this implies size-dependent, transiently amplifying finite-size dynamics; for {\em type II DM} systems, the infinite-size evolution is dynamically stable, and size-enhanced transient stabilization emerges in the finite-size instead. We have constructed and discussed physically motivated models, which exemplify the key features of both flavors of dynamical metastability. While the emergence of type I and type II dynamical metastability depends, in general, on the interplay between coherent and dissipative contributions in the underlying QBL, we have demonstrated that dissipation is {\em not} necessary for DM and the resulting size-dependent transient dynamics to arise. Notably, our analysis shows that dynamical instabilities can emerge even by coupling two QBHs that are independently fully dynamically, and that the presence of special kinds of spectral degeneracies -- Krein collisions -- plays a crucial role in determining the onset of type II DM in the Hamiltonian setting. As we have seen, while type I DM is intrinsically tied to extreme non-normality of the relevant generator, type II DM is not; a deeper understanding of the mechanism that is responsible for the onset of type II DM, and the possible significance of Krein stability theory in general QBLs, are questions we leave for future investigation.  

Aside from its key role in enabling the above-mentioned dynamical classification scheme, pseudospectral theory turns out to 
be useful in predicting and interpreting the way in which systems described by a QBL behave in the presence of weak, external perturbations, within the validity of linear response theory. For instance, this has direct implications for the behavior of the response matrix in the context of directional and topological amplification \cite{NunnenkampTopoAmp,NunnenkampResoration,PorrasTopoAmp}. We have shown that a finite-size QBL exhibits the strongest response when driven at resonant frequencies of the corresponding infinite-size system; remarkably, when the dynamical generator is highly non-normal, including for type I DM systems, this results in an  anomalous, {\em resonant response to highly off-resonant drives.} 

Dynamical metastability is also an attractive framework for probing the dynamical behavior of quantum correlations. For QBHs, 
we have shown that it reflects directly into the `pre-asymptotic scaling' of bipartite EE. In particular, we have argued that the onset of an anomalous transient (DM type I) phase explains the generation of abnormally large asymptotic EE that has been recently reported for a family of QBHs in a different context
\cite{ClerkInterpolation}. 
Related to that, it would be worth exploring whether dynamical metastability plays an equally impactful role for entanglement generation in QBLs, or whether distinctive entanglement signatures could be tied to {\em topological} dynamical metastability and the emergence of bosonic edge modes the latter entails \cite{PostBosoranas}. Altogether, we expect that the effects of non-normality and the non-trivial pseudospectrum will also prove important in studying many-body bosonic dynamical systems described by explicitly time-dependent quadratic generators or weakly interacting ones, beyond the quadratic regime. We leave such topics to future work.

\section*{acknowledgements}

It is a pleasure to thank Michiel Burgelman, Mattias Fitzpatrick, Joseph Gibson, Carlo Presilla, and Andrew Projansky for stimulating discussions on various topics relevant to this study, and Abhijeet Alase for useful input on Wiener-Hopf matrix factorization techniques. Work at Dartmouth was partially supported by the US National Science Foundation through Grants No. PHY-2013974 and No. OIA-1921199.

\vfill

\appendix

\onecolumngrid

\section{The semi-infinite boundary condition}

\subsection{The associated Toeplitz operator and its physical significance}
\label{app: defineSIBCs}

In this appendix, we will make our notion of SIBCs mathematically rigorous and connect to the theory of block-Toeplitz operators. All mathematical results can be found in Ref.\,\cite{BottcherToe}.

In the main text, SIBCs were defined by bisecting a bi-infinite, translationally invariant system and obtaining two separate half-infinite chains -- one with a left boundary and one with a right boundary. Mathematically, the BIBC dynamical matrix is a (banded) block-Laurent operator defined on the Hilbert space $\mathcal{H}_{-\infty,\infty} \equiv \ell^2(\mathbb{Z}) \otimes \mathbb{C}^{2d}$,  where $\ell^2(\mathbb{Z})$ is the space of square summable maps from $\mathbb{Z}$ (the bi-infinite lattice) to $\mathbb{C}$ and $d$ is the number of degrees of freedom per lattice site. 
%In fact, when normalization is of no concern, we may generalize our Hilbert space to a general vector space $\mathcal{S}(\mathbb{Z})\otimes \mathbb{C}^{2d}$ with $\mathcal{S}(\mathbb{Z})$ the vector (not Hilbert) space of bi-infinite complex sequences. 
The banded-block Laurent moniker means that $\mathbf{G}^\text{BIBC}$ has the form
\begin{align}
    \mathbf{G}^\text{BIBC} = \mathds{1}\otimes \mathbf{g}_0+\sum_{r=1}^R \left(\mathbf{V}^r \otimes \mathbf{g}_r + (\mathbf{V}^\dag)^r\otimes \mathbf{g}_{-r} \right),
\end{align}
with $\mathbf{g}_{r}$ the $2d\times 2d$ internal matrices, $R$ the (finite) coupling range, $\mathds{1}$ the identity operator on $\ell^2(\mathbb{Z})$, and $\mathbf{V}$ ($\mathbf{V}^\dag$) the left- (right-) shift operator on $\ell^2(\mathbb{Z})$ defined in Eq.\,\eqref{Sshift}. The aforementioned bisection procedure is performed by 
%decoupling 
removing the coupling between sites $0$ and $1$. Mathematically, this amounts to replacing $\mathbf{V}$ with 
\begin{align}
    \mathbf{T} &= \sum_{j=-\infty,\neq 0}^\infty \vec{e}_j\vec{e}_{j+1}^\dag = \mathbf{T}^\text{R} + \mathbf{T}^\text{L} = \mathbf{V}-\vec{e}_{0}\vec{e}_1^\dag,
    \qquad
    \mathbf{T}^\text{R} \equiv \sum_{j=-\infty}^{-1} \vec{e}_{j}\vec{e}_{j+1}^\dag
    \quad
    \mathbf{T}^\text{L} \equiv \sum_{j=1}^{\infty} \vec{e}_{j}\vec{e}_{j+1}^\dag.
\end{align}
In the above, we have defined the left-shift operator $\mathbf{T}^{\text{L(R)}}$ for the left- (right-) boundary half-infinite chain. It follows that
\begin{align}
\label{eq: GSIBC}
    \mathbf{G}^\text{SIBC} &= \mathbf{G}^\text{LSIBC} + \mathbf{G}^\text{RSIBC},
    \\
   \mathbf{G}^\text{L(R)SIBC}&\equiv  \mathds{1}\otimes \mathbf{g}_0+\sum_{r=1}^R \left((\mathbf{T}^{\text{L(R)}})^r \otimes \mathbf{g}_r + (\mathbf{T}^{\text{L(R)}\dag})^r\otimes \mathbf{g}_{-r} \right) .
\end{align}
The operator $\mathbf{G}^\text{LSIBC}$ is supported on the right-half-space $\ell^2(\mathbb{Z}_{>0}) \otimes \mathbb{C}^{2d}$, with $\mathbb{Z}_{>0}=\{1,2,\ldots\}$, while $\mathbf{G}^\text{RSIBC}$ instead is supported on the left-half-space $\ell^2(\mathbb{Z}_{\leq 0})\otimes \mathbb{C}^{2d}$, with $\mathbb{Z}_{\leq 0} = \{\ldots,-2,-1,0\}$. The former can be straightforwardly interpreted as a Toeplitz operator $\mathbf{G}^\text{LSIBC}:\ell^2(\mathbb{Z}_{>0}) \otimes \mathbb{C}^{2d} \to \ell^2(\mathbb{Z}_{>0}) \otimes \mathbb{C}^{2d}$. Less obviously, the latter can also be mapped to a banded block-Toeplitz operator by defining the map
\begin{align}
    W: \ell^2(\mathbb{Z}_{\leq 0})&\longrightarrow \ell^2(\mathbb{Z}_{>0})
    \\
    \vec{e}_j &\longmapsto \vec{e}_{1-j}.
\end{align}
To simplify notation, we let $\mathbf{G}\equiv\mathbf{G}^\text{LSIBC}$ and $\widetilde{\mathbf{G}} \equiv W \mathbf{G}^\text{RSIBC} W^{-1}$. In matrix notation, we then have 
\begin{align}
    \mathbf{G} &= \begin{bmatrix}
\mathbf{g}_0 & \mathbf{g}_{1} & \cdots &  & 
\\
\mathbf{g}_{-1} & \mathbf{g}_0 & \ddots & & 
\\
\vdots & \ddots & \ddots &  & 
\end{bmatrix},\quad \widetilde{\mathbf{G}} = \begin{bmatrix}
\mathbf{g}_0 & \mathbf{g}_{-1} & \cdots &  & 
\\
\mathbf{g}_{1} & \mathbf{g}_0 & \ddots & & 
\\
\vdots & \ddots & \ddots &  & 
\end{bmatrix} .
\end{align}
In the mathematical literature, the operator $\widetilde{\mathbf{G}}$ is known as the \textit{associated block-Toeplitz operator} \cite{BottcherToe}. If we denote the symbol for $\mathbf{G}$ by $\mathbf{g}(z)$ (Eq.\eqref{Sym}) and the symbol for $\widetilde{\mathbf{G}}$ by $\widetilde{\mathbf{g}}(z)$, then $\widetilde{\mathbf{g}}(z) = \mathbf{g}(1/z)$.

From this analysis, we conclude that:
\begin{itemize}
    \item The dynamical matrix of a semi-infinite chain retaining only the \textit{left} boundary is equivalent to a \textit{banded block-Toeplitz operator}, and
    \item The dynamical matrix of a semi-infinite chain retaining only the \textit{right} boundary is equivalent to the \textit{associated banded block-Toeplitz operator}.
\end{itemize}
In particular, for the pseudospectrum it follows that 
\begin{align}
\label{pspeclim}
\sigma_\epsilon(\mathbf{G}^\text{SIBC}) = \sigma_\epsilon(\mathbf{G}^\text{LSIBC})\cup \sigma_\epsilon(\mathbf{G}^\text{RSIBC})= \sigma_\epsilon(\mathbf{G})\cup \sigma_\epsilon(\widetilde{\mathbf{G}}).
\end{align}
Physically, the semi-infinite spectrum consists of the pseudospectrum (and hence, spectrum when $\epsilon\to 0$) {\em of the left and right boundary chains together}. Remarkably, these two spectra can be {\em different} \cite{BottcherToe}.

\subsection{Semi-infinite spectrum and the Wiener-Hopf factorization}

\label{app:WH}

%VF:Taking a slightly different approach
The spectrum of a block-Toeplitz operator, say $\mathbf{X}$, can be explicitly computed from its symbol, say $\mathbf{x}(z)$. If $\mathbf{x}(z) = \sum_{r\in \mathbb{Z}}\mathbf{x}_r z^r$ is any $m\times m$ matrix valued function such that (i) $\sum_{r\in\mathbb{Z}}\norm{\mathbf{x}_r} < \infty$ for any matrix norm and (ii) $\mathbf{x}(e^{ik})$ is invertible for all $k\in\mathbb{R}$, then $\mathbf{x}(z)$ has a \textit{left Wiener-Hopf factorization} of the form
\begin{align}
\label{eq: LeftWHFactorization}
    \mathbf{x}(z) = \mathbf{A}_{+}(z)\mb D(z)\mb A_{-}(z),
\end{align}
with $\mb A_{+}(z)$ invertible for all $|z|\leq 1$,  $\mb A_{-}(z)$ invertible for all $|z|\geq 1$, and $\mb D(z) = \text{diag}\left(z^{\kappa_1},\ldots,z^{\kappa_{m}}\right)$, with $\kappa_j\in\mathbb{Z}$ the \textit{left partial indices}. Similarly, $\mathbf{x}(z)$ has a \textit{right Wiener-Hopf factorization} of the form 
\begin{align}
    \mathbf{x}(z) = \mathbf{B}_{-}(z)\mb E(z)\mb B_{+}(z),
\end{align}
with $\mb B_{-}(z)$ invertible for all $|z|\geq 1$,  $\mb B_{+}(z)$ invertible for all $|z|\leq 1$, and $\mb E(z) = \text{diag}\left(z^{\mu_1},\ldots,z^{\mu_{m}}\right)$, with $\mu_j\in\mathbb{Z}$ the \textit{right partial indices}. Let us, in addition, define the winding number of $\det \mathbf{x}(z)$ around zero as:
\begin{align}\nu(0, \mb x(z))\equiv \frac{1}{2\pi i}\int^{\pi}_{-\pi}dk\frac{d}{dk}\left[ \log\left(\det(\mb x(e^{ik}))\right)\right]. 
\end{align}
With this, we state a number of key remarks and known results.
\begin{enumerate}
    \item Our convention swaps ``left" and ``right" compared to Ref.\,\cite{BottcherToe}. Our left partial indices are the partial indices of Ref.\,\cite{WienerHopf}.
    \item Both $\{\kappa_j\}$ and $\{\mu_j\}$ are unique up to ordering and are generally distinct. However, they have the same sum. In fact, the sum of either set of indices is equal to $\nu(0,\mb x(z))$. 
    \item The right partial indices of $\mathbf{x}(z)$ are equivalent to the left partial indices of the associated symbol $\widetilde{\mathbf{x}}(z) = \mathbf{x}(1/z)$. 
    \item Condition (i) is always satisfied by symbols of \textit{banded} block-Toeplitz matrices. The banded assumption (i.e., that there exist integers $p$ and $q$ such that $q\leq p$ and $\mathbf{x}_r = 0$ for all $r>p$ and $r<q$) ensures that the sum $\sum_{r\in\mathbb{Z}}\mathbf{x}_r z^r$ has only a finite number of nonzero terms. In this case, we say $\mathbf{x}(z)$ is a \textit{matrix Laurent polynomial}. 
    \item Condition (ii) guarantees that $\mathbf{X}$ is a \textit{Fredholm operator}. That is, an operator with a finite-dimensional kernel and co-kernel. Thus, it may be replaced with the assumption that the block-Toeplitz operator is Fredholm.
    \item A block-Toeplitz matrix whose symbol satisfies (i) and (ii) is invertible if and only if $\kappa_j=0$ for all $j\in\{1,\ldots,m\}$. 
\end{enumerate}
This series of facts can be leveraged to completely characterize the spectrum of an arbitrary banded block-Toeplitz operator. 

\begin{proposition} 
Let $\mathbf{X}$ be a banded block-Toeplitz operator. Then $\lambda\in\sigma(\mathbf{X})$ ($\lambda\in\sigma(\widetilde{\mathbf{X}})$) if and only if (i) $\lambda\in\sigma(\mathbf{x}(e^{ik}))$ for some $k\in\mathbb{R}$ or (ii) $\mathbf{x}(z) - \lambda \mathds{1}_{m}$ has at least one non-zero left (right) partial index. 
\end{proposition}

We note that while this result is not explicitly stated in Ref.\,\cite{BottcherToe}, it is a direct consequence of the above listed series of results found within.

\begin{proof}
    We will prove the theorem for $\mathbf{X}$. The result for $\widetilde{\mathbf{X}}$ follows immediately. Fix $\lambda\in\mathbb{C}$ and let $\mathbf{X}' \equiv \mathbf{X}-\lambda \mathds{1}$ and $\mathbf{x}'(z) \equiv \mathbf{x}(z)-\lambda \mathds{1}_m$. Clearly, $\lambda\in \sigma(\mathbf{X})$ if and only if $\mathbf{X}'$ is not invertible.  There are then two cases: (1) $\mathbf{X}'$ is not Fredholm or (2) $\mathbf{X}'$ is Fredholm. For case (1), $\mathbf{X}'$ is not Fredholm if and only if $\mathbf{x}'(z)$ is not invertible  on the unit circle (or equivalently, $\lambda \in \sigma(\mathbf{x}(e^{ik}))$ for some $k\in\mathbb{R}$) by point 5 above. In case (2), $\mathbf{X}'$ is Fredholm if and only if $\mathbf{x}'(z)$ admits a left Wiener-Hopf factorization. Since $\mathbf{X}'$ is not invertible, it must be that it has at least one nonzero left particle index by point 6 above.
\end{proof}

\subsubsection{Specialization of the Wiener-Hopf factorization}

In order to get to the factorized form in Eq.\,\eqref{eq: LeftWHFactorization}, a sequence of steps have to be undertaken \cite{WienerHopf}. Generally, so as to find the full SIBC spectrum, the factorization needs to be applied to all $\lambda\in \mathbb{C}/\sigma(\mb G_{\text{BIBC}})$. However, for a special class of models, that correspond to systems with nearest-neighbor couplings and one internal degrees of freedom, we can derive conditions that allow us to determine the existence of non-trivial partial indices without performing the full factorization procedure. Here, we apply the steps for $\mb g(z)$. $\tilde{\mb g}(z)$ can be treated analogously. 

\smallskip

\noindent 
\begin{theorem} \label{Th: NonTrivPartial}
Let $\mb g'(z)\equiv \mb g(z)-\lambda \mathds{1}_2=\begin{pmatrix}
    a(z) & b(z)\\
    c(z)& d(z)
\end{pmatrix}$ be a $2 \times 2$ matrix-valued symbol, invertible on the unit circle, such that $p=1, q=-1$, and $\det(\mb g'(z))$ has exactly two roots, $z_1$ and $z_2$, inside the unit circle. In order for $\mb g(z)$ to have non-trivial partial indices, namely $\kappa=\{-1,1\}$, at $\lambda$, it is necessary to satisfy
$b(z_2)a(z_1)=a(z_2)b(z_1)$
and $d(z_2)c(z_1)=c(z_2)d(z_1)$. If $z_1\neq z_2$, these conditions are also sufficient. 
\label{thm:wh}
\end{theorem}

\begin{proof}
For a symbol invertible on the unit circle, the WH procedure starts off by computing the roots of its determinant polynomial $\det(\mb g(z)-\lambda \mathds{1})$, and selecting the ones inside the unit circle. If there are exactly two such roots, the accessible partial indices at $\lambda$ are limited to $\kappa^{(0)}_i=\{0,0\}$ or $\kappa^{(0)}_i=\{-1,1\}$, with $\kappa^{(0)}_i$ obtained as follows. We define:
\begin{align*}
    \mb A^{(2)}_{+}\equiv z(\mb g(z)-\lambda \mathds{1}_2), \quad\quad  \mb D^{(2)}\equiv \mathds{1}_2, \quad \quad \mb A^{(2)}_{-}\equiv \mathds{1}_2.
\end{align*}
We then successively remove each zero $z_i$ from $\mb A^{(i)}_{+}(z)$, obtaining $z(\mb g(z)-\lambda\mathds{1}_2)=\mb A^{(1)}_{+}(z)\mb D^{(1)}(z)\mb A^{(1)}_{-}(z)$, and finally $z(\mb g(z)-\lambda \mathds{1}_2)=\mb  A^{(0)}_{+}(z)\mb D^{(0)}(z)\mb A^{(0)}_{-}(z)$, for the $2\times 2$ case. Here, \begin{align*}
\mb A_{+}^{(i-1)}=\mb A^{(i)}_{+}\mb U_i^{-1}\mb \Pi_i, \quad \mb A_{-}^{(i-1)}=\mb \Pi_i\mb V_i\mb A_{-}^{(i)},\quad \mb D^{(i-1)}=\mb \Pi_i\mb D^{(i)}(\mathds{1}_2-\vec{e}_{s}^{\dag}\vec{e}_{s_i}+z\vec{e}_{s}^{\,\dag}\vec{e}_{s_i})\mb \Pi _i,
\end{align*} 
where 
$$\mb V_i=\mathds{1}_2-\frac{z_i}{z}\vec{e}_{s_i}^{\dag}\vec{e}_{s_i}+\alpha_i z^{\kappa_{s_i}^{(i)}-\kappa_{m_i}^{(i)}}\vec{e}_{m_i}^{\dag}\vec{e}_{s_i},\qquad \mb U_i^{-1}=\mathds{1}_2-\vec{e}_{s_i}^{\dag}\vec{e}_{s_i}+(z-z_i)\vec{e}_{s_i}^{\dag}\vec{e}_{s_i}+\alpha_i\vec{e}_{m_i}^{\dag}\vec{e}_{s_i},$$ 
and $\mb \Pi_i$ reshuffles the elements of $\mb D^{(i)}$, which are of the form $\{z^{\kappa_1^{(i)}},z^{\kappa_2^{(i)}}\}$ on the diagonal (and $0$ otherwise), so that $\kappa_1^{(i)}\geq \kappa_2^{(i)}$. 

Thus, the end-goal is to obtain the structure of $\mb D^{(0)}$, and hence $\kappa_i^{(0)}$. Since $\det \mb A^{(2)}_{+}(z_2)=0$, $\mb A^{(2)}_{+}(z_2)$ has a column which is linearly dependent on previous columns. The same will be true for $\mb A^{(1)}(z_1)$. Let $\mb A^{(i)}_{+,n}$ denote the $n$th column of $\mb A^{(i)}_{+}$, and $\mb A^{(i)}_{+,nm}$ the $nm$-th element. Since our matrices are $2\times 2$, either $\mb A^{(i)}_{+,2}(z_i)$ is linearly dependent on $\mb A^{(i)}_{+,1}(z_i)$, i.e., $\mb A^{(i)}_{+,2}(z_i)=\alpha_i\mb A^{(i)}_{+,1}(z_i)$, or $\mb A^{(i)}_{+,1}(z_i)=[0,0]^{T}$. Let $s_i$ denote the linearly dependent column of $\mb A_{+}^{(i)}(z_i)$ and $m_i$ the independent one. Then, we have $\alpha_i=\frac{\mb A^{(i)}_{+,12}(z_i)}{\mb A^{(i)}_{+,11}(z_i)}=\frac{\mb A^{(i)}_{+,22}(z_i)}{\mb A^{(i)}_{+,21}(z_i)}$, if $s_i=2$ and $\alpha_i=0$, if $s_i=1$.

There are four possible cases to consider: 

(i) $\mb A^{(2)}_{+,2}(z_2)$ is linearly dependent on $\mb A^{(2)}_{+,1}(z_2)$,  $\mb A^{(1)}_{+,2}(z_1)$ is linearly dependent on  $\mb A^{(1)}_{+,1}(z_1)$; 

(ii)  $\mb A^{(2)}_{+,1}(z_2)=[0,0]^T$,   $\mb A^{(1)}_{+,2}(z_1)$ is linearly dependent on  $\mb A^{(1)}_{+,1}(z_1)$; 

(iii)  $\mb A^{(2)}_{+,1}(z_2)=[0,0]^T$ and  $\mb A^{(1)}_{+,1}(z_1)=[0,0]^T$; 

(iv) $\mb A^{(2)}_{+,2}(z_2)$ is linearly dependent on $\mb A^{(2)}_{+,1}(z_2)$,  $\mb A^{(1)}_{+,1}(z_1)=[0,0]^T$.  

\noindent 
Cases (i) and (ii) yield trivial factorizations $\kappa^{(0)}_i=0,0$, while (iii) and (iv) give non-trivial partial indices $\kappa^{(0)}_i=1,-1$.  Let $\mb g(z)-\lambda \mathds{1}_2\equiv\begin{pmatrix}
    a&b\\
    c&d
\end{pmatrix}$. These two non-trivial cases translate into the following conditions on the roots of the determinant equation and the matrix elements of the symbol; for case (iii), we have
$a(z_1)=c(z_1)=0$ and 
$a(z_2)=c(z_2)=0,$
while for case (iv), it follows that  
$b(z_2)a(z_1)=a(z_2)b(z_1)$
and $d(z_2)c(z_1)=c(z_2)d(z_1)$, with $z_1\neq z_2$.
\end{proof}

\begin{cor}
 Let $\mb g'(z)\equiv \mb g(z)-\lambda \mathds{1}_2=\begin{pmatrix}
    a(z) & b(z)\\
    c(z)& d(z)
\end{pmatrix}$, with $\mb g(z)$ a symbol of a QBH with $R=1$ and $d=1$. $\mb g(z)$ has non-trivial partial indices, namely $\kappa=\{-1,1\}$, at a point $\lambda$, only if
$b(z_2)a(z_1)=a(z_2)b(z_1)$
and $d(z_2)c(z_1)=c(z_2)d(z_1)$. If $z_1\neq z_2$, these conditions are also sufficient. 
\end{cor}
\begin{proof}
    For a QBH with $R=1$ and $d=1$, the symbol is a $2\times 2$ matrix with $p=-q=1$, whose determinant polynomial is a fourth-order one, with at most 2 roots inside the unit circle, due to the symmetry properties of the Hamiltonian. Moreover, since $\nu (0, \mb g(z))=0$ for QBHs, and the symbol is invertible on the unit circle, there have to be exactly $2$ roots, $z_1$ and $z_2$, inside the unit circle. The rest follows from Theorem \ref{Th: NonTrivPartial}.
\end{proof}

\subsubsection{Application to model systems} 
\label{App: WHApplication}

For the coupled HN chain, the symbol can be computed as follows:
\begin{align*}
  \mb A(z)-\omega \mathds{1}_2=\begin{pmatrix}
      -\kappa-\omega- J_R z-J_L/z&w\\
      -w&-\kappa-\omega- J_L z-J_R/z
  \end{pmatrix},
\end{align*}
whereas for the KOC, we have
\begin{align}
\mb g(z)-\omega \mathds{1}_2=\begin{pmatrix}
\Omega-\omega+\frac{iJ}{2}\left(1/z-z\right)&\frac{i\Delta}{2}(1/z+z)\\
\frac{i\Delta}{2}(1/z+z)&-\Omega-\omega+\frac{iJ}{2}\left(1/z-z\right).
\end{pmatrix}
\end{align}
The two symbols are unitarily equivalent (when $\kappa=0$), thus we will work with the latter one. 

In region II of the KOC, the condition for case (iii) yields $\omega=\Omega \pm J$. These values are real and already accounted for by the bulk spectrum. The condition for case (iv) is satisfied for $\omega=\pm \tfrac{1}{\Delta} \sqrt{(J^2-\Delta^2)(\Delta^2-\Omega^2)}$. While these points are purely imaginary in region II, they make $z_1(\omega)=z_2(\omega)$. %in contradiction with the condition for non-trivial partial indices. 
Thus, for one of these two points, we have to go through the full WH factorization procedure (the result will be the same for the other by symmetry). Doing so yields trivial partial indices. Therefore, region II is stable under SIBCs. An exactly analogous calculation shows that the GHC with imaginary hopping remains stable under SIBCs.

We also apply the conditions for a non-trivial SIBC spectrum to the symbol of the BKC with real hopping given by Eq. \eqref{ClerkInterp}:
\begin{align}
\mb g(z)-\omega\mathds{1}_2=\left(
\begin{array}{cc}
 \frac{g+i J}{2 z}+\frac{1}{2} z (g-i J)-\omega  & \frac{i\Delta}{2}(z+1/z)\\
\frac{i\Delta}{2}(z+1/z)& \frac{-g+i J}{2 z}-\frac{1}{2} z
   (g+i J)-\omega  \\
\end{array}
\right).
\end{align}
While this model is dynamically stable in both reciprocal ($g>\Delta$) and nonreciprocal ($g<\Delta$) regimes under OBCs \cite{ClerkInterpolation}, we find that the SIBC spectrum is the interior of the ellipse defined by the bulk spectrum in the nonreciprocal regime, but is confined to the real axis in the reciprocal regime. Therefore, while the nonreciprocal phase is type I DM, the reciprocal phase is fully dynamically stable and non-DM.

\section{Proof that the limit of finite-size spectra is contained in the semi-infinite limit}
\label{BoundProof}

In this appendix, we establish a more general result about the relationship between finite and infinite-size spectra for block-Toeplitz systems, from which the bound in Eq.\,\eqref{gap constraint} follows as an immediate corollary. In particular, we prove the following theorem (which may be thought of as an $\epsilon\to 0$ counterpart of Eq.\,\eqref{pspeclim}).

\begin{theorem}
\label{speclimthm}
Let $\{\mathbf{X}_N\}$ denote a sequence of block-Toeplitz matrices whose corresponding symbol is continuous on the unit circle. Then
\begin{align}
    \lim_{N\to\infty}\sigma(\mathbf{X}_N) \subseteq \sigma(\mathbf{X})\cup\sigma(\widetilde{\mathbf{X}}),
\end{align}
where $\mathbf{X}$ is the block-Toeplitz operator obtained in the limit $N\to\infty$ and $\widetilde{\mathbf{X}}$ is the associated Toeplitz operator.
\end{theorem}

The stability gap bound in Eq.\,\eqref{gap constraint} then follows immediately from the fact that $\sigma(\mathbf{G}^\text{SIBC}) = \sigma(\mathbf{G})\,\cup \,\sigma (\widetilde{\mathbf{G}})$ (in the notation of Eq.\,\eqref{pspeclim}), along with the observation that the symbol of a banded block-Toeplitz matrix is always continuous on the unit circle. 

To prove Theorem\,\ref{speclimthm}, we require a Lemma regarding the \textit{stability} (not in the dynamical sense) of the sequence $\{\mathbf{X}_N\}$. Recall that a sequence of matrices $\{\mathbf{A}_N\}_{N=1}^\infty$ is said to be \textit{stable} if there exists an $N_0$ such that $\displaystyle\sup_{N\geq N_0}\norm{\mathbf{A}_N^{-1}} <\infty$. Here, the norm is the induced matrix 2-norm. 

\begin{lemma}[Adapted from Theorem 6.9 in Ref.\,\cite{BottcherToe}]\label{BottcherLemma} Let $\{\mathbf{X}_N\}$ denote a sequence of $mN\times mN$ (with $m$ the block size) block-Toeplitz matrices, whose corresponding symbol is continuous on the unit circle. Then the sequence $\{\mathbf{X}_N\}$ is stable if and only if both the corresponding Toeplitz operator $\mathbf{X}$ and associated Toeplitz operator $\widetilde{\mathbf{X}}$ are invertible.
\end{lemma}

\begin{proof}[Proof of Theorem \ref{speclimthm}]
Let $\lambda\in \displaystyle\lim_{N\to\infty}\sigma(\mathbf{X}_N)$ and assume, by contradiction, that $\lambda\not\in\sigma(\mathbf{X})\cup \sigma(\widetilde{\mathbf{X}})$. Then, by definition, the operators $\mathbf{X}' \equiv \mathbf{X}-\lambda\mathds{1}$ and $\widetilde{\mathbf{X}}'\equiv \widetilde{\mathbf{X}}-\lambda \mathds{1}$ are invertible. By Lemma\,\ref{BottcherLemma}, the sequence $\{\mathbf{X}_N'\equiv \mathbf{X}_N - \lambda \mathds{1}_{mN}\}$ is stable. Let $N_0$ be such that $\displaystyle\sup_{N\geq N_0}\norm{(\mathbf{X}_{N}')^{-1}} <\infty$. Now, for each $N\geq N_0$, let $d_N\equiv \displaystyle\min_{\mu\in\sigma(\mathbf{X}_N)}|\mu-\lambda|$. Stability of the sequence $\{\mathbf{X}_N'\}$ implies $d_N\neq 0$. In terms of this minimal distance, we have the  resolvent bound (see Theorem 4.2 in Ref.\,\cite{TrefethenPS}, for instance)
\begin{align*}
   \norm{(\mathbf{X}_N')^{-1}} =\norm{(\mathbf{X}_N-\lambda)^{-1}} \geq \frac{1}{d_N}.
   \end{align*}
Therefore, it follows that 
\begin{align}
\label{dbound}
    \sup_{N\geq N_0} \frac{1}{d_N} \leq \sup_{N\geq N_0}\norm{(\mathbf{X}_{N}')^{-1}} <\infty.
\end{align}

Since, by assumption, $\lambda\in \displaystyle\lim_{N\to\infty}\sigma(\mathbf{X}_N)$, 
by the definition of the uniform limiting set (see Chap. 3.5 of Ref.\,\cite{BottcherToe}), we also know that there exists a sequence $\{\lambda_N\}$ with $\lambda_N\in\sigma(\mathbf{X}_N)$ such that $\displaystyle\lim_{N\to\infty}\lambda_N = \lambda$. It follows that $d_N \leq |\lambda_N - \lambda|$. In conjunction with Eq.\,\eqref{dbound}, we conclude that
\begin{align*}
    \sup_{N\geq N_0}\frac{1}{|\lambda_N-\lambda|} \leq \sup_{N\geq N_0} \frac{1}{d_N}  <\infty.
\end{align*}
However, we arrive at a contradiction since $\lambda_N\to\lambda$ implies that the far left hand-side must diverge. 
\end{proof}

\section{Bloch dynamical matrices for single species QBLs}
\label{Gks}

\subsection{General derivation}

The Bloch dynamical matrices of finite-range QBHs with one internal degree of freedom satisfy $\mb g(k)=\sum e^{ikj}\mb g_j$, with
\begin{align*}
    \mb g_{j}^{\dag}\nonumber&=\boldsymbol \tau_3\mb g_{-j}\boldsymbol \tau_3, \quad 
    \mb g_j^{*}=-\boldsymbol \tau_1\mb g_j\boldsymbol \tau_1.
\end{align*}
These properties can be leveraged to show that
\begin{align}
\label{dag}
   \mb  g^{\dag}(k)&=\sum_je^{-ikj}\boldsymbol \tau_3\mb g_{-j}\boldsymbol \tau_3=\boldsymbol \tau_3\sum_je^{ikj}\mb g_j\boldsymbol \tau_3=\boldsymbol \tau_3\mb g(k)\boldsymbol \tau_3 , \\
   \label{star}
   \mb  g^*(k)&= \sum_j e^{-ikj}(-\boldsymbol \tau_1\mb g_j\boldsymbol \tau_1)=-\boldsymbol \tau_1\mb g(-k)\boldsymbol \tau_1.
\end{align}
Expanding in a basis of Pauli matrices, we have
\begin{align*}
    \mb g(k)&\equiv d_0(k)\boldsymbol \sigma_0+id_1(k)\boldsymbol \sigma_1+id_2(k)\boldsymbol \sigma_2+d_3(k)\boldsymbol \sigma_3\equiv \mathcal{G}_k(d_0,d_1,d_2,d_2),
\end{align*}
where the factors of $i$ are included for later simplification.
From Eq.\,(\ref{dag}) and the properties of Pauli matrices we deduce that 
\begin{align*}
    \mathcal{G}_k(d_0^*,-d_1^*,-d_2^*,d_3^*)= \mathcal{G}_k(d_0,-d_1,-d_2,d_3),
\end{align*}
i.e. all $d_i(k)$ are real. From Eq.\,(\ref{star}), we obtain
\begin{align*}
    \mathcal{G}_k(d_0,-d_1,d_2,d_3)= \mathcal{G}_{-k}(-d_0,-d_1,d_2,d_3), 
\end{align*}
which implies that $d_0$ is an odd function of $k$, while the rest are all even.
All in all, it follows that 
\begin{align}
    \mb g(k)&=\begin{pmatrix}
       d_0(k)+d_3(k)&i d_1(k)+d_2(k)\\
    i d_1(k)-d_2(k)   &d_0(k)-d_3(k)
    \end{pmatrix},
\end{align}
with $d_0$ encoding imaginary hopping terms, $d_3$ real hopping, $d_2, d_1$ real and imaginary pairing, respectively. The corresponding spectrum is given by
\begin{align*}
\sigma(\mb g(k))=\omega_{\pm}(k)&=d_0(k)\pm\sqrt{d_3(k)^2-d_2(k)^2-d_1(k)^2}.
\end{align*}
Hence, the dynamical stability condition is:
\begin{align}
\label{dynam stability}
    d_3(k)^2-d_2(k)^2-d_1(k)^2\geq 0,\quad \forall k.
\end{align}
The thermodynamic stability condition can be obtained from  
$$\sigma\left(\boldsymbol \tau_3\mb g(k)\right)=d_3(k)\pm\sqrt{d_0(k)^2+d_1(k)^2+d_2(k)^2},$$ which amounts to two possibilities:
\begin{align}
\label{thermo stability}
  \nonumber
  \text{(i)}\hspace{0.5em}d_3(k)\geq 0 &\text{ and } d_3^2(k)-d_2^2(k)-d_1^2(k)-d_0^2(k)\geq 0 \quad \forall k\\
  \text{(ii)} \hspace{0.4em}d_3(k)\leq  0 &\text{ and } d_3^2(k)-d_2^2(k)-d_1^2(k)-d_0^2(k)\geq 0 \quad \forall k,
\end{align}
where (i) (respectively, (ii)) corresponds to a system with energies bounded from below (respectively, above). Comparing Eqs.\,(\ref{dynam stability}) and (\ref{thermo stability}), we make two observations: First, bulk stability is independent of $d_0(k)$; second, a bulk stable system may be taken from thermodynamically stable to a thermodynamically unstable phase by tuning the magnitude of $d_0(k)$. We conclude by noting that this analysis can be carried out in an arbitrary spatial dimension $D$, in which case $k$ is replaced with the $D$-dimensional wavevector $\vec{k}$.

\subsection{Specialization to model systems}
\label{app:models} 

\subsubsection{The Kitaev-coupled oscillator chain}

From Eq.\,(\ref{KOC}) in the main text, we can read off:
\begin{align*}
    \mb g_0=\begin{pmatrix}
        \Omega&0\\0&-\Omega 
    \end{pmatrix}, \quad \mb g_1=\frac{i}{2}\begin{pmatrix}
        -J&\Delta\\\Delta& -J
    \end{pmatrix},\quad \mb g_{-1}=\frac{i}{2}\begin{pmatrix}
        J&\Delta\\\Delta& J
    \end{pmatrix}. 
\end{align*}
Thus, the KOC Bloch dynamical matrix is given by
\begin{align*}
  \mb g(k)=\mb g_0+\mb g_1 e^{ik}+\mb g_{-1}e^{-ik} =\begin{pmatrix}
    \Omega+J\sin k& i \Delta \cos k\\
    i \Delta \cos k& -\Omega+J \sin k
\end{pmatrix},
\end{align*}
with eigenvalues
\begin{align*}
    \omega_{\pm}(k)&= J \sin k\pm \sqrt{\Omega^2-\Delta^2\cos^2k}.
\end{align*}
The expression under the square root is negative, and hence the system is bulk-unstable for $\Omega<\Delta$. The value $\Omega=\Delta$ corresponds to an exceptional point, while between $\Delta\leq \Omega\leq J$ the system is gapless. The eigenvalues 
$$\sigma(\boldsymbol{\tau}_3\mb g(k)) =\Omega \pm \sqrt{\Delta^2\cos^2k+J^2\sin^2k}.$$ 
When $J>\Delta$, the expression under the square root is maximized at $\pm k/2$ and the condition for thermodynamic stability becomes $\Omega>J$.

\subsubsection{The gapped harmonic chain with imaginary hopping}

The internal coupling matrices for Eqs.\,(\ref{TRSB}) in the main text are:
\begin{align*}
    \mb g_0&=\begin{pmatrix}
        \Omega&0\\
        0&-\Omega
    \end{pmatrix}, \quad \mb g_1=\frac{1}{2}\begin{pmatrix}
      -J-i\gamma &-J\\
      J&J-i\gamma
    \end{pmatrix},\quad  \mb g_{-1}=\frac{1}{2}\begin{pmatrix}
      -J+i\gamma &-J\\
      J&J+i\gamma
    \end{pmatrix},
\end{align*}
whereby we obtain:
\begin{align*}
    \mb g(k)&=\begin{pmatrix}
        \Omega-J\cos k+\gamma \sin k&-J\cos k\\
        J\cos k&-\Omega+J\cos k+\gamma \sin k
    \end{pmatrix}.
\end{align*}
The eigenvalues of the Bloch dynamical matrix are now $\gamma \sin k\pm \sqrt{\Omega(\Omega-2J\cos k)}$. The eigenvalues 
$$\sigma(\boldsymbol{\tau}_3\mb g(k))=\Omega-J\cos k \pm \sqrt{J^2\cos ^2k+\gamma^2\sin^2k}.$$ 
Accordingly, dynamical stability holds as long as $\Omega>2J>0$, while thermodynamic stability is broken when $\gamma^2> \gamma_c^2 = \tfrac{1}{2} \Omega^2 \big(1-\sqrt{1-(2J/\Omega )^2} \big)$.

\section{Other forms of boundary-dependent stability disagreement}
\label{app: Boundary-dependent stability}

In our classification of finite and infinite-size disagreements, we have focused on systems which display dynamical stability differences between OBCs and SIBCs, and the dynamics these disagreements engender. Yet, in principle, other types of disagreement can appear and give rise to interesting dynamics in QBLs.  

Incidentally, the representative systems we have looked at are in the {\em same} dynamical stability phase for SIBCs, as well as BIBCs. This, however, is not generically true. A notable counterexample can be found in \cite{BarnettEdgeInstab}, where a Su-Schrieffer-Heeger Hamiltonian, with modified on-site terms is considered. The reference finds that, under a certain range of parameters, while all the bulk modes are dynamically stable, there are unstable edge-modes that appear for all system-size. Furthermore, through the use of the matrix WH factorization, we can show that the complex edge modes persist in the infinite-size limit.  Thus, while this model is dynamically stable under PBCs and BIBCs, it is unstable under OBCs and SIBCs. This is yet another instance of a single boundary having a dramatic effect on the dynamical stability of a system.

\section{Pure bosonic Gaussian states}
\label{app: Gaussian states}

\subsection{Basic properties}

In this appendix, we review the basic properties of pure bosonic, Gaussian states, as they prove to be a useful tool for demonstrating the effect of dynamical metastability on entanglement production in Sec.\,\ref{EE}. A Gaussian state is uniquely specified by its first two cumulants, the mean-vector $\vec{m} = \braket{{\cal R}}$ and the CM defined in Eq.\,\eqref{covM}, 
$$\mathbf{\Gamma}_{ij} = \braket{\{{\cal R}_i - \braket{{\cal R}_i},{\cal R}_j - \braket{{\cal R}_j}\}}, \quad {\cal R}=[x_1,p_1,\ldots, x_N,p_N]^T, \quad [{\cal R}_k,{\cal R}_l]=(-\boldsymbol{\tau}_2)_{kl}\mathds{1}\equiv i\mb \Omega_{kl}\mathds{1}.$$ 
As noted in the main text, the symplectic form $\boldsymbol{\Omega}$ encodes the CCRs in the quadrature basis. The bipartite EE of any subsystem in a Gaussian state is completely encoded in the CM \cite{AdessoEE,Rigol2018}. Therefore, we recount the necessary and sufficient requirements for a CM to be that of a pure, bosonic Gaussian state.

In the quadrature basis, any CM corresponding to a physical state is a positive semi-definite, real, symmetric matrix. The Heisenberg uncertainty relation obeyed by the quadratures imposes an additional constraint, which translates into the following condition on the CM \cite{GaussianCov,AdessoEE}:
\begin{align}
\label{state condition}
    \mb \Gamma+i\boldsymbol{\Omega}\geq 0.
\end{align}
By the Williamson theorem, any symmetric, positive $2N\times 2N$ matrix can be decomposed as:
\begin{align}
\label{CM decomp}
    \mb \Gamma= \mb S^T \boldsymbol{\nu} \mb S,
\end{align} with $\mb S$ a \textit{symplectic transformation}, one that preserves CCRs, $\mb S^T \boldsymbol{\Omega} \mb S=\boldsymbol{\Omega}$, and $\boldsymbol{\nu}=\bigoplus\limits_{j=1}^{N}\nu_j \mathbb{1}_2$, with $\nu_j>0$ the \textit{symplectic eigenvalues} of $\mb\Gamma$. The symplectic eigenvalues of $\mb \Gamma$ are the positive eigenvalues of $i\boldsymbol{\Omega}\mb\Gamma$, as $\sigma(i\boldsymbol{\Omega}\mb\Gamma)=\{\pm \nu_j\}$. Finally, imposing Eq.\,\eqref{state condition}, it follows that a valid CM has to satisfy $\nu_j\geq 1$.

The purity condition on a Gaussian state, $\tr(\rho^2)=1$, is equivalent \cite{AdessoPurity,AdessoEE} to $\prod\limits_{i=1}^{N}\nu_i^{-1}=1$. Using Eq.\,\eqref{CM decomp} and the property $\det \mb S=1$, leads to $\det\mb \Gamma=\prod\limits_{j=1}^{N}\nu_j^2$. Therefore, the purity of a Gaussian state requires $\frac{1}{\sqrt{\det{\mb \Gamma}}}=1$. This fixes $\nu_j=1, \forall j$ and saturates the minimum uncertainty in Eq.\,\eqref{state condition}. Notably, it follows that every pure Gaussian state is the vacuum state with respect to its Williamson normal form in Eq.\,\eqref{CM decomp}.

\subsection{Construction of generic pure Gaussian states}

Now we provide a roadmap for constructing generic pure Gaussian states. From the purity condition, the matrix $\boldsymbol{\nu}$ in Eq. \eqref{CM decomp} is fixed. Therefore, we need only to find a general form for $\mb S$. Any symplectic transformation has an Euler (or Bloch-Messiah) decomposition $\mb S=\mb O \bigoplus\limits_{j=1}^{N}\begin{pmatrix}
    e^{r_j}&0\\
    0&e^{-r_j}
\end{pmatrix} \mb O{}'$, with $r_j\in \mathbb{R}$ and $\mb O,\mb O'$ symplectic and orthogonal \cite{QRdecomp,ArvindSymp}. Furthermore, the structure of a generic symplectic orthogonal matrix can be expressed as 
$$\mb O=\begin{pmatrix}
   \mb  X&\mb Y\\
    -\mb Y&\mb X
\end{pmatrix},  \qquad \mb X\mb X^T+\mb Y\mb Y^T=1,\quad \mb X\mb Y^T-\mb Y\mb X^T=0.$$ 
This fact may be seen by computing the (unique) polar decomposition of $\mathbf{S}$ and then diagonalizing the positive definite part. Notably, $\mb X, \mb Y$ that satisfy the latter conditions give $\mb X+i\mb Y=\mb Z$, with $\mb Z$ unitary. It follows that any $2N\times 2N$ symplectic matrix may be computed from two $N\times N$ unitary matrices (defining $\mathbf{O}$ and $\mathbf{O'}$) and a set of $N$ real numbers $r_j$. In case of generating a random symplectic transformation, one may sample these two unitaries from the uniform distribution associated to the Haar measure. The real numbers $r_j$ may, for instance, be sampled from a normal distribution with arbitrarily specified mean and variance. This arbitrariness reflects the fact that there is no uniform distribution for the space of symplectic matrices due to its intrinsic noncompactness.

\section{Linear response for quadratic bosonic Lindbladians}
\label{app:LR}

In this appendix, we recount the basic linear response theory for Lindbladians of Ref.\,\cite{ZanardiResponse} and apply it to the special case of a linearly driven QBL. The general setup consists of the unperturbed Lindbladian $\mathcal{L}_0$ subject to a time-dependent perturbation of the form $\xi(t) \mathcal{L}_1$, with $\xi(t)$ some time-dependent scalar function and $\mathcal{L}_1$ a valid Lindbladian. Given an initial condition $\rho_\text{in}$ at $t=0$, the full time dependence is given by:
\begin{align*}
\rho(t) = \mathcal{E}^t(\rho_\text{in}),\quad \mathcal{E}^t = \mathcal{T}\exp\left[\int_{0}^t d\tau\left(\mathcal{L}_0+\xi(\tau)\mathcal{L}_1\right)\right],
\end{align*}
with $\mathcal{T}$ the time-ordering symbol. We also denote by $\mathcal{E}_0^t = e^{t\mathcal{L}_0} = \mathcal{E}^t|_{\xi=0}$ the unperturbed propagator and $\rho_0(t) = \mathcal{E}_0^t(\rho_\text{in})$ the state at time $t$ for the unperturbed system. Given an observable $A$, linear response theory is concerned with evaluating the quantity:
\begin{align*}
\delta\braket{A}(t) \equiv \braket{A}(t) - \braket{A}_0(t) = \tr[A\rho(t)] - \tr[A\rho_0(t)].
\end{align*}
These quantities are most easily expressed in the interaction picture dynamics, whose propagator is given by $\mathcal{E}_I^t \equiv \mathcal{E}_0^{-t}\mathcal{E}^t$. Note that, while formally $\mathcal{E}^{-t}_0$ is not defined as $\mathcal{E}_0^t$ forms only a semigroup, we may consistently define it as $\mathcal{E}_0^{-t} = e^{-t\mathcal{L}_0}$.  The interaction-picture propagator satisfies the differential equation:
\begin{align*}
\frac{d}{dt}\mathcal{E}_I^t = \xi(t) \mathcal{L}_I^t \mathcal{E}_I^t,\quad \mathcal{L}_I^t \equiv \mathcal{E}_0^{-t}\mathcal{L}_1 \mathcal{E}_0^t,
\end{align*}
where we have defined the interaction picture generator $\mathcal{L}_I^t$. It follows that
\begin{align*}
\delta\braket{A}(t) = \tr[A(\mathcal{E}^t-\mathcal{E}^t_0)\rho_\text{in}]=\tr[A_0(t)(\mathcal{E}_I^t- 1)\rho_\text{in}], \quad A_0(t)\equiv (\mathcal{E}^t_0)^\star A, 
\end{align*}
where $1$ is the identity superoperator and 
 we have invoked the unperturbed Heisenberg dynamics of the operator $A$, i.e., $A_0(t)$. One may verify that
\begin{align*}
\mathcal{E}_I^t-1 = \int_0^t d\tau\, \xi(\tau)\mathcal{L}_I^\tau \mathcal{E}_I^\tau,
\end{align*}
which is expanded in a Dyson series. Keeping only the term proportional to $t$, one obtains
\begin{align*}
\delta\braket{A}(t) = \int_0^\infty d\tau \xi(\tau) \chi_A(t,\tau),\quad \chi_A(t,\tau) \equiv \Theta(t-\tau)\tr\left[ A_0(t) \mathcal{L}_I^\tau\rho_\text{in}\right],
\end{align*}
where we have uncovered the linear susceptibility $\chi(t,\tau)$. For later convenience, we can rearrange this to obtain
\begin{align*}
 \chi_{A}(t,\tau) = \Theta(t-\tau)\tr\left[ (\mathcal{E}_0^\tau)^\star(\mathcal{L}_1^\star(A_0(t-\tau))) \rho_\text{in}\right],
\end{align*}
which is completely in terms of the Heisenberg picture. In the case where we have several perturbations, $\sum_{j}\xi_j(t)\mathcal{L}_j$, linearity allows us to individually compute the response of $A$ to each individual perturbation. It follows that:
\begin{align*}
\delta\braket{A}(t) = \int_0^\infty d\tau\,\sum_{j=1}^n \xi_j(\tau) \chi_{A,j}(t,\tau),\quad  \chi_{A,j}(t,\tau) \equiv \Theta(t-\tau)\tr\left[ A_0(t) \mathcal{L}_{j,I}^\tau\rho_\text{in}\right] ,
\end{align*}
where $\mathcal{L}_{j,I}^\tau = \mathcal{E}_0^{-\tau}\mathcal{L}_j\mathcal{E}_0^\tau$ and 
\begin{align}
\label{multiresp}
 \chi_{A,j}(t,\tau) = \Theta(t-\tau)\tr\left[(\mathcal{E}_0^\tau)^\star(\mathcal{L}_j^\star(A_0(t-\tau))) \rho_\text{in}\right].
 \end{align}

Let us now specialize the above formalism to QBLs subjected to linear driving. Specifically, consider a QBL $\mathcal{L}_0$ and consider linear Hamiltonian perturbations, i.e., $H_0\mapsto H_0+H_d(t)$ with:
\begin{align*}
H_d(t) = i\sum_j z_j^*(t) a_j - z_j(t) a_j^\dag = i\vec{\beta}(t)^\dag \bm{\tau}_3\Phi,\quad \vec{\beta}(t) \equiv [z_1(t),z_1^*(t),\ldots,z_N(t),z_N^*(t)]^T,
\end{align*}
where we have utilized the usual $\bm{\tau}$-matrices. In the language of the general theory, we have a number of perturbations \footnote{Technically, the superoperators $-i[\Phi_j,\cdot]$ are not valid Lindbladians since $\Phi_j\neq\Phi_j^\dag$. In this sense, it would be more appropriate to rotate to a quadrature basis whereby the perturbations take the form $\xi'(t)\mathcal{L}_j$, with $\xi'(t)$ real and $\mathcal{L}_j'$ a commutator with $x_j$ or $p_j$. Ultimately, the results will be the same. }:
\begin{align*}
\mathcal{L}_0 &\mapsto \mathcal{L}_0 + \sum_{j=1}^{2N}\xi_j(t)\mathcal{L}_j,\quad \xi_j(t) = i(\vec{\beta}^\dag(t)\bm{\tau}_3)_j,\quad \mathcal{L}_j = -i[\Phi_j,\cdot]. 
\end{align*} 
To evaluate the response of $\braket{\Phi_i}$, we compute the response functions $\chi_{ij}(t,\tau) \equiv \chi_{\Phi_i,j}(t,\tau)$ using the formula in Eq.\,\eqref{multiresp} with $A = \Phi_i$. First off, we note that:
\begin{align*}
\Phi_i(t-\tau) = (\mathbf{V}(t-\tau)\Phi)_i,\quad \mathbf{V}(t) = e^{-i\mathbf{G}_0t},
\end{align*}
in terms of the dynamical matrix $\mathbf{G}_0$ of the unperturbed system. Proceeding, bosonic algebra provides great simplification:
\begin{align*}
(\mathcal{E}_0^t)^\star(\mathcal{L}_j^\star(\Phi_i(t-\tau))) 
&= (\mathcal{E}_0^t)^\star\left( i[\Phi_j^\dag,\Phi_i(t-\tau)]\right)
= - i(\mathbf{V}(t-\tau)\bm{\tau}_3)_{ij} 1_\text{F},
\end{align*}
with $1_\text{F}$ the Fock space identity. We thus conclude:
\begin{align*}
\chi_{ij}(t,\tau) =  - i\Theta(t-\tau)(\mathbf{V}(t-\tau)\bm{\tau}_3)_{ij},
\end{align*}
which is both \textit{time-translation invariant} and \textit{state-independent}.  In fact, the former can be seen as a consequence of the latter since, in particular, one may choose a stationary state as the initial condition. Altogether,
\begin{align*}
\delta \braket{\Phi_i}(t) = \int_0^\infty d\tau\,\sum_{j=1}^n \chi_{ij}(t,\tau) \xi_j(\tau) = \int_0^\infty d\tau\, \Theta(t-\tau)\left(\mathbf{V}(t-\tau)\vec{\beta}^*(\tau)\right)_i
\end{align*}
Defining the displacement vector $\delta\braket{\Phi}(t)$ with $i$-th element $\delta\braket{\Phi_i}(t)$, we may more succinctly write:
\begin{align*}
\delta \braket{\Phi}(t) =  \int_0^\infty d\tau\, \Theta(t-\tau)\mathbf{V}(t-\tau)\vec{\beta}^*(\tau).
\end{align*}
By decomposing $\delta \braket{\Phi}(t)$ into its frequency components via a Fourier transform, we find:
\begin{align*}
    \delta \braket{\widetilde{\Phi}}(\omega) = \int_{-\infty}^\infty\, dt e^{i\omega t} \delta \braket{\Phi}(t) &= \int_{-\infty}^\infty dt e^{i\omega t}\int_0^\infty d\tau\, \Theta(t-\tau)\mathbf{V}(t-\tau)\vec{\beta}^*(\tau) 
    \\
    &= \int_{-\infty}^\infty dt' e^{i\omega t'}\Theta(t')\mathbf{V}(t')\int_0^\infty d\tau\, e^{i\omega \tau}\vec{\beta}^*(\tau) 
\end{align*}
Adopting the convention that the perturbation is turned on at $t=0$, the integration range of the rightmost integral can be extended to $\tau\in(-\infty,\infty)$ and recognized as the Fourier transform of $\vec{\beta}^*(\tau)$, which we denote by $\vec{b}(\omega)$. We further write:
\begin{align*}
    \bm{\chi}(\omega) = \int_{-\infty}^\infty dt' e^{i\omega t'}\Theta(t')\mathbf{V}(t') = \int_0^\infty dt \,e^{-i(\mathbf{G}_0-\omega)t},
\end{align*}
so that $\braket{\widetilde{\Phi}}(\omega)=\bm{\chi}(\omega)\vec{b}(\omega)$. Locality in frequency space is characteristic of linear response theory more generally. Importantly, the \textit{frequency space response function} $  \bm{\chi}(\omega)$ can be computed exactly
\begin{align*}
    \bm{\chi}(\omega) = \int_0^\infty dt \,e^{-i(\mathbf{G}_0-\omega)t} = i(\omega-\mathbf{G}_0)^{-1},
\end{align*}
as long as the largest real parts of the eigenvalues of $(\mathbf{G}_0-\omega)$ are bounded away from the imaginary axis on the left half plane. In particular, the frequency response function is well-defined for any QBL with a strictly negative stability gap. If the stability gap is zero (as can happen in the case of a dynamically stable QBH, for instance), we may introduce a regularization parameter $\eta>0$ and compute:  
\begin{align*}
    \bm{\chi}(\omega) = \lim_{\eta\to 0}\int_0^\infty dt \, e^{-\eta t}e^{-i(\mathbf{G}_0-\omega)t} = \lim_{\eta\to 0}i(\omega-i\eta-\mathbf{G}_0)^{-1} = i(\omega-\mathbf{G}_0)^{-1}.
\end{align*}

Thus, as mentioned in the main text, the frequency-space response is proportional to the resolvent, $\mathbf{R}(z) \equiv (\mathbf{G}_0-z)^{-1}$, of the dynamical matrix evaluated at $z=\omega\in\mathbb{R}$.

\twocolumngrid

\end{document}